\documentclass[a4paper,UKenglish,cleveref, autoref, thm-restate]{lipics-v2021}
\usepackage{commath}
\usepackage{mathtools}
\usepackage{csquotes}
\usepackage{bm}
\usepackage{tikz}
\usetikzlibrary{shapes,patterns,arrows,graphs,decorations.pathreplacing,calc}
\usepackage{xcolor}
\usepackage{graphicx,color}
\usepackage{paralist}
\usepackage[inline]{enumitem}
\usepackage{xspace}
\usepackage{array}
\usepackage{listings}
\usepackage{cite}

\title{Fully Dynamic Algorithms for Knapsack Problems with Polylogarithmic Update Time}

\titlerunning{Fully Dynamic Algorithms for Knapsack Problems with Polylogarithmic Update Time} 

\author{Franziska Eberle}{Faculty of Mathematics and Computer Science, University of Bremen, Germany}{feberle@uni-bremen.de}{}{}
\author{Nicole Megow}{Faculty of Mathematics and Computer Science, University of Bremen, Germany}{nmegow@uni-bremen.de}{}{}
\author{Lukas Nölke}{Faculty of Mathematics and Computer Science, University of Bremen, Germany}{noelke@uni-bremen.de}{}{}
\author{Bertrand Simon}{IN2P3 Computing Center, CNRS, Villeurbanne, France}{bertrand.simon@cc.in2p3.fr}{}{}
\author{Andreas Wiese}{Universidad de Chile, Chile}{awiese@dii.uchile.cl}{}{}

\authorrunning{F. Eberle, N. Megow, L. Nölke, B. Simon, and A. Wiese} %

\Copyright{Franziska Eberle, Nicole Megow, Lukas Nölke, Bertrand Simon, and Andreas Wiese} %



\acknowledgements{We thank Martin Böhm, Peter Kling, and Jens Schlöter for the fruitful discussions.
This work
is partly funded by the Deutsche Forschungsgemeinschaft
(DFG, German Research Foundation) — Project Number
146371743 — TRR 89 Invasive Computing.} 

\relatedversion{Accepted for Publication at FSTTCS 2021.} 

\nolinenumbers 

\hideLIPIcs  

\EventEditors{John Q. Open and Joan R. Access}
\EventNoEds{2}
\EventLongTitle{42nd Conference on Very Important Topics (CVIT 2016)}
\EventShortTitle{CVIT 2016}
\EventAcronym{CVIT}
\EventYear{2016}
\EventDate{December 24--27, 2016}
\EventLocation{Little Whinging, United Kingdom}
\EventLogo{}
\SeriesVolume{42}
\ArticleNo{23}

\newcommand{\eps}{\varepsilon}
\newcommand{\opt}{\textsc{Opt}\xspace}
\newcommand{\opteps}{\ensuremath{\textsc{Opt}_{\smash{\frac{1}{\eps}}}}\xspace}

\newcommand{\optmeps}{\ensuremath{\textsc{Opt}_{\smash{\frac m {\eps^2}}}\xspace}}

\newcommand{\candeps}{\ensuremath{H_{\smash{\frac{1}{\eps}}}}\xspace}

\newcommand{\candmeps}{\ensuremath{H_{\smash{\frac m {\eps^2}}}}\xspace}

\newcommand{\vmax}{\ensuremath{\overline{v}\xspace}\xspace}
\newcommand{\vvmax}{\ensuremath{v_{\max}}\xspace}

\newcommand{\lmax}{\ensuremath{\ell_{\max} }}
\newcommand{\lmin}{\ensuremath{\ell_{\min} }}
\newcommand{\llmin}{\ensuremath{{\bar\ell} }}
\newcommand{\N}{\mathbb{N}\xspace}
\newcommand{\epsfrac}{{1/\eps}}
\newcommand{\NP}{\textup{NP}\xspace}
\newcommand{\classP}{\textup{P}\xspace}
\newcommand{\lp}{\text{LP}}

\newcommand{\capa}{S\xspace}
\newcommand{\OO}{\ensuremath{\mathcal{O}}}
\newcommand{\knapsack}{\textsc{Knapsack}\xspace}
\newcommand{\multiknapsack}{\textsc{Multiple Knapsack}\xspace}
\newcommand{\dmk}{\textsc{Dynamic Multiple Knapsack}\xspace}

\newcommand{\red}{ordinary\xspace}
\newcommand{\yellow}{extra\xspace}
\newcommand{\green}{special\xspace}

\newcommand{\re}{\ensuremath{O}\xspace\xspace}
\newcommand{\ye}{\ensuremath{E}\xspace\xspace}
\newcommand{\gr}{\ensuremath{S}\xspace\xspace}
\newcommand{\wwlog}{without loss of generality\xspace}
\newcommand{\configs}{\ensuremath{\mathcal{C}}}
\newcommand{\types}{\ensuremath{\mathcal{T}}}
\newcommand{\groups}{\ensuremath{\mathcal{G}}}
\newcommand{\smalls}{\ensuremath{\mathcal{F}}}
\newcommand{\vals}{\ensuremath{\mathcal{V}_L}\xspace}
\newcommand{\mcp}{\ensuremath{\mathcal P}\xspace}

\newcommand{\lowval}{low-value\xspace}

\newcommand{\vc}[1]{\ensuremath{\smash{(1+\eps)^{#1}}}}

\newcounter{knaps}
\newcounter{knaps2}

\definecolor{ubred}{RGB}{182,18,49}
\definecolor{utgreen}{RGB}{52,178,51}
\definecolor{utmagenta}{RGB}{207,0,114}
\definecolor{utyellow}{RGB}{254,209,0}
\definecolor{utblueORIGINAL}{RGB}{99,177,229}
\definecolor{utblue}{RGB}{0,156,235}  %
\definecolor{utforest}{RGB}{0,106,77}
\definecolor{utmoos}{RGB}{0,106,77}
\definecolor{utpurple}{RGB}{79,45,127}
\definecolor{utnavy}{RGB}{10, 81, 163} %
\definecolor{utgrey}{RGB}{97,82,88}
\definecolor{utorange}{RGB}{236,122,8}

\usepackage[textsize=tiny,textwidth=2cm,color=green!50!gray]{todonotes} %

\renewcommand{\log}{\ensuremath{\mathrm{log}\;}}
\renewcommand{\epsilon}{\eps}

\begin{document}

\maketitle

\begin{abstract}
Knapsack problems are among the most fundamental problems in optimization. In the \textsc{Multiple Knapsack} problem, we are given multiple knapsacks with different capacities and items with values and sizes. The task is to find a subset of items of maximum total value that can be packed into the knapsacks without exceeding the capacities. We investigate this problem and special cases thereof in the context of \emph{dynamic algorithms} and design data structures that efficiently maintain near-optimal knapsack solutions for dynamically changing input.
More precisely, we handle the arrival and departure of individual items or knapsacks during the execution of the algorithm with worst-case update time polylogarithmic in the number of items. As the optimal and any approximate solution may change drastically, we maintain implicit solutions and support  {polylogarithmic time query operations that can return the computed solution value and the packing of any given item.}

While dynamic algorithms are well-studied in the context of graph problems, there is hardly any work on packing problems (and generally much less on non-graph problems).  {Motivated by the theoretical interest in knapsack problems and their practical relevance, our work bridges this gap.} 
\end{abstract}

\section{Introduction}
\label{sec:intro}

Knapsack problems are among the most fundamental optimization  {problems.} In their most basic form,  {we are given} a knapsack capacity $\capa\in \N$ and a set of $n$ items, where each item $j\in [n] := \{1,2,\ldots,n\}$ has a size~$s_j\in \N$ and a value~$v_j\in \N$. The \knapsack problem asks for a subset of items, $P\subseteq [n]$, with maximal total value $\smash{v(P):=\sum_{j\in P} v_j}$ and with a total size $\smash{s(P):=\sum_{j\in P}s_j}$ that does not exceed the knapsack capacity~$\capa$. In the more general \multiknapsack problem, we are given $m$ knapsacks with capacities~$\capa_i$ for~$ i \in [m]$. Here, the task is to select~$m$ disjoint subsets $P_1,P_2,\ldots,P_m \subseteq [n]$ such that subset $P_i$ satisfies the capacity constraint $s(P_i)\leq \capa_i$ and the total value of all subsets $\smash{\sum_{i\in [m]} v(P_i)}$ is maximized.

\multiknapsack is strongly \NP-hard, even for identical knapsack capacities, as it is a special case of bin packing. \knapsack, on the other hand, is only weakly \NP-hard and admits pseudo-polynomial time algorithms, the first one being  {already} published in the 1950s~\cite{Bellman1957}.

As a consequence of these hardness results, each of the knapsack variants has been studied extensively through the lens of approximation algorithms. Of particular interest are \textit{approximation schemes}, families of polynomial-time algorithms that compute for each~$\eps >0$ a $(1-\eps)$-approximate solution, i.e., a feasible solution with value within a factor of~$(1-\eps)$ of the optimal solution value. Based on the dependency on~$\eps$ of the respective running time, we distinguish \emph{Polynomial Time Approximation Schemes (PTAS)} with arbitrary dependency on~$\eps$, \emph{Efficient PTAS (EPTAS)} where arbitrary functions~$f(\eps)$ may only appear as a multiplicative factor, and \emph{Fully Polynomial Time Approximation Schemes (FPTAS)} with polynomial dependency on~$\frac1\eps$.

The first approximation scheme for \knapsack was an FPTAS by Ibarra and Kim~\cite{IbarraK75} and initiated a long sequence of follow-up work,  {which is}  {still active} \cite{Chan18,Jin19}. \multiknapsack is substantially harder and does not admit {an FPTAS}, unless $\classP = \NP$, even with two identical knapsacks~\cite{ChekuriK05}. However, approximation schemes with {running times of the form $n^{f(\eps)}$ (PTASs) are known \cite{Kellerer99,ChekuriK05} as well as improvements to only $f(\eps) n^{\OO(1)}$ (EPTASs) \cite{Jansen10mkp,Jansen12mkp}}. All these algorithms are {\em static} in the sense that the full instance is given to an algorithm and is then solved.

Given the ubiquitous dynamics of real-world instances, it is natural to ask for {\em dynamic algorithms} that adapt to small changes in the packing instance while spending only little computation time. More precisely, during the execution of the algorithm, items and knapsacks arrive and depart and the algorithm needs to maintain an approximate knapsack solution with an {\em update time} polylogarithmic in the number of items in each step. A dynamic algorithm is then a data structure that implements these updates efficiently and supports relevant query operations.

A practical application is the dynamic estimation of the profit for scheduling jobs in computing clusters  {in which virtual machines can be moved among physical machines~\cite{BeloglazovB10}}.
This allows the service provider to adapt the provided capacity, i.e., the currently running servers, to the current demand, see, e.g., \cite{BobroffKB07,LiTC14,DaudjeeKL14}. An efficient framework for \multiknapsack can be viewed as a first-stage decision tool: In real-time, it determines whether the customer in question should be allowed into the system based on the cost of possibly powering and using additional servers. As the service provider has to decide immediately which request she wants to accept, she needs to obtain the information \emph{fast}, i.e., sublinear in the number of requests already in the system.

Generally, dynamic algorithms constitute a vibrant research field in the context of graph problems. {We} refer to surveys~\cite{DemetrescuEGI2010-survey,Henzinger18-sota,BoriaP11survey} for an overview on dynamic graph algorithms. Interestingly, only for a small number of graph problems there are dynamic algorithms known with {\em polylogarithmic} update time, among them connectivity problems \cite{HenzingerK99,HolmLT01}, the minimum spanning tree \cite{HolmLT01}, and vertex cover~\cite{BhattacharyaHN17,BhattacharyaK19}. Recently, this was complemented by conditional lower bounds that are typically \emph{linear} in the number of nodes or edges; see, e.g., \cite{AbboudW14}. Over the last few years, the generalization of dynamic vertex cover to dynamic set cover gained interest leading to near-optimal approximation algorithms with polylogarithmic update times \cite{BhattacharyaHN19,BhattacharyaHI15,GuptaK0P17,AbboudA0PS19}.  {Also, recently, algorithms have been developed for maintaining maximal independent sets, e.g., \cite{behnezhad2019fully, chechik2019fully, monemizadeh2019dynamic}, and approximate maximum independent sets in special graph classes \cite{henzinger_et_al:LIPIcs:2020:12209, compton2020new, bhore2020dynamic}}.

For packing problems, there are hardly any dynamic algorithms with small update time known. 
A notable exception is a result for bin packing that maintains a $\frac54$-approximative solution with $\OO(\log n)$ update time~\cite{Ivkovic98dynamicbinpacking}. This lack of efficient dynamic algorithms is in stark contrast to the aforementioned intensive research on computationally efficient algorithms for  {packing} problems. Our work bridges this gap initiating the design of data structures and algorithms that efficiently maintain near-optimal solutions. %

\subparagraph*{Our Contribution}
In this paper, we present dynamic algorithms for maintaining approximate solutions for three problems of increasing complexity: \knapsack, \multiknapsack with identical knapsack sizes, and  {general} \multiknapsack. Our algorithms are {\em fully dynamic} which means that in an update operation they can handle  {the arrival or departure of an item and of a knapsack. } Further, we consider the \emph{implicit solution} or \emph{query} model, in which an algorithm is not required to store the solution explicitly in memory such that the solution can be read in linear time at any given point of the execution. Instead, the algorithm may maintain the solution implicitly with the guarantee that a query about the packing can be {answered} in polylogarithmic time.  

We give {\em worst-case} guarantees for update and query times that are polylogarithmic in~$n$, the number of items currently in the input, and bounded by a function of $\eps>0$, the desired approximation accuracy. For some special cases, we can even ensure a polynomial dependency on~$\frac 1 \eps$. In others, we justify the exponential dependency with  {corresponding lower bounds.} Denote by \vvmax the currently largest item value and by \vmax an upper bound on \vvmax that is known in advance. 

\begin{enumerate}%
	\item For \multiknapsack, we design a dynamic algorithm maintaining a $(1-\eps)$-approximate solution with update time
	$2^{f(\epsfrac)} \big(\frac1\eps \log n \log \vmax \big)^{\OO(\epsfrac)} (\log S_{\max})^{\OO(1)}$, where~$f$ is quasi-linear, 
	and query time  {$\smash{\big(\frac1\eps \log n \big)^{\OO(1)}}$}.

	\item The exponential dependency on $\frac1\eps$ in the update time for \multiknapsack is indeed necessary, even for two identical knapsacks. We show that there is no $(1-\eps)$-approximate dynamic algorithm with update time~$\smash{\big(\frac1\eps \log n\big)^{\OO(1)}}$,
	unless~$\classP = \NP$. %
	\item For \knapsack, we give a dynamic $(1-\eps)$-approximation algorithm with update time $\big(\frac1\eps\smash{\log(n \vvmax)}\big)^{\OO(1)}+\OO\big(\frac{1}{\eps}\log n \log \vmax\big)$ and constant query times. %
	
	\item For \multiknapsack with {\em identical knapsacks}  {with capacity $S$ each}, we  {improve the update time to %
		~$\smash{\big(\frac1\eps \log n \; \log \vvmax \; \log S\big)^{\OO(1)}}$ if~$m \geq \frac{16}{\eps^7}\log^2 n$ with query time~$\smash{\big(\frac1\eps \; \log n\big)^{\OO(1)}}$}. %
\end{enumerate}

In each update step, we compute only implicit solutions and provide query {operations} for the solution value, the knapsack of a queried item, {and} the complete solution.
These queries are consistent between two update steps and run efficiently, i.e.,  {run in time} polynomial in~$\log n$ and~$\log \vmax$ and {linear in} the output size.
We remark that it is not possible to maintain a solution with a non-trivial approximation guarantee explicitly with only polylogarithmic update time (even amortized) since it might be necessary to change $\Omega(n)$ items per iteration, e.g., if a very large and very profitable item is inserted and removed in each iteration.

We remark that our result yields a static algorithm with a near-linear running time in $n$.

\subparagraph*{Our Techniques}

Maybe surprisingly, we recompute a $(1-\epsilon)$-approximate solution from scratch in polylogarithmic time after each update. More precisely, we compute a $(1-\epsilon)$-estimate of the value of $\opt$ and additionally store all information that is needed in order to answer any query in polylogarithmic time. Interestingly, this shows that for {such computations}, we do not need exact knowledge about the whole input, but only a small amount of information of polylogarithmic size. We show that this information can be extracted efficiently from suitable data structures in which we store the input items and knapsacks. Even more, we show that we can maintain these data structures in polylogarithmic time per update.

On a high level, we reduce the overall problem to two subproblems {solved} independently. In the first  {one}, we are given only few knapsacks, $m=\big(\frac1\eps \log n \big)^{\OO(1)}$ many, which are the largest knapsacks in the original input. {Here, we} observe that if we select the $\frac m \eps$ most valuable items in the optimal solution correctly, we can afford to fill the remaining space in the knapsacks greedily, i.e., {highest density (value divided by size) first}, and charge the resulting loss to the valuable items. We cannot guess these most valuable items explicitly, but we show that we can select a small {set of candidates} for these items and guess a few placeholder items  {for the remaining ones}. {This} yields an instance with only \smash{$\big(\frac1\eps \log n\big)^{\OO(1)}$} items on which we run a known EPTAS for \multiknapsack \cite{Jansen12mkp} yielding a running time of \smash{$\big(\frac1\eps \log n\big)^{\OO(1)}$}. For the special case of a single knapsack, we show that we can invoke an FPTAS instead, which improves the running time.

In the second subproblem, we are given a potentially large set of
knapsacks, and we are allowed to use an additional set of \smash{$\big(\frac1\eps \log n\big)^{\Theta(1)}$}
knapsacks that the optimal solution does not use (resource augmentation).
We introduce a technique that we call \emph{ {oblivious} linear grouping}.
Linear grouping is a standard technique used in order to round a set
of one-dimensional items that need to be packed into a given set of
containers (e.g., in bin packing), such that they have at most $\frac 1\eps $
different sizes after the rounding (at the expense of leaving an $\epsilon$-fraction
of the items out). However, in our setting we do not know a priori
which input items need to be packed, and therefore we cannot apply
this technique directly. Instead, we show that we can round the input
items to \smash{$(\frac 1\eps \log n )^{\OO(1)}$} different sizes such that we
lose at most a factor of $(1-\eps)$ \emph{independently} of what
the optimal solution looks like. In fact, our rounding method is even
oblivious to the input knapsacks. Therefore, we believe that it might
be useful also for other dynamic packing problems or for speeding
up static algorithms. After rounding the items to \smash{$\big(\frac 1\eps \log n \big)^{\OO(1)}$}
different sizes, we set up a configuration-LP that has a configuration
for each possible set of relatively large items that together fit
inside a knapsack. Thanks to our rounding, there are only polylogarithmically
many configurations and we can solve this LP in time \smash{$\big( \frac 1\eps \log n \big)^{\OO(\epsfrac)}$}.
We use the additional knapsacks in order to compensate errors when
rounding the LP, i.e., due to rounding up the fractional variables
and adding small items greedily into the remaining space of the knapsacks.
Special care is necessary since the sizes of the knapsacks can differ
and hence some item might be relatively large in some knapsack, but
relatively small in another knapsack.

\subparagraph*{Further Related Work}
Since the first approximation scheme for \knapsack \cite{IbarraK75} running times have been improved steadily \cite{GensL79,Lawler1979,GensL1980,KellererP04,rhee2015,Chan18,Jin19} with $\OO(n \log\frac1\eps + (\frac1\eps)^{9/4})$ by Jin~\cite{Jin19} being the currently fastest. Recent work on conditional lower bounds \cite{CyganMWW19,KunnemannPS17} implies that \knapsack does not admit an FPTAS with running time  {of} $\OO((n+\frac{1}{\eps})^{2-\delta})$, for any $\delta >0$, unless $(min,+)$-convolution has a subquadratic algorithm \cite{MuchaW019,Chan18}.

A PTAS for \multiknapsack was first  {presented} by Chekuri and Khanna~\cite{ChekuriK05} and EPTAS {s} due to Jansen~\cite{Jansen10mkp,Jansen12mkp}  {are} also known. {The fastest of these algorithms \cite{Jansen12mkp} has a} running time of $2^{\OO(\log^4(1/\eps)/\eps)} + n^{\OO(1)}$. The mentioned algorithms are all static and assume full knowledge about the instance for which a complete solution has to be found. {In particular, their solutions might change completely when a single item is added to the input which makes a full recomputation necessary. The algorithm in~\cite{ChekuriK05} invokes a guessing step with $n^{f(1/\eps)}$ many options which are too many for a polylogarithmic update time. The EPTASs in~\cite{Jansen10mkp,Jansen12mkp} use a configuration linear program of size $\Omega(n)$ which is also prohibitively large for such an update time.}

The dynamic arrival and removal of items exhibits some similarity to knapsack models with incomplete information. {For example,} in the {\em online}\ knapsack problem \cite{MarchettiV95} items arrive online one by one. When an item arrives, an algorithm must  {irrevocably} accept or reject it before the next item arrives. Various problem variants have been studied, e.g., with resource augmentation~\cite{IwamaZ10}, the removable online knapsack problem~\cite{IwamaT02,HanM10,HanKMG14,HanKM13,CyganJS16}, and with advice~\cite{BockenhauerKKR14}. Other models with uncertainty in the item set or the knapsack capacity include the \emph{stochastic} knapsack problem~\cite{DeanGV08,bhalgatGK11,Ma18} and \emph{robust} knapsack problems~\cite{yu96,MegowM13,DisserKMS17,busingKK11-discrete}. Related to our setting are also online models with a softened irrevocability requirement,  e.g., online optimization with {\em recourse} \cite{MegowSVW16,ImaseW91,GuG016,FeldkordFGGKRW18} or {\em migration} \cite{SandersSS09,SkutellaV16,JansenK19} allows to adapt previously taken decisions in a limited way. We are not aware of work on knapsack problems  {in these settings} and, again, the goal is to bound the amount of change needed to maintain good online solutions regardless of the computational effort.

\section{Roadmap and Preliminaries}
\label{sec:prelim}

{First, in this section, we formalize the operations that our data structures support, describe auxiliary data structures that we  need, and define how we round the item values. Then, in \cref{sec:singleKS}, we describe algorithms for one knapsack and for a polylogarithmic number of knapsacks. In \cref{sec:mik-and-mmdk}, we  present an algorithm for (many) identical knapsacks and an algorithm under resource augmentation (in the form of a polylogarithmic number of additional knapsacks) in the setting  of (many) knapsacks with possibly different capacities. Finally, we present in \cref{sec:general-diff-knapsack} an algorithm for the general case that uses the previously mentioned algorithms as subroutines. Additionally, in \cref{sec:hard-ks}, we show that our update time cannot be improved to $(\log n/\eps)^{\OO(1)}$, unless P=NP.}

From the perspective of a data structure that implicitly maintains near-optimal solutions for \multiknapsack, our algorithms support {several update and query operations which are listed below.
	They allow for the output of (parts of) the current solution, or for specific changes to the input of \multiknapsack, causing the computation of a new solution.} 
\begin{itemize}
	\item \textbf{Insert (Remove) Item:} Inserts (removes) an item into (from) the input.
	\item \textbf{Insert (Remove) Knapsack:} Inserts (removes) a knapsack into (from) the input.
\end{itemize}
A new solution can be output, entirely or in parts, using the following query operations.
\begin{itemize}
	\item \textbf{Query Item $\bm j$:} Returns whether item $j$ is packed in the current solution and if this is the case, additionally returns the knapsack containing it.
	\item \textbf{Query Solution Value:} Returns the value of the current solution.
	\item \textbf{Query Entire Solution:} Returns all items in the current solution,
	together with the information in which knapsack each such item is packed.
\end{itemize}
{Importantly,} queries are consistent in-between two update operations. {However, their answers} are not independent of each other but depend on the queries as well as their order.

For simplicity, we assume that elementary operations  {(e.g., additions)} can be handled in constant time. {Additionally, we assume \wwlog that~$\frac1\eps \in \N$.} {We also assume that at the very beginning we start with no items and no knapsacks, and initialize all needed auxiliary data structures accordingly. If one wants to start with a specific set of items and/or knapsacks, one can insert  {them} 
	with our insertion routines, using polylogarithmic time per insertion.}

\subparagraph*{Auxiliary Data Structures}
\label{sec:ds}
We employ auxiliary data structures in which we store (subsets of) input items and input knapsacks, sorted according to some specific values, e.g., size or capacity. We  need to be able to quickly access elements, compute the largest prefix of elements such that the sum according to some property, e.g., the total size, is below a given threshold, and compute in such a prefix the sum according to some element property, e.g., the total value. Note that these prefixes are w.r.t. the fixed ordering of the elements, while the 
element property for the threshold or computing the sum might be different. To this end, we employ as an auxiliary data structure a variation of balanced search trees that store elements according to some given ordering. For computing the mentioned prefix sums, we store in each internal node $v$ the sums of the elements in the subtree rooted at $v$ according to 
each property, e.g., size, value, or capacity. When we need to compute some largest prefix, we simply output the index of its last element.

\begin{lemma}\label{lem:data-structures:updates}\label{lem:data-structure}
	There is a data structure maintaining a sorting of~$n'$ elements  {w.r.t.\ to some key value} such that (i) insertion, deletion, or search by key value of an element takes~$\OO(\log n')$ time, and (ii) prefixes and prefix sums w.r.t.\ to  {any element property} can be computed in time $\OO(\log n')$.
\end{lemma}

\subparagraph*{Rounding Values}
A crucial ingredient  {of} our algorithms is the partitioning of items into only few \emph{value classes}~$V_\ell$,  {where for each $\ell$ the class $V_\ell$ consists of each input item~$j$ with $(1+\eps)^{\ell} \leq v_j < (1+\eps)^{\ell+1}$. }
Upon arrival of  {some} item~$j$, we calculate  {the index $\ell_j$ such that $j\in V_{\ell_j}$} and store  {the tuple $(j, v_j, s_j, \ell_j)$ representing $j$ in the auxiliary data structures of the respective algorithm. 
	In the following, we  pretend for each $\ell$ that each 
	item in~$V_\ell$ has value~$(1+\eps)^{\ell}$, which loses only a factor of $\frac1{1+\epsilon}$ in the total profit of any solution.}

\begin{lemma}\label{lem:Round}
	(i) There are at most~$\OO\big(\frac {\log \vvmax} \eps\big)$ many value classes.
	(ii) For optimal solutions~$\opt$ and~$\opt'$ for the original and rounded instance, $v(\opt') \geq (1-\eps) \cdot v(\opt)$.
\end{lemma}

\section{A Single Knapsack}
\label{sec:singleKS}
{In this section, we first present a dynamic algorithm for the case of one single knapsack, summarized in the following theorem. Afterwards, we will
	argue how to extend our techniques to the setting of a polylogarithmic number of knapsacks.}

\begin{restatable}{theorem}{singlethm}
	\label{thm:MDK:single}
	For~$\eps>0$, there is a fully dynamic algorithm for \knapsack that maintains $(1-\eps)$-approximate solutions with update time~$\OO\big(\frac{\smash{\log^4 (n \vvmax)}}{\eps^{9} }\big)+\OO\big(\frac{1}{\eps}\log n \log \vmax\big)$. %
	Furthermore, queries of single items and the solution value can be answered in time~$\OO(1)$.
\end{restatable}

We partition the items in the optimal solution \opt into high- and low-value items, respectively. The high-value items are the~$\frac{1}{\eps}$ most valuable items of~\opt, and the low-value items are the remaining items of~\opt. We compute a small set of candidate items $\candeps$ that intuitively contains all relevant high-value items in \opt. Also, we guess a placeholder item for the low-value items, that is large enough to accomodate low-value items of enough profit fractionally. We can assume that in an optimal fractional solution (of low-value items) at most one item is selected non-integrally. Hence, we can drop this item and charge it to the $\frac{1}{\eps}$ high-value items.
This results in a knapsack instance with only $\OO\left(\frac{1}{\epsilon^3}\right)$ items which we solve with an FPTAS.

{Formally, denote} by~\opteps a set of~$\frac{1}{\eps}$ most valuable items of~\opt.  {We break ties by picking smaller items.} 
Denote by~$V_{\lmax}$ and~$V_{\lmin}$ the highest resp.\ lowest value class of an element in~$\opteps$ and let~$n_{\min} \coloneqq \abs{\opteps \cap V_{\lmin}} \leq \frac{1}{\eps}$. %
Furthermore, denote by~$\vals$ the value of the items in~$\opt\setminus\opteps$, rounded down to  {the next} power of~$(1+\eps)$. 
To efficiently  {implement} our algorithm, we maintain several data structures,  {using Lemma~\ref{lem:data-structures:updates}}. {We} store items of each non-empty value class~$V_\ell$ (at most $\log_{1+\eps} \vvmax$) in a data structure ordered  {non-decreasingly} by size. 
Second, for each possible value class~$V_\ell$ (at most $\log_{1+\eps} \vmax$),  {we} maintain a data structure  {that contains each input item $j$ with
	$j \in V_{\ell'}$ for some $\ell' \le \ell$, ordered non-increasingly by density~$\frac{v_j}{s_j}$.} 
In particular, we maintain such a data structure even if~$V_\ell$  {itself} is empty  {(since
	the data structure might still contain items from classes $V_{\ell'}$ with $\ell' < \ell$).}
This leads to the additive term in the update time of~$\OO(\log n \log_{1+\eps} \vmax)$.
We use additional  {auxiliary} data structures %
to store our solution and support queries.

\subparagraph*{Algorithm} 
The algorithm computes an implicit solution as follows.
\begin{enumerate}%
	\item[1)] \textbf{Compute a set~\textnormal{\candeps} of high-value candidates:} 
	Guess the values%
	~${\lmax}$,~${\lmin}$, and~$n_{\min}$. 
	If~$\vc{\lmin} \geq \eps^2 \cdot \vc{\lmax}$, define~\candeps to be the set containing the~$\frac 1 {\eps}$ smallest items of each of the value classes~$V_{\lmin+1},\ldots,V_{\lmax}$, plus the~$n_{\min}$ smallest items from~$V_{\lmin}$.
	Otherwise, set~\candeps to be the union of the $\frac 1 \eps$ smallest items of each of the	value classes with values in~$[\eps^2 \cdot \vc{\lmax}, \vc{\lmax}]$.
	
	\item[2)] \textbf{Create a placeholder item~$B$:} %
	Guess~\vals and consider %
	items with value at most~$\vc{\lmin}$ sorted by density.
	Remove the~$n_{\min}$ smallest items of~$V_{\lmin}$ %
	until the next iteration.
	For the remaining items, compute the minimal size of fractional items necessary to reach a value~$\vals$. 
	{We do this via prefix sum computations on the data structure that contains all 
		items in $V_{\ell'}$ for each $\ell' \le \lmin$, ordered non-increasingly by density.
	}
	Then~$B$ is given by~$v_B=\vals$ and with~$s_B$ equal to the size of those low-value items.
	
	\item[3)] \textbf{Use an FPTAS:} %
	On the instance~$I$, consisting of~$\candeps$ and the placeholder item~$B$,
	run an FPTAS parameterized by $\eps$ (we use the one by
	Jin~\cite{Jin19}) to obtain a packing $P$. %
	
	\item[4)] \textbf{Implicit solution:} %
	Among all guesses, keep the solution~$P$ with the highest value.
	Pack items from~\candeps as in~$P$ and, if~$B\in P$, also pack the low-value items completely contained in~$B$  {(note that at most one item
		is packed fractionally in $B$)}.
	While used candidate  {items from $\candeps$} can be stored explicitly, low-value items are given only implicitly by saving the correct guesses and computing membership in~$B$ on a query.
\end{enumerate}

\subparagraph*{Analysis}
{We show that} the above algorithm attains an approximation ratio of~$(1-\eps)$.
{A factor of $(1-\eps)$ is lost due to the approximation ratio of the FPTAS. An additional factor of $(1-\eps)$ is lost in each of the following steps.} To obtain a candidate set~\candeps of constant cardinality, we restrict  {the} item values to~$[\eps^2 \cdot \vc{\lmax}, \vc{\lmax}]$. Since~$\abs{\opteps} = \frac 1 \eps$, this  {excludes items from $\opt$} with a total value of at most~$\frac 1 \eps \cdot \eps^2 \, \vc{\lmax} \leq \eps\cdot \opt$.
Furthermore, due to guessing \vals up to a power of $(1+\eps)$, we get $v_B = \vals \geq \frac1{1+\eps} \cdot v(\opt \setminus \opteps)$.
Finally, in Step~2,  {at most} one item was cut fractionally. It is charged to the $\frac 1 \eps$ items of \opteps,  {using that each of them has a} larger value.

The running time can be verified easily by multiplying the numbers of guesses for each value as well as the running time of the FTPAS. The latter is~$\OO\big(\frac 1 {\eps^4}\big)$, since we designed~\candeps to contain %
only a constant number of items, namely~$\OO\big(\frac 1 {\eps^3}\big)$ many.

\subparagraph*{Queries}
We show how to efficiently handle the different types of queries. %
\begin{itemize}
	\item \textbf{Single Item Query:}
	If the queried item is contained in~\candeps, its packing was saved explicitly.
	Otherwise, if~$B$ is packed, we save the last, i.e., least dense, item contained entirely in~$B$. By comparing with this item, membership in~$B$ can be decided in constant time on a query.
	\item \textbf{Solution Value Query:} While the algorithm works with rounded values, we use the data structures  {of} \cref{lem:data-structures:updates} 
	to  {retrieve} the actual  {item} values. %
	We store the actual solution value in the update step by adding the actual values of  {the packed items from \candeps} 
	and determining the actual value of items in~$B$ with a prefix computation. On query, we return the stored value.
	\item \textbf{Query Entire Solution:}
	Output the stored
	packing of candidates. If~$B$ was packed, iterate over items in~$B$ in the respective density-sorted data structure and output them.
\end{itemize}

\subparagraph*{Polylogarithmically many knapsacks}

{One can show that the queries can be performed in the claimed running times which completes the proof of \cref{thm:MDK:single}, see~\Cref{apx:single}. We can extend the above technique to the setting of $m$ knapsacks, at the expense of increasing the update time and query time by a factor $m^{\OO(1)}$, and using an EPTAS for \multiknapsack \cite{Jansen12mkp} instead of an FPTAS (see \cref{sec:MDK:few}).
}
{
	\begin{restatable}{theorem}{fewthm}
		\label{thm:MDK:few}
		For $\eps>0$, there is a dynamic algorithm for
		\multiknapsack %
		that achieves an approximation factor of~$(1-\eps)$ %
		with update time~$2^{f(1 / \eps)}\big(\frac m \eps \log (n\vvmax)\big)^{\OO(1)} + \OO\big(\frac 1 \eps \; \log \vmax \; \log n\big)$, with~$f$ quasi-linear.
		Item queries are answered in time~$\OO\big( \log \frac {m^2}{\eps^6}\big)$, solution value queries in time~$\OO(1)$, and queries of one knapsack or the entire solution in time linear in the output.%
\end{restatable}}

\section{Identical Knapsacks}
\label{sec:mik-and-mmdk}   

{In this section, we present our algorithm for an arbitrary (large) number of identical knapsacks. Also, we describe an extension to the case where the knapsacks have different sizes and we can use some additional knapsacks as resource augmentation.}

\subsection{ {Oblivious} Linear Grouping}\label{subsec:harmonicgrouping}

We start with our  {oblivious} linear grouping routine that we use in order to round the item sizes,  {aiming at} only few different types of items. We say that two items $j$, $j'$ are of the same \emph{type} if $\{j,j'\} \subseteq~ V_\ell$ for some $\ell$ and if $s_j = s_{j'}$. We round the items implicitly, i.e., we compute thresholds $\{\bar{s}_1,...,\bar{s}_k\}$ and we round up the size $s_j$ of each item $j$ to the next larger value in this set.

\begin{lemma}\label{theo:harmonic}
	Given a set~$J'$ with~$|\opt \cap J'| \leq n'$ for all optimal solutions~\opt, there is an algorithm with running time~$\smash{\OO\big(\frac{\log^5 n'}{\eps^5}\big)}$ that rounds the items in~$J'$ to item types~$\types$ with~$\smash{|\types| \leq \OO\big( \frac{\log^2 n' }{\eps^4}\big)}$ and ensures~$v(\opt_{\types}) \geq \frac{(1-\eps)(1-2\eps)}{(1+\eps)^2} v(\opt)$. Here,~$\opt_{\types}$ is the optimal solution attainable by packing item types~$\types$ instead of the items in~$J'$ and using~$J \setminus J'$ as is. 
\end{lemma}

\subparagraph*{Algorithm}
In the following, we use the notation~$X'$ for a set~$X$ to refer to~$X \cap J'$ while~$X''$ refers to~$X \setminus J'$. 
Recall that item values of items in~$J$ are rounded to powers of~$1+\eps$ to create the value classes~$V_\ell$ where each item~$j \in V_\ell$ has value $(1+\eps)^\ell$. 
{We guess~$\lmax$ which is defined to be the guess for the highest value $\ell$} with~$V_\ell' \cap \opt \neq \emptyset$ and let~$\llmin := \lmax - \left\lceil {\log_{1+\eps} (n'/\eps)} \right\rceil$.
\begin{enumerate}
	\item[1)]For each {$\ell$ with}~$\llmin \leq \ell \leq \lmax$ and each~$n_\ell = (1+\eps)^{ {k}}$  {with}~$0 \leq {k} \leq \log_{1+\eps} n'$ do: %
	Consider the $n_\ell$ smallest elements of $V_\ell'$ (sorted by increasing size) and 
	determine the~$\frac1\eps$ many (almost) equal-sized groups $G_1(n_\ell),\ldots, G_{1/\eps}(n_\ell)$ of~$\lceil \eps n_\ell \rceil $ or~$\lfloor \eps n_\ell \rfloor$ elements. 
	If~$\eps n_\ell \notin \N$, ensure that~$|G_k(n_\ell)| \leq |G_{k'}(n_\ell)| \leq |G_k(n_\ell)| +1$ for~$k \leq k'$. If~$\frac1\eps$ is not a natural power of~$(1+\eps)$,
	create~$G_1(\frac1\eps), \ldots,G_{1/\eps}(\frac1\eps)$ where~$G_k(\frac1\eps)$ is the~$k$th smallest item in~$V_\ell'$. 	
	Let~$G_1(n_\ell), \ldots, G_{1/\eps}(n_\ell)$ be the corresponding groups sorted increasingly by the size of the items. Let~$j_k(n_\ell) = \max\{ j \,:\, j \in G_k(n_\ell) \}$ be the last index belonging to group~$G_k(n_\ell)$. After having determined~$j_k(n_\ell)$ for each possible value~$n_\ell$ (including~$\frac1\eps$) and for each~$1 \leq k \leq \frac1\eps$, the size of each item~$j$ is rounded up to the size of the next larger item~$j'$ such that there exists $k$ and $\ell$  satisfying $j'=j_k(n_\ell)$. %
	
	\item[2)]  Discard each item $j$ with $j \in V_\ell'$ for $\ell < \llmin$.
\end{enumerate}

\subparagraph*{Analysis}
Despite the new approach to apply linear grouping simultaneously to many possible values of~$n_\ell$, the analysis builds on standard techniques. The loss in the objective function due to rounding item values %
is bounded by a factor of~$\frac{1}{1+\eps}$ by \cref{lem:Round}. As~$\llmin$ is chosen such that~$n'$ items of value at most~$\smash{(1+\eps)^{\llmin}}$ contribute less than an~$\eps$-fraction of $\opt'$, the loss in the objective function by discarding items in value classes $V_\ell'$ with $\ell < \llmin$ is bounded by a factor~$(1-\eps)$. By taking only~$(1+\eps)^{\lfloor \log_{1+\eps} n_\ell \rfloor}$ items of~$V_\ell '$ instead of~$n_\ell$, we lose at most a factor~$\frac{1}{1+\eps}$. The groups created by  {oblivious} linear grouping are an actual refinement of the groups created by classical linear grouping. Thus, we pack our items similarly: not packing the group with the largest items (at the loss of a factor of~$(1-2\eps)$) allows us to ``move'' all rounded items of group~$G_{k}(n_\ell)$ to the positions of the (not rounded) items in group~$G_{k+1}(n_\ell)$. %
Combining, we obtain~$v(\opt_{\types}) \geq \frac{(1-\eps)(1-2\eps)}{(1+\eps)^2} v(\opt)$.

Since~$\types$ contains at most~$\smash{\frac{1}{\eps} \big(\big \lceil \frac{\log n'/\eps}{\log (1+\eps)} \big\rceil +1 \big)}$ different value classes, and as it suffices to use~$\smash{\big \lceil \frac{\log n'}{\log (1+\eps)}\big\rceil +1}$ many different values for~$n_\ell = |\opt \cap V_\ell'|$, %
we have~$\smash{|\types| \leq \OO(\frac{\log^2 n'}{\eps^4})}$. Using the access times given in \cref{lem:data-structure} bounds the running time. For details, see \cref{app:harmonicgrouping}.

\subsection{A Dynamic Algorithm for Many Identical Knapsacks}
\label{sec:mik}

We give a dynamic algorithm with approximation ratio~$(1-\eps)$ for \multiknapsack,  {assuming that all knapsacks have the same size $S$}. {We assume $m\leq n$ as otherwise, the problem is trivial.} We focus on instances where~$m$ is large, i.e.,~$m \geq \frac{16}{\eps^7} \log^2 n$. {If}~$m \leq \frac{16}{\eps^7} \log^2 n$, we use the algorithm due to \cref{thm:MDK:few}. In the following, we prove \cref{theo:mik}.

\begin{theorem}\label{theo:mik}	
	If~$m \geq \frac{16}{\eps^7} \log^2n $, there is a dynamic algorithm for \multiknapsack with identical knapsacks with approximation factor $(1-\eps)$ and update time~$\smash{\big( \frac{\log U}{\eps} \big)^{\OO(1)}}$, where $U = \max\{\capa m, n \vvmax\}$. Queries for single items and the solution value can be answered in time ${\OO\big(\frac{\log n}{\eps}\big)^{\OO(1)}}$ and~$\OO(1)$, respectively.
	The solution~$P$ can be returned in time $\smash{|P|\big(\frac{\log n}{\eps}\big)^{\OO(1)}}$. %
\end{theorem} 

Our strategy is the following: we  partition the input items into large and small items, which are defined w.r.t.\ the size $S$ of each knapsack. To the large items, we apply  {oblivious} linear grouping, obtaining a polylogarithmic number of item types. We guess the total size of the small items in the optimal solution. Then, we  formulate the problem as a configuration linear program (LP) which has a variable for each feasible configuration for a knapsack. A configuration describes how many large items of each type are packed in a knapsack. Also, we  ensure that there will be enough space for the small items left. This is similar in spirit to the LPs used in \cite{Jansen10mkp,Jansen12mkp}; however, we use variables only for the configurations of the big items and we have only a polylogarithmic number of item types, which yields a smaller LP which we can solve faster. We  round the obtained fractional solution, using that $m > \frac{16}{\eps^7} \log^2 n$ and that basic feasible solutions to the LP are sparse.

\subparagraph*{Definitions and Data Structures}
We partition the items into two sets,~$J_B$, the \emph{big} items, and~$J_S$, the \emph{small} items, with sizes~$s_j \geq \eps \capa$ and~$s_j < \eps \capa$, respectively. For an optimal solution $\opt$, define $\opt_B := \opt \cap J_B$ and $\opt_S := \opt\cap J_S$.
We maintain three types of  {auxiliary} data structures from Lemma~\ref{lem:data-structure}:  we maintain one such data structure in which we store all items in the order of their arrivals and store the size~$s_j$, the value~$v_j$, %
and the value class~$\ell_j$ of each item~$j$. For each value class~$V_\ell$, we maintain a data structure which contains all big items of $V_\ell$, ordered non-decreasingly by size. Finally, for the small items (of all value classes together), we maintain a data structure in which they are sorted non-increasingly by density. Upon arrival of a new item $j$, we insert $j$ into each corresponding data structure.

\subparagraph*{Algorithm}

\begin{enumerate}%
	\item[1)] \textbf{Linear grouping of big items:} Guess~$\lmax$,  {which we define to be the largest index $\ell$ with $V_\ell \cap \opt_B \ne \emptyset$. Via}
	{oblivious} linear grouping with~$J'= J_B$ and~$n' = \min \{\frac m\eps, n_B\}$ we obtain~$\types$;
	{for each item type~$t$, denote by $n_t$ the number of items of this type (the multiplicity of $t$).} %
	\item[2)] \textbf{Configurations:}  {Let $\configs$ denote the set of all configurations, i.e., of all multisets of item types 
		whose total size is at most $S$. For each $c\in \configs$, denote by $v_c$ and $s_c$ the total value and size of the item types in~$c$.} 
	\item[3)] \textbf{Small items:}  {We guess $v_S$ which we define to be the largest power of $1+\eps$ that is at most $v(\opt_S)$}. Let~$P$ be the maximal prefix of small items (sorted by non-increasing density) with~$v(P) < v_S$. Set~$s_S := s(P)$. 
	\item[4)] \textbf{Configuration ILP:}  {We compute an extreme point solution of the LP relaxation of the} following configuration ILP with variables~$y_c$ for~$c \in \configs$ for the current guesses~$\lmax$ and~$v_S$ (implying $s_S$). Here,~$y_c$ counts how often a certain configuration~$c$ is used
	and $n_{tc}$ denotes the number of items of type $t$ in configuration $c$. 
	\begin{equation}\tag{P}\label{eq:mik:ilp}
		\begin{array}{llcll} 
			\max & \displaystyle{\sum_{c \in \configs} y_c v_c }\\
			\text{subject to } & \displaystyle{\sum_{c \in \configs} y_c s_c} & \leq & \lfloor(1-3\eps)m\rfloor \capa  - s_S \\
			& \displaystyle{\sum_{c \in \configs} y_c} & \leq & \lfloor(1-3\eps)m\rfloor \\
			& \displaystyle{\sum_{c \in \configs} y_c n_{tc}} & \leq & n_t & \text{for all } t \in \types \\ %
			& y_c & \in  & \mathbb{Z}_{\ge 0}& \text{for all } c \in \configs
		\end{array}
	\end{equation}	
	By the first inequality, the configurations fit into~$\lfloor(1-3\eps)m\rfloor$ knapsacks while reserving sufficient space for the small items. The second constraint limits the total number of configurations that are packed. The third inequality ensures that only available items are used. 	
	\item[5)] \textbf{Obtaining an integral solution:} %
	{We round up each variable of the obtained fractional solution, yielding an integral solution $\bar{y}$. 
		As~$m \geq \frac{16}{\eps^7} \log^2n$ and extreme point solutions have only $|\types|+2$ non-zero variables, one can show that $\bar{y}$ still satisfies the relaxed constraints 
		$\sum_{c \in \configs} \bar{y}_c s_c \leq \lfloor (1-2\eps)m \rfloor \capa  - s_S$ and $\sum_{c \in \configs} \bar{y}_c  \leq  \lfloor (1-2\eps)m \rfloor$. In case that a constraint 
		$\sum_{c \in \configs} \bar{y}_c n_{tc}  \leq  n_t$ is violated for some type $t$, we intuitively drop items of type $t$ from some knapsacks until the constraint is satisfied.
		Let $P_B$ denote the resulting packing.
	}
	\item[6)] \textbf{Packing small items:} Consider the maximal prefix~$P$ of small items with~$v(P) < v_S$ and let~$j^\star$ be the densest small item not in~$P$. Pack~$j^\star$ into one of the knapsacks kept empty by~$P_B$. Then, fractionally fill up the~$\lfloor (1-2\eps)m \rfloor$ knapsacks used by~$P_B$ and place any ``cut'' item into the~$\lceil \eps m \rceil$ additional knapsacks that are still empty. 
\end{enumerate}

\subparagraph*{Analysis} The loss in the objective function value due to linear grouping of big items is bounded by~$\frac{(1-\eps)(1-2\eps)}{(1+\eps)^2}$ by \cref{theo:harmonic}. Restricting a solution to its $\lfloor(1-3\eps)m\rfloor$ most valuable knapsacks and guessing the value of small items in these knapsacks only up to a factor of~$(1+\eps)$ as done by \eqref{eq:mik:ilp} costs at most a factor of~$\frac{1-4\eps}{1+\eps}$ in the objective function value. 

For solving the LP-relaxation of the configuration ILP (P), we apply the Ellipsoid method~\cite{GrotschelLS81} on its dual, using an FPTAS for \knapsack as a separation oracle. For this, we need to handle some technical complications due to the first two constraints of (P), which yield additional variables in the dual, and due to the fact that we can solve the separation problem only up to a factor of $(1+\eps)$ (see \cref{app:multiknapsack} for details). Via Gaussian elimination, we transform  {the obtained fractional solution into a basic feasible solution with the same objective function value}. {As argued above, since any basic feasible solution has at most $|\types|+2$ non-zero variables, our integral solution $\bar{y}$ uses at most $\lfloor (1-2\eps)m \rfloor$ knapsacks and it has at least the profit of the fractional solution.} Given the packing of big items, we pack the small items in a \textsc{First Fit} manner as described in the algorithm.

{To bound the running time of our algorithm, we use \cref{theo:harmonic}, show that the relaxation of the configuration ILP can be solved in time 
	$\smash{\big( \frac{\log U}{\eps} \big)^{\OO(1)}}$ with the Ellipsoid method, and use the fact that the algorithm needs at most $\smash{\OO\big( \frac{\log (n \vvmax) \log \vvmax}{\eps^2} \big)}$ many guesses, see \cref{app:multiknapsack} for details.}

\subparagraph*{Queries}
In contrast to the previous section, for transforming an implicit solution into an explicit packing, 
the query operation has to compute the knapsack where  {a queried} item~$j$ is packed. 
We do not explicitly store the packing of any item,  {but instead} we define and update pointers for small items and for each item type, that  {indicate} the knapsacks where the corresponding items are packed. To stay consistent  {with} the precise packing of a particular item between two update operations, we additionally cache query answers. %
\begin{itemize}
	\item \textbf{Single Item Query:} For small items, only the prefix of densest items is part of our solution. For big items of a certain type, only the smallest items are packed by the implicit solution. In both cases, we use the corresponding pointer to determine the knapsack.
	\item \textbf{Solution Value Query:} As the algorithm works with rounded values, we use prefix computation{s} on the small items and on any value class of big items to calculate and store the  {current} solution value.  {Given a} query, we return the stored solution value.
	\item \textbf{Query Entire Solution:} We use prefix computation{s} on the small items as well as on the value classes of the big items to determine the packed items. Then, we use the Single Item Query to determine their respective knapsacks. 
\end{itemize}

\begin{restatable}{lemma}{mikQueries}\label{lem:mik:queries} The solution determined by the query algorithms is feasible and achieves the claimed total value. The query times of our algorithm are as follows: 
	Single item queries can be answered in time $\OO\big(\log n + \max \big\{ \log \frac{\log n}{\eps}, \frac{1}{\eps}\big\} \big)$, 
	solution value queries can be answered in time~$\OO(1)$, and 
	queries of the entire solution~$P$ can be answered in time $\smash{\OO\big(|P|  \frac{\log^4 n}{\eps^4} \log \frac{\log n}{\eps}  \big)}$.
\end{restatable}

We extend our techniques above to an algorithm for knapsacks of arbitrary sizes, assuming that we have $\big(\frac{\log n}\eps\big)^{\Theta(\epsfrac)}$ additional knapsacks (of capacity at least as large as the largest original knapsack) as resource augmentation available.
The intuition is that these additional knapsacks are sufficient to compensate errors when rounding the LP-relaxation of (P). However, additional care is needed since whether an item is big or small now depends on the knapsack.
{
	\begin{restatable}{theorem}{mmdkthm}
		\label{theo:mmdk}
		For~$\eps >0$, there is a dynamic algorithm for \multiknapsack that, given $\big(\frac{\log n}\eps\big)^{\Theta(\epsfrac)}$ additional knapsacks as resource augmentation, achieves an approximation factor of~$(1+\eps)$ with update time~$\big(\frac1\eps \log n \big)^{\OO(\epsfrac)} (\log m \log S_{\max} \log \vvmax)^{\OO(1)}$. Item queries are answered in time~$\OO\big(\log m + \frac{\log n}{\eps^2}\big)$, and the solution~$P$ %
		is output in time~$\OO\big(|P|\frac{\log^3 n}{\eps^4}\big(\log m + \frac{\log n}{\eps^2}\big) \big)$.
	\end{restatable}
}

\section{Solving \multiknapsack}
\label{sec:general-diff-knapsack}

Having laid the groundwork with the previous two sections, we finally show how to maintain solutions for arbitrary instances of the \multiknapsack problem, and give the main result of this paper,  {summarized in the following theorem.} Note that we assume~$n \geq m$ as otherwise only the~$n$ largest knapsacks are used. 
\begin{theorem}
	\label{thm:MDK:gen}	For $\eps>0$, there is a dynamic,~$(1-\eps)$-approximate algorithm for \multiknapsack with update time $2^{f(\epsfrac)} \big(\frac1\eps \log n \log \vvmax \big)^{\OO(\epsfrac)} (\log S_{\max})^{\OO(1)} + \OO\big(\frac 1 \eps \log \vmax \, \log n\big)$, where~$f$ is quasi-linear. Item queries are served in time $\OO\big (\smash{\frac {\log  n}{\eps^2}}\big)$ and the solution~$P$ can be output in time~$\OO\big(\smash{\frac{\log^4 n}{\eps^6} \abs{P}}\big)$.
\end{theorem}
We obtain this result by partitioning the knapsacks into three sets, \green, \yellow and \red knapsacks, and solving the respective subproblems.
This has similarities to the approach in \cite{Jansen10mkp}; however, there it 
was sufficient to have only two groups of knapsacks. 
On a high level, the special knapsacks are the $\left(\log n\right)^{\OO(1/\eps)}$ largest input knapsacks and, intuitively, we  apply the algorithm due to \cref{thm:MDK:few} to them (for a suitably defined set of input items). The extra knapsacks are $\left(\log n\right)^{\OO(1/\eps)}$ knapsacks that are smaller than the \green knapsacks, but larger than the \red knapsacks. We ensure that there is a (global) $(1-\eps)$-approximate solution in which they are all empty. We  apply the algorithm due to \cref{theo:mmdk} to the \red and \yellow knapsacks, where the \yellow knapsacks  form the additional knapsacks used as resource augmentation.

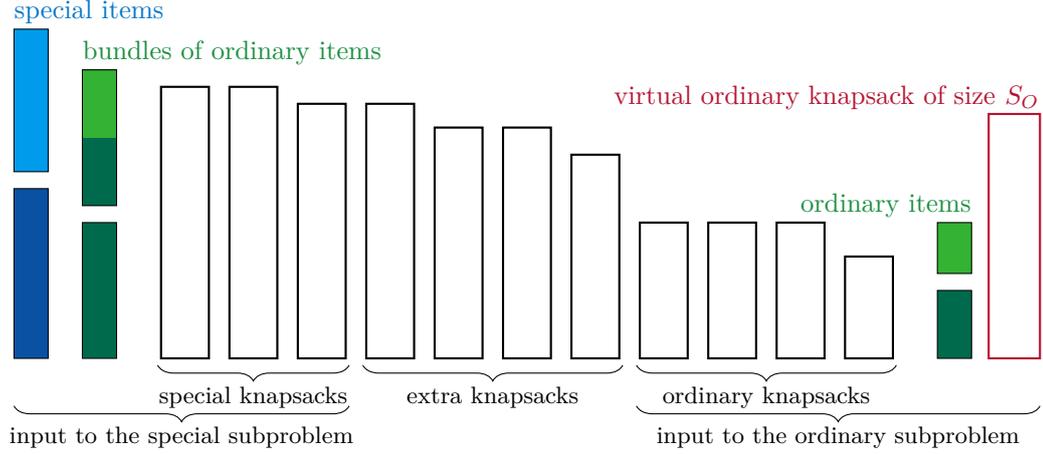
\begin{figure}
\centering
\begin{tikzpicture}[scale = .9]
	
	\foreach \x/\h/\n in {12/1.5/0, 9/2/2, 8/3/0, 6/3.4/1, 4/3.75/1, 2/4/1}{ %
		\foreach \y in {0,...,\n}{
			\draw[draw=black, fill=none, thick] (.25+\x+\y-.1,0) rectangle (.25+\x+\y+.6, \h); 
		}
	}

	\draw[draw=ubred, fill=none, thick] (14.25, 0) rectangle (15, 3.6);  
	
	\node[anchor = east, ubred, inner sep = 0pt] at (15, 3.85) {virtual ordinary knapsack of size $S_O$};
	
	\draw [decorate,decoration={brace,amplitude=6pt}]  (4.9, -.1) -- (2.1,-.1) node [black,midway,yshift =  -12pt,font=\small, black] {special knapsacks};
	\draw [decorate,decoration={brace,amplitude=6pt}]  (8.9, -.1) -- (5.1,-.1) node [black,midway,yshift =  -12pt,font=\small, black] {extra knapsacks};
	\draw [decorate,decoration={brace,amplitude=6pt}]  (12.9, -.1) -- (9.1,-.1) node [black,midway,yshift =  -12pt,font=\small, black] {ordinary knapsacks};
	
	\node[anchor = west, utnavy!50!utblue, inner sep = 0pt] at (0,5.1) {special items};
	\draw[draw = black, fill = utnavy] (0, 0) rectangle (.5, 2.5);
	\draw[draw = black, fill = utblue] (0,2.75) rectangle (.5, 4.85);  
	
	
	\node[anchor = west, utforest!50!utgreen, inner sep = 0pt] at (1,4.5) {bundles of ordinary items};
	\draw[draw = black, fill = utforest] (1,0) rectangle (1.5,2); 	
	\fill[utforest] (1,2.25) rectangle (1.5,3.25);
	\fill[utgreen] (1,3.25) rectangle (1.5,4.25);
	\draw[draw = black, fill = none] (1,2.25) rectangle (1.5, 4.25);	
	
	\node[anchor = east, utforest!50!utgreen, inner sep = 0pt] at (14,2.25) {ordinary items};
	\draw[draw = black, fill = utforest] (13.5,0) rectangle (14,1); 
	\draw[draw = black, fill = utgreen] (13.5, 1.25) rectangle (14,2);

	\draw [decorate,decoration={brace,amplitude=6pt}]  (4.9, -.7) -- (0,-.7) node [black,midway,yshift =  -12pt,font=\small, black] {input to the special subproblem};
	\draw [decorate,decoration={brace,amplitude=6pt}]  (15, -.7) -- (9.1,-.7) node [black,midway,yshift =  -12pt,font=\small, black] {input to the ordinary subproblem};
	
\end{tikzpicture}
\caption{Input of the special and the ordinary subproblems: Based on the current guess for the extra knapsacks, the knapsacks are partitioned into three groups (special, extra, and ordinary). When an item fits into at least one ordinary knapsack, it is ordinary and special otherwise. The total size of ordinary items placed by \opt in special knapsacks gives the size of the virtual ordinary knapsack. The ordinary items packed into this virtual knapsack are further assigned to bundles of equal size, which are then part of the input to the special subproblem.} \label{fig:general:partition}
\end{figure}

\subparagraph*{Definitions and Data Structures}

Let $L=\big(\frac{\log n}\eps\big)^{\Theta(\epsfrac)}$.
We assume that~$m > {\big(\frac 1 \eps\big)^{  4 /\eps}} \cdot L$, since otherwise we simply apply \Cref{theo:mmdk}.
Consider $\frac{1}{\eps}$ groups of knapsacks with sizes $\frac L {\eps^{3i}}$, for $i = 0,1,\ldots,\frac 1 \eps -1$, such that the first group, i.e., $i=0$, consists of the $L$ largest knapsacks, the second, i.e., $i=1$, of the $\frac L{\eps^3}$ next largest, and so on. In \opt, one of these contains items with total value at most $\eps\cdot\opt$. Let~$k\in\{0,1,\ldots,\frac 1 \eps -1\}$ be the index of such a group and let $L_\gr:= \sum_{i=0}^{k-1} \frac L {\eps^{3i}}$.
We define the  $L_\gr$ largest input knapsacks to be the \emph{\green} knapsacks. The \emph{\yellow} knapsacks are the~$\frac L{\eps^{3k}} > \frac{L_\gr}{\eps^2} + L$ next largest, and the \emph{\red} knapsacks the remaining ones.

Call an item \emph{\red} if it fits into the largest \red knapsack and \emph{\green} otherwise. Denote by $J_\re$ and $J_\gr$ the set of \red and \green items, respectively, and by  {$S_\re$ the total size of \red items that \opt places in \green knapsacks, rounded down to the next power of $(1+\eps)$}; see Figure~\ref{fig:general:partition}. 
Since we use the algorithms from \Cref{thm:MDK:few,theo:mmdk} as subroutines, we require the maintenance of the corresponding data structures. 
\subparagraph*{Algorithm}

\begin{enumerate}%
	\item[1)] \textbf{ {Oblivious} linear grouping:} Compute $\OO\big(\frac{\log^2 n}{\eps^4}\big)$ item types as described in \Cref{subsec:harmonicgrouping}. %
	Guess~$k$ and determine whether items of a certain type are \red or \green.
	
	\item[2)] \textbf{High-value \red items: }%
	Place each of the $\frac{L_\gr}{\eps^2}$ most valuable \red items in an empty \yellow knapsack. On a tie choose the larger item.
	Denote this set of items by~$J_\ye$. 
	
	\item[3)] \textbf{Virtual \red knapsack:} Guess~$S_\re$ and add a virtual knapsack with capacity $S_\re$ to the \red subproblem. 
	In the LP used in the proof of \Cref{theo:mmdk}, treat every ordinary item as small item in this knapsack and do not use configurations.

	\item[4)] \textbf{Solve \red instance:} 
	Remove temporarily the set $J_\ye$
	from the data structures of the \red subproblem. Solve the subproblem with the virtual knapsack as in \Cref{theo:mmdk} and use \yellow knapsacks for resource augmentation. When rounding up variables, fill the~$\OO\big(\frac{\log^2 n}{\eps^4}\big)$  %
	rounded items from the virtual knapsack into \yellow knapsacks.
	
	\item[5)] \textbf{Create bundles}\label{step:MDKG:bundles}
	Consider the items that remain in the virtual \red knapsack after rounding.
	Sort them by type (first value, then size) and cut them to form $\frac{L_\gr}\eps$ bundles~$B_\re$ of equal size.
	For each bundle, remember how many items of each type are placed entirely inside it.
	Place cut items into \yellow knapsacks. 
	Consider each $B \in B_\re$ as an item of size and value equal to the fractional size respectively value of items placed entirely in $B$.
	
	\item[6)] \textbf{Solve \green instance:} Temporarily insert the bundles in $B_\re$ into the data structures used in the \green subproblem. Solve this subproblem 
	{with the algorithm due to \cref{thm:MDK:few}}.
	
	\item[7)] \textbf{Implicit solution:}
	Among all guesses, keep the solution~$P_F$ with the highest value.
	Store items in~$J_\ye$ and their placement explicitly.
	Revert the removal of~$J_\ye$ from the \red data structures after the next update.
	For the remaining items, the solutions are given as in the respective subproblem, with the exception of items packed in the virtual \red knapsack.
	The solution of these items is stored implicitly by deciding membership in a bundle on a query.
\end{enumerate}

\subparagraph*{Queries}
We essentially use the same approach as in  {\cref{thm:MDK:few,theo:mmdk}} for the \red and \green subproblem, respectively. 
However, special care has to be taken with items in the virtual knapsack.
In the \red subproblem, we assume that items of a certain type which are packed in the virtual knapsack are the first, i.e., smallest, of that type.
We can therefore decide in constant time whether or not an item is contained in the virtual knapsack and, if this is the case, fill it into the free space in \green knapsacks reserved by bundles.
We do this efficiently by using a first fit algorithm on the knapsacks with reserved space.
Since items in \yellow knapsacks are stored explicitly, they can be accessed in constant time. See \cref{app:general} for details.

\subparagraph*{Hardness of approximation} 
It is a natural question whether the update time of our algorithms for \multiknapsack  
can be improved to $\big(\frac1\eps \log n\big)^{\OO(1)}$. We show that this is impossible, unless P=NP; see \cref{sec:hard-ks}.
\begin{restatable}[]{theorem}{HardnessMultiKnapsack}
	\label{theo:HardnessMultiKnapsack}
	Unless $\classP=\NP$, there is no fully dynamic algorithm for \multiknapsack that maintains a $(1-\eps)$-approximate solution in update time polynomial in $\log n$ and $\frac 1 \eps$, for~$m < \frac{1}{3\eps}$. 
\end{restatable}


%
\section{Conclusion}
\label{sec:conclusion}
Any dynamic algorithm can be turned into a non-dynamic one by having~$n$ items arrive one by one,  incurring  an additional linear  {factor} in the running time. Hence,  {lower bounds for the running times of static approximation schemes yield lower bounds for update times of dynamic algorithms.} Our  {running times for} the problems with identical capacities are tight in the sense that the algorithms yield a static FPTAS  (resp. EPTAS) matching  {known lower bounds}. 

Clearly, it would be interesting to generalize our results beyond \multiknapsack.
A natural generalization is~$d$-dimensional \knapsack, where the items and knapsacks have a size in each of the $d$ dimensions, and a feasible packing of a subset of items must meet the capacity constraint in each dimension. {A} reduction to one dimension by~\cite{VegaL81} immediately yields a dynamic $\frac{1-\eps}d$-approximation, but designing a dynamic framework with a better guarantee than this remains~open. Note that unless~$\text{W}[1] = \text{FPT}$, 2-dimensional knapsack \emph{does not} admit a dynamic algorithm maintaining a~$(1-\eps)$-approximation in worst-case update time~$f(\eps) n^{\OO(1)}$~\cite{KulikS10}. 

A recent line of research exploits fast techniques for solving  convolution problems to speed up knapsack algorithms (exact and approximate); see, e.g.,~\cite{Chan18,Jin19,AxiotisT19,PolakRW21,KellererP04}. In fact, it has been shown that \knapsack is computationally equivalent to the $(\min,+)$-convolution problem~\cite{CyganMWW19}. It seems worth exploring whether such techniques are useful in the dynamic setting. Here, it is unclear whether the re-computation of a solution in a new iteration can be done in polylogarithmic time. It is also open whether such techniques can be applied for solving \multiknapsack, even in the static setting.

We hope to foster further research for other packing, scheduling and, generally, non-graph problems. For bin packing and for makespan minimization on uniformly related machines, we notice that existing PTAS techniques from \cite{Karmarkar1982} and \cite{Jansen10,HochbaumS87} combined with rather straightforward data structures can be lifted to a fully dynamic algorithm framework for the respective problems.

\begin{small}
	\bibliography{dynamic-packing}
\end{small}

\clearpage

\appendix
\section*{Appendices}

\section{Proofs for Single Knapsack}
\label{apx:single}

In this section, we give the detailed analysis of our algorithm for \knapsack in \Cref{sec:singleKS}. We consider the iteration in which the guesses~${\lmax},{\lmin}, n_{\min}$ and~$\vals$ are correct and show that the obtained solution has a value of at least~$(1-4\eps)\cdot v(\opt)$. 

Let~$\mcp_1$ be the set of solutions respecting:
\begin{enumerate*}[label=(\roman*), ref=(\roman*),leftmargin=9mm]
	\item\label{prop:single:prop1}packed items not in~\candeps have a value of at most~$\vc{\lmin}$ but are not part of the~$n_{\min}$ smallest items of the value class~$V_{\lmin}$, and
	\item\label{prop:single:prop2}the total value of these items lies in~$[\vals,(1+\eps)\vals]$. 
\end{enumerate*}~
Denote by~$\opt_1$ the solution of highest value in~$\mcp_1$.

\begin{lemma} \label{lem:single:candeps}
	Consider~$\opt_1$ defined as above. Then,~$v(\opt_1) \geq (1-\eps) \cdot v(\opt)$.
\end{lemma}

\begin{proof}
	Let~$\opt^*$ be the packing obtained from~\opt by removing
	all items belonging to~$\opteps$ whose value is strictly smaller than~$\eps^2 \, \vc{\lmax}$. 
	Since~\opteps consists of~$\frac 1 \eps$ many items, the total value of removed items is at most~$\frac 1 \eps \cdot \eps^2 \, \vc{\lmax} \leq \eps\cdot \opt$.
	We show that~$\opt^* \in \mcp_1$.
	
	Consider an item~$j$ in~$\opteps$ of value~$v_j\geq \eps^2 \, \vc{\lmax}$.
	If~$v_j=\vc{\lmin}$, then~$j \in \candeps$ by definition of~$n_{\min}$ and~$\opteps$, specifically, due to the tie-breaking rules.
	Assume now that~$v_j>\vc{\lmin}$ and~$j \notin \candeps$. Recall that~\candeps
	contains the~$\frac 1 \eps$ smallest items of value~$v_j$, and~$|\opteps| = \frac 1 \eps$.
	Thus, there exists an item of value~$v_j$, smaller than~$j$, which
	belongs to~\candeps but not to~$\opteps$.
	Exchanging~$j$ for this item contradicts
	the definition of~\opt. Therefore,~$j\in\candeps$ and Condition~\ref{prop:single:prop1} is satisfied.
	Condition~\ref{prop:single:prop2} follows directly from the definition of~$\vals$, and therefore~$\opt^*\in \mcp_1$, concluding the proof.
\end{proof}

\begin{lemma}\label{lem:single:FPTAS}
	Let $\opt_2$ be the optimal solution of the instance~$I$ on which the FPTAS is run at Step~3.%
	Then,~$v(\opt_2) \geq (1-\eps)\cdot v(\opt_1)$.
\end{lemma}

\begin{proof}
	Consider the fractional solution~$\opt_1^*$ for~$I$ that is obtained from~$\opt_1$ as
	follows.
	Place items from~\candeps as in~$\opt_1$ and additionally place the placeholder item~$B$.
	Denote by~$J_L$ the set of items packed by~$\opt_1$ that are not in~\candeps, i.e., the low-value items.
	By definition of~$B$, we have~$v_B = \vals \geq \frac1{1+\eps} v(J_L) \geq (1-\eps) V(J_L)$.
	Further, since~$B$ consists of the densest low-value items, it must be the case that~$s_B \leq s(J_L)$.
	Therefore,~$\opt_1^*$ is a feasible solution for~$I$ and the statement follows.
\end{proof}

\begin{lemma}
	\label{lem:single:lastval}
	For the solution~$P_F$ of the algorithm, we have $v(P_F) \geq (1-4\eps) \cdot v(\opt)$.
\end{lemma}

\begin{proof}
	The solution~$P_\mathrm{FPTAS}$ returned by the FPTAS in Step~3 %
	has a value of at least~$(1-\eps)\cdot v(\opt_2)$. The solution~$P_F$ is obtained from~$P_\mathrm{FPTAS}$ by replacing the placeholder with the corresponding low-value items, except possibly the fractional item $j$.
	Since there are~$\frac 1 \eps$ items in~$\opt$ that are of higher value than~$j$, namely the ones in $\opteps$, this implies
	$$v(P_F) \geq v(P_\mathrm{FPTAS}) - \eps\cdot v(\opt).$$
	Using \Cref{lem:single:candeps,lem:single:FPTAS}, we obtain:
	\begin{align*}
		v(P_F) &\geq v(P_\mathrm{FPTAS}) - \eps\cdot v(\opt)\\
		&\geq (1-\eps)^2 \cdot v(\opt_1)-\eps \cdot v(\opt)\\
		&\geq (1-\eps)^3 \cdot v(\opt)-\eps \cdot v(\opt)\\
		&\geq (1-4\eps) \cdot v(\opt).
	\end{align*}
\end{proof}

\begin{lemma}
	\label{lem:single:time}
	The algorithm has update time $\OO\big(\frac{1}{\eps^{9}} \cdot \log n \cdot \log (n \cdot \vvmax) \cdot \log^2 \vvmax + \frac 1 \eps \log \vmax \log n\big)$. 
\end{lemma}

\begin{proof}	
	In the first step, guessing~${\lmax}$ and ${\lmin}$, and therefore enumerating over all possible values, leads to $\OO(\frac 1 {\eps^2}  \cdot \log^2 \vvmax)$ many iterations. Guessing~$n_{\min}$ adds an additional factor of~$\frac 1 \eps$.
	
	In the second step, again guessing~\vals adds a factor to the running time, specifically~$\OO(\frac 1\eps\log (n\cdot
	\vvmax))$. Temporarily removing the~$n_{\min}\leq \frac 1 \eps$ elements from the data structure costs a total of~$\OO(\frac 1 \eps \log n)$, as does adding back removed items from a previous iteration. Computing the size of~$B$ can be done by querying the
	prefix of value just above~$\vals$ in time~$\OO(\log n)$, see \cref{sec:ds}.
	
	For Step~3, %
	note that the set~\candeps spans value classes ranging from values of~$\eps^2 \cdot \vc{\lmax}$ or higher to~$\vc{\lmax}$.
	As values are rounded to powers of~$(1+\eps)$, we consider at most~$\log_{1+\eps} \frac 1 {\eps^2}$ many.
	Hence,~\candeps is composed of~$\OO(\frac 1 {\eps^3})$ items and the FPTAS %
	runs in time~$\OO\big (\big ( \frac{1}{\eps^{9/4}}\frac{1}{\eps^{3/2}} + \frac{1}{\eps^2}\big) /2^{\Omega(\sqrt{\log(1/\eps)})}\big) = 
	\OO(\frac 1 {\eps^4})$.
	
	Recall, that we need to maintain one data structure for every existing and one for each possible value class, that is,~$\OO(\frac 1 \eps \log \vmax)$ many data structures in total.
	Maintenance of these, i.e., insertion or deletion of an item, takes time~$\OO(\frac 1 \eps \log \vmax \log n)$ in total.
\end{proof} 

\begin{lemma}\label{lem:single:query}
	The query times of our algorithm are as follows.
	\begin{inparaenum}[(i)]
		\item Single item queries are answered in time~$\OO(1)$.  
		\item Solution value queries are answered in time~$\OO(1)$.
		\item Queries of the entire solution~$P$ are answered in time~$\OO(\abs{P})$.
	\end{inparaenum}
\end{lemma}

\section{Few Different Knapsacks}
\label{sec:MDK:few}
It is not very difficult to extend the approach from \Cref{sec:singleKS} to the case of multiple but few knapsacks.
While theoretically applicable for any number of knapsacks, the running time is reasonable when~$m=(\frac 1\eps\log n )^{\OO_\eps(1)}$.
The main difference to \cref{sec:singleKS} comes from the fact that in order to reserve space for low-value items, a single placeholder is no longer sufficient.
Instead, we utilize several smaller placeholders.
Since guessing the size of low-value items for every knapsack would lead to a running time exponential in~$m$, we instead employ a sufficiently large number of placeholder items, namely~$\frac{m}{\eps}$ many.

This leads to additional changes as there are more fractionally cut items, i.e., one per placeholder.
To be able to charge them as before in \cref{lem:single:lastval}, we now consider the~$\frac{m}{\eps^2}$ most profitable items in~\opt.
This in turn leads to a larger candidate set of size~$\smash{\frac{m}{\eps^2}}$.
Furthermore, since we consider multiple knapsacks, we need to utilize an EPTAS instead of an FPTAS.
Besides these changes, the algorithm remains unchanged.

\fewthm*

\paragraph*{Definitions and Data Structures}
Let~$\opt$ be the set of items used in an optimal solution and~$\optmeps$ the set containing the~$\frac m {\eps^2}$ most valuable items of~$\opt$; in both cases, break all ties by picking smaller-size items.
Further, denote by~$V_{\lmax}$ and~$V_{\lmin}$ the highest and lowest value (class) of an element in~$\optmeps$ respectively and by~$n_{\min}$ the number of elements of~$\optmeps$ with value~$\vc{\lmin}$. 
Let~$\vals$ be the total value of the items in~$\opt\setminus\optmeps$, rounded down to a power of~$(1+\eps)$. 
The data structures used are identical to those of \Cref{sec:singleKS}.

\paragraph*{Algorithm}
\begin{enumerate}[topsep=6pt,itemsep=0pt,parsep=6pt,partopsep=0pt, label = \arabic*)]
	\item \textbf{Compute high-value candidates \textnormal{\candmeps}:} 
	Guess the three values~${\lmax}$,~${\lmin}$ and~$n_{\min}$. 
	If~$\vc{\lmin} \cdot m \geq \eps^3 \cdot \vc{\lmax}$, then define~\candmeps to be the set that contains the~$\frac m {\eps^2}$ smallest items of each of the value classes~$V_{\lmin+1},\ldots,V_{\lmax}$, plus the~$n_{\min}$ smallest items from~$V_{\lmin}$.\\	
	Otherwise, we set~\candmeps to be the union of the $\frac m {\eps^2}$ smallest items of each of the value classes with values in~$[\frac{\eps^3}{m} \cdot \vc{\lmax}, \vc{\lmax}]$.

	\item \textbf{Create bundles of \lowval items as placeholders:} %
	Guess the~value \vals and consider the data structure containing all the items of value at most~$\vc{\lmin}$ sorted by decreasing density.
	Remove from it (temporarily) the~$n_{\min}$ smallest items of value~$\vc{\lmin}$.
	Insert them back into the data structure right before the next iteration.
	From the remaining items, compute the amount of fractional items necessary to reach a value of~$\vals$. 
	That is, sum the sizes of the densest items until their total value equals~$\vals$ and, if necessary, cut the last item fractionally.
	In the same manner, cut this range of items again fractionally to obtain bundles~$B_1, B_2,\dots, B_{\frac m \eps}$ of equal value~$\frac \eps m \cdot \vals$.
	
	\item \textbf{Use an EPTAS:} %
	Consider the instance~$I$ consisting of the items in~\candmeps and the placeholder bundles~$B_1, B_2,\dots, B_{\frac m \eps}$.
	Run the EPTAS designed by Jansen~\cite{Jansen10mkp,Jansen12mkp}, parameterized by $\eps$, to obtain a packing $P$ for this instance. %
	
	\item \textbf{Implicit Solution:} %
	Among all guesses, keep the feasible solution~$P$ with the highest value.
	Then, for any knapsack, place into the knapsack items from~\candmeps as in~$P$ and, if~$B_k$ is placed in~$P$ on this knapsack, also place the low-value items that constitute~$B_k$, except possibly items cut fractionally.
	While used candidates can be stored explicitly, low-value items are given only implicitly by saving the correct guesses and recomputing~$B_k$ on a query.
\end{enumerate}

\paragraph*{Analysis}

The analysis is almost identical to that of \cref{sec:singleKS} with only slight changes to accommodate the alterations described above.
For completeness, we give the full proofs.
We consider the iteration in which all guesses~(${\lmax},{\lmin}, n_{\min}, \vals$) are correct, and show that the obtained solution has a value of at least~$(1-6\eps)\cdot v(\opt)$.
To this end, we consider intermediate results to analyze the impact of each step.

Let~$\mcp_1$ be the set of solutions respecting:
\begin{enumerate*}[label=(\roman*), ref=(\roman*),leftmargin=9mm]
	\item\label{prop:few:prop1}items not in~\candmeps have a value of at most~$\vc{\lmin}$ but are not part of the~$n_{\min}$ smallest items of the value class~$V_{\lmin}$, and
	\item\label{prop:few:prop2}the total value of these items lies in~$[\vals,(1+\eps)\vals]$. 
\end{enumerate*}
Denote by~$\opt_1$ the solution of highest value in~$\mcp_1$.

\begin{lemma} \label{lem:FMDK:candmeps}
	Consider~$\opt_1$ defined as above. Then,~$v(\opt_1) \geq (1-\eps) \cdot v(\opt)$.
\end{lemma}

\begin{proof}
	Let~$\opt^*$ be the packing obtained from~\opt by removing
	all items belonging to~$\optmeps$ whose value is strictly smaller than~$\frac{\eps^3}{m} \cdot \vc{\lmax}$. 
	Since~\optmeps consists of~$\frac m {\eps^2}$ many items, the total value of removed items is at most~$\smash{\frac{m}{\eps^2}\cdot \frac{\eps^3}{m}} \cdot \vc{\lmax} \leq \eps \cdot \opt$.
	We show that~$\opt^* \in \mcp_1$.
	
	Consider an item~$j$ in~$\optmeps$ of value~$v_j\geq \frac{\eps^3}{m} \cdot \vc{\lmax}$.
	If~$v_j=\vc{\lmin}$, then~$j \in \candmeps$ by definition of~$n_{\min}$ and~$\optmeps$; specifically, due to the tie-breaking rules.
	Assume now that~$v_j>\vc{\lmin}$ and~$j \notin \candmeps$. Recall that~\candmeps	contains the~$\frac m {\eps^2}$ smallest items of value~$v_j$, and~$|\optmeps| = \frac m {\eps^2}$.
	Thus, there exists an item of value~$v_j$, smaller than~$j$, which
	belongs to~\candmeps but not to~$\optmeps$.
	Exchanging~$j$ for this item contradicts
	the definition of~\opt. Therefore,~$j\in\candmeps$ and Condition~\ref{prop:few:prop1} is satisfied.
	Condition~\ref{prop:few:prop2} follows directly from the definition of~$\vals$, and therefore~$\opt^*\in \mcp_1$, concluding the proof.
\end{proof}

\begin{lemma}\label{lem:FMDK:EPTAS}
	Let $\opt_2$ be the optimal solution of the instance~$I$ on which the EPTAS is run in Step~3. %
	Then,~$v(\opt_2) \geq (1-2\eps)\cdot v(\opt_1)$.
\end{lemma}

\begin{proof}
	Consider the fractional solution~$\opt_1^*$ for~$I$ that is obtained from~$\opt_1$ as follows.
	First, place items from~\candmeps as in~$\opt_1$. Next, consider the placeholder bundles~$B_1, B_2,\dots, B_{\frac m \eps}$ in any order, and place them fractionally into the remaining space.
	That is, place remaining bundles in the first non-full knapsack. If a bundle does not fit, fill the current knapsack with a fraction of the bundle and place the remaining fraction	in the next non-full knapsack using the same process.
	Finally, discard the fractionally cut bundles.
	
	Denote by~$J_L$ the set of items packed by~$\opt_1$ that are not in~\candmeps, i.e., the low-value items.
	Since the bundles consists of the densest low-value items, it must be the case, that~$\sum_{k=1}^{m} s(B_k) \leq s(J_L)$.
	Therefore,~$\opt_1^*$ is a feasible solution for~$I$ and the statement follows.
	
	By definition of the bundles, we have~$\sum_{k=1}^{m} v(B_k) = \vals \geq \frac{1}{1+\eps} v(J_L) \geq (1-\eps) v(J_L)$.
	Further, since there are~$\frac m \eps$ bundles of equal value and at most~$m$ of them are cut fractionally and discarded, we conclude that~$v(\opt_1^*) \geq (1-2\eps)\cdot v(\opt_1)$.
\end{proof}

\begin{lemma}
	\label{lem:FMDK:lastval}
	For the solution~$P_F$ of the algorithm, we have $v(P_F) \geq (1-6\eps) \cdot v(\opt)$.
\end{lemma}
\begin{proof}
	The solution~$P_\mathrm{EPTAS}$ returned by the EPTAS in Step~3 %
	has a value of at least~$(1-\eps)\cdot v(\opt_2)$. 
	The solution~$P_F$ is obtained from~$P_\mathrm{EPTAS}$ by replacing the placeholder bundles with the corresponding low-value items with the exception of fractionally cut ones, of which there are at most~$\frac{m}{\eps}$ many.
	Since there are~$\frac m {\eps^2}$ items in~$\opt$ that are of higher value than these items, namely the ones in $\optmeps$, this implies
	$$v(P_F) \geq v(P_\mathrm{EPTAS}) - \eps\cdot v(\opt).$$
	Using \Cref{lem:FMDK:candmeps,lem:FMDK:EPTAS}, we obtain:
	\begin{align*}
		v(P_F) &\geq v(P_\mathrm{EPTAS}) - \eps\cdot v(\opt)\\
		&\geq (1-2\eps)^2 \cdot v(\opt_1)-\eps \cdot v(\opt)\\
		&\geq (1-2\eps)^2 \cdot (1-\eps) \cdot v(\opt)-\eps \cdot v(\opt)\\
		&\geq ((1-4\eps) \cdot (1-\eps) -\eps) \cdot v(\opt)\\
		&\geq (1-6\eps) \cdot v(\opt),
	\end{align*}
	where the second to last equation follows from Bernoulli's inequality.
\end{proof}

\begin{lemma}
	\label{lem:FMDK:time}
	The algorithm has update time~$2^{\OO(\frac 1 \eps \cdot \log^4(\frac 1 \eps))} \cdot (\frac m \eps \log (n\vvmax))^{\OO(1)} + \OO(\frac 1 \eps \log \vmax \log n)$.
\end{lemma}

\begin{proof}	
	In the first step, guessing~${\lmax}$ and ${\lmin}$ leads to $\OO(\frac 1 {\eps^2}  \log^2
	\vvmax)$ many iterations. Guessing~$n_{\min}$ adds an additional factor of~$\frac m {\eps^2}$.
	In the second step, guessing~\vals leads to~$\OO(\frac 1\eps\log (n	\vvmax))$ many additional iterations, so the factor due to guessing is $\OO(\frac m {\eps^5} \log^2 \vvmax \log (n\vvmax))$ 
	
	Temporarily removing the~$n_{\min}\leq \frac m {\eps^2}$ elements from the data structure costs a total of~$\OO(\frac m {\eps^2} \log n)$, as does adding back removed items from a previous iteration.
	Computing the size of the bundles can be done by querying the	prefixes of value just above~$\vals$, so in time~$\OO(\log n)$. Computing the cut items of the bundles takes time $\frac{m}{\eps}\log n$.
	
	The set~\candmeps spans value classes ranging from values of~$\vc{\lmax}$ to a
	value at least~$\frac {\eps^3} m \cdot \vc{\lmax}$. 
	As the value classes correspond to powers of~$(1+\eps)$, this means we consider at most~$\log_{1+\eps} \frac m {\eps^3}$ many. Since each of them contains at most~$\smash{\frac{m}{\eps^2}}$ items,~\candmeps contains~$\OO(\frac {m^2} {\eps^6})$ items in total.
	Thus, in the third step, the EPTAS, used on $\OO(\frac {m^2} {\eps^6})$ many items, runs in time
	$2^{\OO(\frac 1 \eps \cdot \log^4(\frac 1 \eps))}+(\frac m \eps)^{\OO(1)}$.
	Together, this gives the first term in the desired update time.
	
	Recall, that we need to maintain one data structure for every existing and one for each possible value class, that is,~$\OO(\frac 1 \eps \log \vmax)$ many data structures in total.
	Maintaining these takes time~$\OO(\frac 1 \eps \log \vmax \log n)$.
\end{proof}

\paragraph*{Queries}
We show how to efficiently handle the different types of queries and state their running time.
\begin{compactitem}
	\item \textbf{Single Item Query:}
	If the queried item is contained in~\candmeps, its packing was saved explicitly.
	For low-value items, we save the first and last element entirely inside a bundle and on query of an item decide its membership in a bundle by comparing its density with those pivot elements.
	\item \textbf{Solution Value Query:} While the algorithm works with rounded values, we may set up the data structure of Section~\ref{sec:ds} to additionally store the actual values of items and enable prefix computation on the actual values.
	We can compute and store the actual solution value after an update by summing the actual values of packed candidates and determining the actual value of items in~$B$ using prefix computations while subtracting the values of discarded fractional bundles and items. On query, we return the stored solution value. 
	\item \textbf{Single Knapsack Query:}
	Output the saved packing of all candidates packed in the knapsack. Then, in the respective density sorted data structure, iterate over items in bundles that were packed in the queried knapsack and output them. As above, this is possible since the first and last item of a bundle were saved during the update step.
	\item \textbf{Query Entire Solution:}
	Output saved packing of all candidates and iterate over items in packed bundles in the respective density sorted data structure as above.
\end{compactitem}

\begin{lemma}\label{lem:few:query}
	The query times of our algorithm are as follows.
	\begin{compactenum}[(i)]
		\item Single item queries are answered in time~$\OO( \log \frac {m^2}{\eps^6})$. %
		\item Solution value queries are answered in time~$\OO(1)$.
		\item Queries of a single knapsack~$P_K$ are answered in time~$\OO(\abs{P_K})$.
		\item Queries of the entire solution~$P$ are answered in time~$\OO(\abs{P})$.
	\end{compactenum}
\end{lemma}
\begin{proof}
	(i): Since the packing of candidates is stored explicitly, each of the packed candidates can be output in time~$\OO(1)$. 
	The part of the solution corresponding to low-value items is stored implicitly, by saving the correct guesses and the first and last items of each bundle. The latter are stored in a tree sorted by density first and item index second, as in the data structure that was used to compute the bundles.
	Also save a pointer to and from the respective adjoining bundles of these items.
	This preparation is done during an update.
	When a low-value item is queried, use these pivot items to determine whether it is contained in packed bundles and if so in which it lies.
	This takes time~$\OO( \log \frac {m^2}{\eps^6})$.
	
	(ii): The computations for this query are done during an update of the instance, with the update clearly dominating the running time. Thus, on a query, the answer can be given in constant time.
	
	(iii): As in (i), the packing of candidates in~$P_K$ can be output in time~$\OO(1)$. 
	For low-value items, we create, during an update, pointers from bundles to the first, i.e., densest, item contained in them.
	On a query, we then simply consider each bundle in~$P_K$ and iterate over the density sorted data structure used to find and output all items of the bundle. 
	
	(iv): We use the approach from~(iii) on all knapsacks.
\end{proof}

\section{Proofs for Oblivious Linear Grouping (Section \ref{subsec:harmonicgrouping})}
\label{app:harmonicgrouping}

In this section, we give the technical details of the analysis of the  {oblivious} linear grouping approach developed in \cref{subsec:harmonicgrouping}. We start by analysing the approximation ratio, i.e., by formally proving \cref{lem:hg:apxratio}. Recall that~$\opt$ is the optimal solution and~$\opt_{\types}$ is the optimal solution attainable by packing item types~$\types$ instead of items in~$J'$ and using~$J \setminus J'$ without any changes.

\begin{restatable}[]{lemma}{hgApxRatio}\label{lem:hg:apxratio}
	Let~$\opt$ and~$\opt_{{\types}}$ be as defined above. Then,~$v(\opt_{\types}) \geq \frac{(1-\eps)(1-2\eps)}{(1+\eps)^2} v(\opt)$. 
\end{restatable}

The loss in the objective function due to rounding item values to natural powers of~$(1+\eps)$ is bounded by a factor of~$(1-\eps)$ by \cref{lem:Round}. As already pointed out, the analysis of the approximation ratio consists of three steps. In \cref{lem:hg:GuessLMax} we show that the loss in the objective function value when restricting the items in~$J'$ to the value classes with~$\llmin \leq \ell \leq \lmax$ is bounded by a factor of~$(1-\eps)$. 
If an optimal solution contains~$n_\ell$ items of~$V_\ell'$, it is feasible to pack the~$n_\ell$ \emph{smallest} such items. 
Then, \cref{lem:hg:GuessNL} shows that we do not need~$n_\ell$ exactly but it suffices to guess~$n_\ell$ up to a factor of~$(1+\eps)$. Finally, in \cref{lem:hg:harmonicGrouping}, we argue that using the introduced  {oblivious} linear grouping approach costs at most a factor~$(1-2\eps)$. In \cref{lem:hg:NoOfTypes}, we show that the number of item types within one value class is reduced to~$\OO(\frac{\log n'}{\eps^2})$. In \cref{lem:hg:RunningTime}, we bound the running time of the algorithm.

Let~$\mcp_{1}$ be the set of solutions that (i) may use all items in~$J''$ and (ii) uses items in~$J'$ only of the value classes~$V_\ell$ with~$\llmin \leq \ell \leq \lmax$. Let~$\opt_{1}$ be an optimal solution in~$\mcp_{1}$. The following lemma bounds the value of $\opt_{1}$ in terms of \opt.

\begin{lemma}\label{lem:hg:GuessLMax}
	Let $\opt_{1}$ be defined as above. Then,~$v(\opt_{1}) \geq (1-\eps) v(\opt)$. 
\end{lemma}

\begin{proof}
	Given~$\lmax$, it follows that $v(\opt) \geq (1+\eps)^{\lmax}$. As~$n'$ is an upper bound on the cardinality of $\opt'$, the items in the value classes with~$\ell  < \llmin$ contribute at most~$n'-1$ items to $\opt'$ while the value of one item is bounded by~$(1+\eps)^{\llmin}$. Thus, the total value of items in~$V_0, \ldots, V_{\llmin}$ contributing to~$\opt'$ is bounded by 
	\begin{equation*}
		n' (1+\eps)^{\llmin} =  n' (1+\eps)^{\lmax - \left\lceil \frac{\log n'/\eps}{\log(1+\eps)} \right\rceil } \leq \eps (1+\eps)^{\lmax} \leq \eps v(\opt).
	\end{equation*}

	Let~$J_{1}$ be the items in~$\opt'$ restricted to the value classes with~$\llmin \leq \ell \leq \lmax$. Clearly,~$J_{1}$ and~$\opt''$ can be feasibly packed. Hence,
	\[
	v(\opt_{1}) \geq v(J_{1}) + v(\opt'') \geq v(\opt') - \eps v(\opt) + v(\opt'') \geq (1-\eps) v(\opt). 
	\]
\end{proof}

From now on, we only consider packings in~$\mcp_1$, i.e., we restrict to the value classes~$V_\ell$ with $\llmin \leq \ell \leq \lmax$ for the items in~$J'$. Let~$V_\ell$ be a value class contributing to $\opt_{1}'$. As explained above, knowing~$n_\ell = |V_\ell \cap \opt_{1}'|$ would be sufficient to determine the items of~$V_\ell'$ contributing to $\opt_{1}$, i.e., to determine~$V_\ell' \cap \opt_{1}$. In the following lemma we show that we can additionally assume that~$n_\ell = (1+\eps)^{k_\ell}$ for some~$k_\ell \in \N_0$. To this end, let~$\mcp_2$ contain all the packings in~$\mcp_1$ where the number of big items of each value class~$V_\ell$ is a natural power of~$(1+\eps)$. Let~$\opt_2$ be an optimal packing in~$\mcp_2$. 

\begin{lemma}\label{lem:hg:GuessNL}
	$v(\opt_{2}) \geq \frac1{(1+\eps)} v(\opt_{1})$. 
\end{lemma}

\begin{proof}
	Consider~$\opt_1$, the optimal packing in~$\mcp_1$. 
	We set $\opt_{1}' := \opt_1 \cap J'$ and~$\opt_{1}'' := \opt_1 \setminus J' = \opt_1 \cap J''$.
	We construct a feasible packing in~$\mcp_2$ that achieves the desired value of $\frac1{(1+\eps)} v(\opt_{1})$.
	
	Let~$J_{2}$ be the subset of~$\opt_{1}'$ where each value class~$V_\ell'$ is restricted to the smallest~$(1+\eps)^{\lfloor \log_{1+\eps} n_\ell\rfloor}$ items in~$V_\ell'$ if~$V_\ell \cap \opt_{1}' \neq \emptyset$. 
	
	Fix one value class~$V_\ell$ with~$V_\ell \cap \opt_{1}' \neq \emptyset$. Restricting to the first~$(1+\eps)^{\lfloor \log_{1+\eps} n_\ell\rfloor}$ items in~$V_\ell \cap \opt_{1}'$ implies  
	\begin{equation*}		
		v(V_\ell \cap J_{2}) 	\geq (1+\eps)^{\lfloor \log_{1+\eps} n_\ell\rfloor} (1+\eps)^\ell \geq \frac{1}{1+\eps} (1+\eps)^{\log_{1+\eps} n_\ell} (1+\eps)^\ell = \frac{1}{1+\eps} v(V_\ell \cap \opt_{1}').
	\end{equation*}
	Clearly,~$J_{2}\cup\opt_{1}''$ is a feasible packing in~$\mcp_2$. Since~$v(\opt_{1}') = \sum_{\ell = \llmin}^{\lmax} v(V_\ell \cap \opt_{1}') $,
	\[
	v(\opt_{2}) \geq v(J_2) + v (\opt_{1}'') \geq \frac{1}{1+\eps} v(\opt_{1}') + v(\opt_{1}'') \geq \frac1{1+\eps} v(\opt_1). 
	\]
\end{proof}

From now on, we only consider packings in~$\mcp_2$. This means, we restrict the items in $J'$ to value classes~$V_\ell'$ with~$\llmin \leq \ell \leq \lmax$ and assume that~$n_\ell = (1+\eps)^{k_\ell}$ for~$n_\ell \in \N_0$ or~$n_\ell = 0$. Even with $n_\ell$ being of the form $(1+\eps)^{k_\ell}$, guessing the exponent for each value class $V_\ell'$ independently is intractable in time polynomial in~$\log n$ and~$\frac1\eps$. To resolve this, the  {oblivious} linear grouping creates groups that take into account all possible guesses of $n_\ell$. This rounding is done for each value class individually and results in item types~$\types_\ell$ for the set~$V_\ell'$. Let~$\mcp_{\types}$ be the set of all feasible packings of items in~$\types_\ell$ for $\llmin \leq \ell \leq \lmax$ and any subset of items in~$J''$. 
That is, instead of the original items in $J'$ the packings in $\mcp_{\types}$ pack the corresponding item types.
Note that packings in~$\mcp_{\types}$ are not forced to pack natural powers of~$(1+\eps)$ many items per value class. Let~$\opt_{\types}$ be the optimal solution in~$\mcp_{\types}$. 
The next lemma shows that~$v(\opt_{\types})$ is at most a factor~$(1-2\eps)$ less than $v(\opt_{2})$, the optimal solution in~$\mcp_2$. 

\begin{lemma}\label{lem:hg:harmonicGrouping}
	$v(\opt_{{\types}}) \geq (1-2\eps) v(\opt_{2})$. 
\end{lemma}

\begin{proof}
	We construct a feasible packing~$J_3$ in~$\mcp_{\types}$ based on the optimal packing~$\opt_2$. Let~$\opt_2':= J'\cap \opt_2$ and~$\opt_2'' := J'' \cap \opt$. We let~$J_3 '' :=\opt_2''$ be the items of~$J''$ in our new packing~$J_3$. These items will be packed exactly where they are packed in~$\opt_2$. For items in~$J'$, we consider each value class~$V_\ell\cap \opt_2'$ individually and carefully construct the set~$J_{3,l}$, the items of~$V_\ell'$ contributing to~$J_3$. Then, we show that the items in~$J_{\ell,3}$ can be packed into the knapsacks where the items in~$V_\ell\cap \opt_2'$ are placed while ensuring that~$v(J_{\ell,3}) \geq (1-2\eps) v(V_\ell \cap \opt_2')$. 
	
	If~$V_\ell \cap \opt_{2}' = \emptyset$, we set~$J_{\ell,3} = \emptyset$. Then, both requirements are trivially satisfied. 	
	
	If~$|V_\ell' \cap \opt_2'| \leq \frac1\eps$. Then, we set~$J_{\ell,3} := V_\ell' \cap \opt_2'$. Clearly,~$v(J_{\ell,3}) \geq (1-2\eps) v(V_\ell \cap \opt_2')$. 
	For packing~$J_{\ell,3}$, we observe that~$\types_\ell$ actually contains the smallest~$\frac1\eps$ items as item types.
	Hence, their sizes are not affected by the rounding procedure and whenever $\opt_{2}$ packs one of these items, we can pack the same item into the same knapsack. 
	
	Let~$\ell$ be a value class with~$n_\ell := |V_\ell' \cap \opt_2'| > \frac1\eps$. Let~$G_1(n_\ell),\ldots,G_{1/\eps}(n_\ell)$ be the corresponding~$\frac1\eps$ groups of~$\lfloor \eps n_\ell \rfloor$ or~$\lceil \eps n_\ell \rceil$ many items created by the (traditional) linear grouping for~$n_\ell$. We set~$J_{\ell,3} = G_1(n_\ell) \cup \ldots \cup G_{1/\eps -1}(n_\ell)$. As \(v(G_{1/\eps}(n_\ell)) = \lceil \eps n_\ell \rceil (1+\eps)^\ell \leq 2 \eps n_\ell (1+\eps)^\ell = 2\eps v_{\ell,2} \), we have~$v(J_{\ell,3}) \geq (1-2\eps) v(V_\ell'\cap \opt_2')$. For packing these items, we observe that the item types created by our algorithm are a refinement of~$G_1(n_\ell),\ldots,G_{1/\eps}(n_\ell)$. As the  {oblivious} linear grouping ensures~$|G_{1/\eps}(n_\ell)| \geq |G_{1/\eps -1 }(n_\ell)| \geq \ldots \geq |G_1(n_\ell)|$ and that the item sizes are increasing in the group index, we can pack the items of group~$G_k(n_\ell)$ where~$\opt_{2}$ packs the items of group~$G_{k+1}(n_\ell)$ for~$1 \leq k < \frac1\eps$. 	
	We conclude 
	\[
	v(\opt_{{\types}}) \geq v(J_3) + v(\opt_{2}'') \geq (1-2\eps) v(\opt_{2}') + v(\opt_{2}'') \geq (1-2\eps) v(\opt_2). 	
	\]	
\end{proof}

As~$\types$ contains at most~$\frac{1}{\eps} \big(\big \lceil \frac{\log n'/\eps}{\log (1+\eps)} \big\rceil +1 \big)$ many different value classes and using~$\big \lceil \frac{\log n'}{\log (1+\eps)}\big\rceil +1$ many different values for~$n_\ell = |\opt \cap V_\ell'|$ suffices as explained above, the next lemma follows.

\begin{lemma}\label{lem:hg:NoOfTypes}
	The algorithm reduces the number of item types to~$\OO(\frac{\log^2 n'}{\eps^4})$.
\end{lemma}

Next, we formally prove the bound on the running time, i.e., the following lemma.

\begin{restatable}[]{lemma}{hgRunningTime}\label{lem:hg:RunningTime}
	For a given guess~$\lmax$, the set~$\types$ can be determined in time~$\OO(\frac{\log^4 n'}{\eps^4})$.
\end{restatable}

\begin{proof}
	Remember that~$n'$ is an upper bound on the number of items in~$J'$ in any feasible solution. Observe that the boundaries of the linear grouping created by the algorithm per value class are actually independent of the value class and only refer to some~$k$th item in class~$V_\ell$. Hence, the algorithm first computes the different indices needed in this round. We denote the set of these indices by~$I' = \{ j_1, \ldots \}$ sorted in an increasing manner. There are at most~$\lfloor \log_{1+\eps} n' \rfloor$ many possibilities for~$n_\ell$. Thus, the algorithm needs to compute at most~$\frac1\eps (\log_{1+\eps} n' + 1)$ many different indices. This means that these indices can be computed and stored in time~$\OO(\frac{\log n'}{\eps^2} )$ while each index is bounded by~$n$. 
	
	Given the guess~$\lmax$ and~$\llmin$, fix a value class~$V_\ell$ with~$\llmin \leq l \leq \lmax$. We want to bound the time the algorithm needs to transform the big items in~$V_\ell$ into the modified item set~$T_\ell$. We will ensure that the dynamic algorithms in the following sections maintain a balanced binary search for each value class~$V_\ell$ that stores the items in~$J'$ sorted by increasing size. Hence, the sizes of the items corresponding to~$J'$ can be accessed in time~$\OO(\frac{\log^3 n'}{\eps^2})$. These sizes correspond to the item size~$s_t$ for an item type~$t \in \types_\ell$. 
	Given an item type~$t \in \types_\ell$,~$n_t = j_t - j_{t-1}$, which can again be pre-computed independently of the value class. Thus,~$\types_\ell$ can be computed in time~$\OO(\frac{\log^3 n'}{\eps^2})$. 
	
	As there are~$\OO (\frac{\log n'}{\eps^2})$ many value classes that need to be considered for a given guess~$\lmax$, calculating the set~$\types^{(\lmax)}$ needs~$\OO(\frac{\log^4 n'}{\eps^4})$ many computation steps.	
\end{proof}

\begin{proof}[Proof of \cref{theo:harmonic}]
	\cref{lem:Round,lem:hg:apxratio} %
	bound the approximation ratio, \cref{lem:hg:NoOfTypes} bounds the number of item types, and \cref{lem:hg:RunningTime} bounds the running time of %
	{oblivious} linear grouping. %
\end{proof}

\section{Proofs for Identical Knapsacks (Section \ref{sec:mik})}
\label{app:multiknapsack}

In this section, we give the technical details of \cref{sec:mik}. The first step is to analyze the loss in the objective function value due to the linear grouping. Set~$J' = J_B$ and~$n'= \min\{ \frac m \eps, n\}$. 
Moreover, let~$\opt_{\types}$ be the optimal packing when using the corresponding item types $\types$ obtained from applying  {oblivious} linear rounding instead of the items in~$J_B$. 
Then, the next corollary immediately follows from  \cref{theo:harmonic}.

\begin{corollary}\label{cor:mik:harmonic}
	Let~$\opt$ and~$\opt_{\types}$ be defined as above. Then,~$v(\opt_{\types}) \geq \frac{(1-\eps)(1-2\eps)}{(1+\eps)^2}v(\opt)$. 
\end{corollary}

In Figure~\ref{fig:mik:solution} we give a feasible solution to the configuration ILP. 
In the next lemma, we show that there is a guess~$v_S$ with the corresponding size~$s_S$ such that~$v_{\text{ILP}}^*+v_S+v_{j^\star}$ for the optimal solution value~$v_{\text{ILP}}^*$ of \eqref{eq:mik:ilp} is a good guess for the optimal solution~$v(\opt_{\types})$. Here,~$j^\star$ is the densest small item not contained in~$P$ while~$P$ is the maximal prefix of small items with~$v(P) < v_S$. The high-level idea of the proof is to restrict an optimal solution~$\opt_{\types}$ to the~$\lfloor (1-3\eps)m \rfloor $ most valuable knapsacks and show that~$s_S$ underestimates the size of small items in these~$\lfloor (1-3\eps)m \rfloor$ knapsacks. Interpreting these knapsacks as configurations gives a feasible solution for the configuration ILP. 

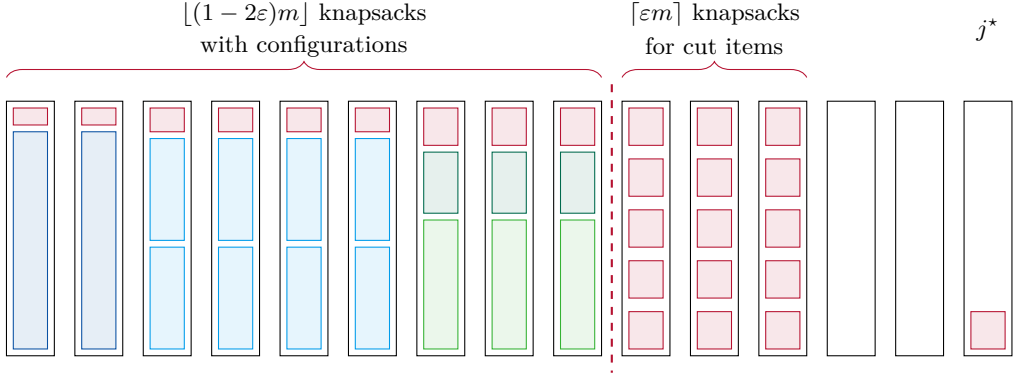
\begin{figure}[tbh]
	\centering 
	\begin{tikzpicture}[scale=.9]
		\foreach \x in {0,...,14}{
			\draw[draw=black, fill=none] (\x-.1,0) rectangle (\x+.6, 3.75); 
		}

		\draw[ubred, dashed, thick] (8.75,4) -- (8.75, -.25);
		\foreach \x in {9,...,11}{
			\foreach \y in {0,...,4}
			\draw[draw=ubred, fill = ubred!10] (\x, 0.1+\y*.75) rectangle (\x+.5, 0.65 + \y*.75);
		}
		\draw[draw=ubred, fill = ubred!10] (14, 0.1) rectangle (14.5, 0.65);

		\foreach \x in {0,1}{
			\draw[draw=utnavy, fill = utnavy!10] (\x, 0.1) rectangle (\x+.5, 3.3);
			\draw[draw=ubred, fill = ubred!10] (\x, 3.4) rectangle (\x+.5, 3.65);
		}
		\foreach \x in {2,...,5}{
			\draw[draw=utblue, fill = utblue!10] (\x, 0.1) rectangle (\x+.5, 1.6);
			\draw[draw=utblue, fill = utblue!10] (\x, 1.7) rectangle (\x+.5, 3.2);
			\draw[draw=ubred, fill = ubred!10] (\x, 3.3) rectangle (\x+.5, 3.65);
		}
		\foreach \x in {6,...,8}{
			\draw[draw=utgreen, fill = utgreen!10] (\x, 0.1) rectangle (\x+.5, 2);
			\draw[draw=utforest, fill = utforest!10] (\x, 2.1) rectangle (\x+.5, 3);
			\draw[draw=ubred, fill = ubred!10] (\x, 3.1) rectangle (\x+.5, 3.65);
		}
		
		\draw [decorate,decoration={brace,amplitude=6pt},ubred]  (8.9, 4.1) -- (11.6, 4.1) node [black,midway,yshift =  24pt,font=\small, black] {$\lceil \eps m \rceil$  knapsacks} ;
		\draw [decorate,decoration={brace,amplitude=6pt},ubred]  (-.1, 4.1) -- (8.6, 4.1) node [black,midway,yshift =  24pt,font=\small, black] {$\lfloor (1-2\eps) m \rfloor$ knapsacks};
		\node[yshift = 12pt, font = \small] at (4.25,4.1) {with configurations}; 
		\node[yshift = 18pt, font = \small] at (14.25, 4.1) {$j^\star$};
		\node[yshift = 12pt, font = \small] at (10.25,4.1) {for cut items};
	\end{tikzpicture}
	\caption{A possible solution: Blue and green rectangles represent the packed big item types. Red rectangles on the left side represent the space left empty by the configurations and on the right represent the slots for cut small items reserved in~$\lceil \eps m \rceil$ knapsacks while the densest small item, $j^\star$, not in~$P$ gets its own knapsack.} \label{fig:mik:solution}
\end{figure}

\begin{restatable}{lemma}{mikOPTofConfigILP}\label{lem:mik:OPTofConfigILP}
	There is a guess~$v_S$ with \(v_{\text{ILP}} + v_S \geq \frac{1-4\eps}{1+\eps} v(\opt_{{\types}}).\) Moreover,~$v(P_S) + v_{j^\star} \geq v_S$. 
\end{restatable}

\begin{proof}
	Let $\opt_{B,\types} := \opt_{\types} \cap J_{B}$ and $\opt_{S,\types} := \opt_{\types} \cap J_{S}$.
	We construct a candidate set~$J_{\text{ILP}}$ of items that are feasible for \eqref{eq:mik:ilp} and obtain a value of at least~$(1-4\eps) v(\opt_{B,{\types}})$. To this end, take an optimal packing~$\opt_{\types}$ and consider the~$\lfloor(1-3\eps)m\rfloor$ most valuable knapsacks in this packing. Let~$J_{B,{\types}}$ and~$J_{S,{\types}}$ consist of the big and small items, respectively, in these knapsacks. Since~$m \geq \frac{16}{\eps^7}\log^2 n \geq \frac1\eps$, we have~$\lfloor(1-3\eps)m\rfloor \geq (1-4\eps)m$. Hence, $$v(J_{B,{\types}}) + v(J_{S,{\types}}) \geq (1-4\eps) v(\opt_{\types})\,.$$ 
	
	Create the variable values~$y_c$ corresponding to the number of times configuration~$c$ is used by the items in~$J_{B,{\types}}$. We observe that~$J_{B,{\types}} \cup J_{S,{\types}}$ can be feasibly packed into~$\lfloor(1-3\eps)m\rfloor$ knapsacks. Therefore, 
	\(
	\sum_{c \in \configs} y_c \leq \lfloor (1-3\eps) m \rfloor \, ,
	\)
	and 
	\(
	\sum_{c \in \configs} y_c s_c + s(J_{S,\types}) \leq \lfloor (1-3\eps) m \rfloor \capa \, .
	\)
	
	Since we guess the value of the small items in the dynamic algorithm up to a factor of~$(1+\eps)$, there is one guess~$v_S$ satisfying~$v_S \leq v(J_{S,{\types}}) < (1+\eps) v_S$. 		
	Let $P_S$ be the maximal prefix of small items with~$v(P_S) < v_S$ and let $j^\star$ be the densest small item not in~$P_S$. Hence, \(v(P_S) + v_{j^\star} \geq v_S \geq \frac{1}{1+\eps} v(J_{S,{\types}}).\)
	As~$P_S$ contains the densest small items, this implies~$s_S : = s(P_S) \leq s(J_{S,{\types}})$. 
	Thus, 
	\[
	\sum_{c \in \configs} y_c s_c \leq \lfloor (1-3\eps)m \rfloor \capa - s(J_{S,{\types}}) \leq \lfloor (1-3\eps)m\rfloor \capa - s_S. 
	\]		
	Therefore, the just created~$y_c$ are feasible for the ILP with the guess~$v_S$, and 
	\begin{equation*}
		v_{\text{ILP}} + v_S \geq v(J_{B,{\types}}) + \frac{1}{1+\eps}  v(J_{S,{\types}})  \geq  \frac{1}{1+\eps} (v(J_{B,{\types}}) + v(J_{S,{\types}})  )   \geq \frac{1-4\eps}{1+\eps} v(\opt_{\types}),
	\end{equation*}
	which concludes the proof.
\end{proof}

\paragraph*{Solving the configuration LP} In this part, we provide the full proof of our approach to solving the LP relaxation of the configuration ILP when~$m$ satisfies~~$\frac{16}{\eps^7} \log^2n \leq m$. 
More specifically, we prove the following lemma. 
\begin{restatable}{lemma}{mikSolveConfigLP}\label{lem:mik:SolveConfigLP}
	Let~$U = \max\{ \capa m, n \vvmax\}$. %
	There is an algorithm that finds a feasible solution for the LP relaxation of \eqref{eq:mik:ilp} with value at least~$\frac{1-\eps}{1+\eps}v_{\text{LP}}$ with running time~$\left(\frac{\log U}{\eps}\right)^{\smash{\OO(1)}}$.
\end{restatable}

\noindent In the following, we abuse notation and also refer to the LP relaxation of \eqref{eq:mik:ilp} by \eqref{eq:mik:ilp}: 

\begin{equation}\tag{P}\label{eq:mik:lp}
	\begin{array}{llcll}
		\max & \displaystyle{\sum_{c \in \configs} y_c v_c }\\
		\text{subject to } & \displaystyle{\sum_{c \in \configs} y_c s_c} & \leq & \lfloor (1-3\eps) m \rfloor \capa- s_S \\
		& \displaystyle{\sum_{c \in \configs} y_c} & \leq & \lfloor (1-3\eps) m \rfloor\\
		& \displaystyle{\sum_{c \in \configs} y_c n_{tc}} & \leq & n_t & \text{for all } t \in \types^{(l_{\max})} \\
		& y_c & \in & {\geq 0} & \text{for all } c \in \configs
	\end{array}
\end{equation}

Let~$\gamma$ and~$\beta$ be the dual variables of the capacity constraint and the number of knapsacks constraint, respectively. We set~$\types:=\types^{(\lmax)}$ for simplicity. Let~$\alpha_t$ for~$t \in \types$ be the dual variables of the constraint ensuring that only~$n_t$ items of type~$t$ are packed. Then, the dual is given by the following linear program. 
\begin{equation}\tag{D}\label{eq:mik:dual:app}
	\begin{array}{llcll}
		\min & \displaystyle{\lfloor(1-3\eps)m\rfloor\beta + (\lfloor (1-3\eps)m\rfloor \capa - s_S)\gamma + \sum_{t \in \types} n_t \alpha_t}\\
		\text{subject to } & \beta + s_c \gamma + \displaystyle{\sum_{t \in T} \alpha_t n_{tc}} & \geq & v_c & \text{for all } c \in \configs \\
		& \alpha_t & \geq & 0 & \text{for all } t \in \types \\
		& \beta, \gamma & \geq & 0 .
	\end{array}
\end{equation} 

As discussed above, for applying the Ellipsoid method we need to solve the separation problem efficiently. The separation problem decides if the current solution~$(\alpha^*, \beta^*, \gamma^*)$ is feasible or finds a violated constraint. As verifying the first constraint of~\eqref{eq:mik:dual:app} corresponds to solving a \knapsack problem, we do not expect to optimally solve the separation problem in time polynomial in~$\log n$ and~$\frac{1}{\eps}$. Instead, we apply a dynamic program (DP) for the single knapsack problem after restricting the item set further and rounding the item values as follows.

Let~$\bar{v}_t := v_t - \alpha_t^* - \gamma^* s_t$ for~$t \in \types$. If there exists an item type with~$\bar{v}_t > \beta^*$, we return the configuration using only this item. Otherwise, we define~$\tilde{v}_t := \big\lfloor \frac{\bar{v}_t}{\eps^4 \beta^*} \big\rfloor \cdot \eps^4 \beta^*$. 
By running the dynamic program for the \textsc{Knapsack} problem on the item set~$\types$ with multiplicities~$\min\{\frac1\eps, n_t\}$ and values~$\tilde{v}_t$, we obtain a solution~$x^*$ where~$x_t^*$ indicates how often item type~$t$ is packed. If~$\sum_{t \in \types} x^*_t \tilde{v}_t > \beta^*$, we return the configuration defined by~$x^*$ as separating hyperplane. Otherwise, we return \textsc{declared feasible} for the current solution. 

The next lemma shows that this algorithm approximately solves the separation problem by either correctly declaring infeasibility or by finding a solution that is almost feasible for~\eqref{eq:mik:dual:app}. The slight infeasibility for the dual problem translates to a slight decrease in the optimal objective function value of the primal problem. In the proof we use that~$x^*$ is optimal for the rounded values~$\tilde{v}_t$ to show that~$(\alpha^*, \beta^*, \gamma^*)$ is almost feasible if~$\sum_{t \in \types} x^*_t \tilde{v}_t \leq \beta^*$. Noticing that~$\bar{v}_t \geq \tilde{v}_t$ then concludes the proof. 

\begin{restatable}{lemma}{mikSeparation}\label{lem:mik:separation}
	Given~$(\alpha^*, \beta^*, \gamma^*)$, there is an algorithm with running time~$\OO\left( \frac{\log^2 n}{\eps^{10}}\right)$  which either finds a configuration~$c \in \configs$ such that $\beta^* + s_c \gamma^* + \sum_{t \in \types} \alpha_t^* n_{tc} < v_c$ or guarantees that $\beta^* + s_c \gamma^* + \sum_{t \in \types} \alpha_t^* n_{tc} \geq (1-\eps)v_c$ holds for all~$c \in \configs$. 
\end{restatable}

\begin{proof}
	Fix a configuration~$c$ and recall that~$s_c = \sum_{t \in \types} n_{tc} s_t$ and~$v_c = \sum_{t \in \types} n_{tc} v_t$. Then, checking~$\beta^* + s_c \gamma^* + \sum_{t \in \types} \alpha_t^* n_{tc} \geq v_c$ for all configurations~$c \in \configs$ is equivalent to showing $\max_{c \in \configs} \sum_{t \in \types} (v_t - \alpha_t^* - \gamma^* s_t)n_{tc} \leq \beta^*$. This problem translates to solving the following ILP and comparing its objective function value to~$\beta^*$. 
	\begin{equation*}\tag{S}\label{eq:mik:separation}
		\begin{array}{llcll}
			\max & \displaystyle{\sum_{t \in \types} (v_t - \alpha_t^* - \gamma^* s_t) x_t }\\
			\text{s.t.} & \displaystyle{\sum_{t \in \types} s_t x_t} & \leq & \capa \\
			& x_t & \leq & n_t & \text{for all } t \in \types \\
			& x_t & \in & \mathbb{Z}_{\geq 0} 
		\end{array}
	\end{equation*}
	This ILP is itself a (single) \knapsack problem. Hence, the solution~$x^*$ found by the algorithm is indeed feasible for~\eqref{eq:mik:separation}. 	
	
	We start by bounding the running time of the algorithm. Recall that, for each $t \in \types$, $\bar{v}_t := v_t - \alpha_t^* - \gamma^* s_t$ and $\tilde{v}_t := \big\lfloor \frac{\bar{v}_t}{\eps^4 \beta^*} \big\rfloor \cdot \eps^4 \beta^*$. Observe that~$\types$ only contains big items. Hence, it suffices to consider~$\min \{n_t,\frac{1}{\eps}\}$ items per value class in the DP. It can be checked in time~$\OO(\frac{\log^2 n}{\eps^4})$, if~$\bar{v}_t \leq \beta^*$ is violated for one~$t \in \types$. Otherwise,~$\tilde{v}_t \leq \bar{v}_t$ and~$\bar{v}_t - \tilde{v}_t \leq \eps^4 \beta^*$ hold. Thus, the running time of the DP is bounded by~$\OO\left(\frac{|\types|^2}{\eps^6}\right) = \OO\left( \frac{\log^2 n}{\eps^{10}}\right)$ \cite{IbarraK75}. %
	
	It remains to show that the solution~$x^*$ either defines a configuration with~$\beta^* + s_c \gamma^* + \sum_{t \in \types} \alpha_t^* n_{tc} < v_c$ or ensures that $\beta^* + s_c \gamma^* + \sum_{t \in \types} \alpha_t^* n_{tc} \geq (1-\eps)v_c$ holds for all~$c \in \configs$. If~$\sum_{t \in \types} x^*_t \tilde{v}_t > \beta^*$, it holds
	\(
	\sum_{t \in \types} x_t^* \bar{v}_t \geq \sum_{t \in \types} x^*_t \tilde{v}_t > \beta^*
	\) 	
	and, thus,~$x^*$ defines a separating hyperplane. 
	
	Suppose now that $\sum_{t \in \types} x^*_t \tilde{v}_t \leq \beta^*$. We assume for the sake of contradiction that there is a configuration ~$c'$, defined by packing~$x_t$ items of type~$t$, such that 
	\(
	\sum_{t \in \types} x_t\big((1-\eps)v_t - \alpha_t^* - \gamma^* s_t \big) > \beta^*.
	\)
	As~$\types$ contains only big item types, we have that~$\sum_{t \in \types} x_t \leq \frac{1}{\eps}$. This implies that there exists at least one item type~$t'$ in~$\types$ with~$x_{t'}\geq 1$ and $(1-\eps)v_{t'} - \alpha_{t'}^* - \gamma^* s_{t'} \geq \eps \beta^*$. 
	Moreover, 
	\begin{equation*}
		\bar{v}_{t} = v_{t} - \alpha_{t}^* - \gamma^* s_{t} \geq (1-\eps)v_{t} - \alpha_{t}^* - \gamma^* s_{t} 
	\end{equation*} 
	holds for all item types~$t \in \types$. This implies for~$t'$ that~$\bar{v}_{t'} \geq \eps \beta^*$. 
	Hence, 
	\[
	\sum_{t \in \types} x_t \bar{v}_t \geq \eps x_{t'} \bar{v}_{t'} +  \sum_{t \in \types} x_t \big((1-\eps) v_t -	 \alpha_t^* - \gamma^* s_t\big) > \eps v_{t'} + \beta^*  \geq (1+\eps^2) \beta^*. 
	\] 
	By definition of~$\tilde{v}$, we have~$\bar{v}_t - \tilde{v}_T \leq \eps^4 \beta^*$ and~$\sum_{t\in \types} x_t (\bar{v}_t - \tilde{v}_t)  \leq \eps^3\beta^*$. This implies 
	\[
	\sum_{t \in \types} x_t \tilde{v}_t = \sum_{t \in \types} x_t \tilde{v}_t   + \sum_{t\in\types} x_t (\bar{v}_t - \tilde{v}_t)  > (1+\eps^2)\beta^* - \eps^3 \beta^* \geq \beta^*,
	\] 
	where the last inequality follows from~$\eps \leq 1$. By construction of the DP,~$x^*$ is the optimal solution for the values~$\tilde{v}$ and achieves a total value less than or equal to~$\beta^*$. Hence,
	\[
	\beta^* \geq  \sum_{t \in \types} x_t^*  \tilde{v}_t \geq \sum_{t \in \types} x_t \tilde{v}_t > \beta^*;
	\]
	a contradiction. 	
\end{proof}

We now present the proof of \Cref{lem:mik:SolveConfigLP}. 

\begin{proof}[Proof of \cref{lem:mik:SolveConfigLP}]
	As discussed above, the high-level idea is to solve~\eqref{eq:mik:dual:app}, the dual of~\eqref{eq:mik:ilp}, with the Ellipsoid Method and to consider only the variables corresponding to constraints added by the Ellipsoid Method for solving~\eqref{eq:mik:ilp}. 
	
	As~\eqref{eq:mik:separation} is part of the separation problem for~\eqref{eq:mik:dual:app}, there is no efficient way to exactly solve the separation problem, unless~$\classP=\NP$. \cref{lem:mik:separation} provides us a way to approximately solve the separation problem. As an approximately feasible solution for~\eqref{eq:mik:dual:app} cannot be directly used to determine the important variables in~\eqref{eq:mik:ilp}, we add an upper bound~$r$ on the objective function as a constraint to~\eqref{eq:mik:dual:app} and search for the largest~$r$ such that the Ellipsoid Method returns infeasible. This implies that~$r$ is an \emph{upper bound} on the objective function of~\eqref{eq:mik:dual:app} which in turn guarantees a \emph{lower bound} on the objective function value of~\eqref{eq:mik:ilp} by weak duality. 
	
	Of course, testing all possible values for~$r$ is intractable and we restrict the possible choices for~$r$. Observe that~$v_{\text{LP}} \in [v_{\max}, nv_{\max}]$ where~$v_{\lp}$ is the optimal value of~\eqref{eq:mik:ilp}. Thus, for~$k \in \N$ with~$\lceil \log_{1+\eps} v_{\max} \rceil \leq k \leq \lceil \log_{1+\eps} (nv_{\max}) \rceil$, we use~$r = (1+\eps)^k$ as upper bound on the objective function. That is, we test if \eqref{eq:mik:dual:app} extended by the objective function constraint \[\lfloor(1-3\eps)m\rfloor\beta + (\lfloor(1-3\eps)m\rfloor \capa- s_S)\gamma + \sum_{t \in \types} n_t \alpha_t \leq r\] is declared feasible by the Ellipsoid Method with the approximate separation oracle for~\eqref{eq:mik:separation}. We refer to the feasibility problem by~(D$_r$). 
	
	For a given solution~$(\alpha^*,\beta^*, \gamma^*)$ of~(D$_r$), the separation problem asks for one of the two: either the affirmation that the point is feasible or a separating hyperplane that separates the point from any feasible solution. The non-negativity of~$\alpha_t^*, \beta^*$, and~$\gamma^*$ an be checked in time~$\OO(|\types|) = \OO\big(\frac{\log^2 n}{\eps^4}\big)$. In case of a negative answer, the corresponding non-negativity constraint is a feasible separating hyperplane. Similarly, in time~$\OO(|\types|)$, we can check whether the objective function constraint $\lfloor(1-3\eps)m\rfloor\beta + (\lfloor(1-3\eps)m\rfloor\capa - s_S)\gamma + \sum_{t \in \types} n_t \alpha_t \leq r$ is violated and add it as a new inequality if necessary.  
	In case the non-negativity and objective function constraints are not violated, the separation problem is given by the knapsack problem in~\eqref{eq:mik:separation}. 
	The algorithm in \cref{lem:mik:separation} either outputs a configuration that yields a valid separating hyperplane or declares~$(\alpha^*,\beta^*, \gamma^*)$ feasible, i.e.,~$\beta^* + s_c \gamma^* + \sum_{t \in \types} \alpha_t^* n_{tc} \geq (1-\eps)v_c$ for all~$c \in \configs$. This implies that~$(\alpha^*,\beta^*, \gamma^*)$ is feasible for the following LP. 
	Note that we changed the right side of the constraints when compared to~\eqref{eq:mik:dual:app}.	
	\begin{equation}\tag{D$^{(1-\eps)}$}\label{eq:mik:dual:appeps}
		\begin{array}{llcll}
			\min & \displaystyle{\lfloor(1-3\eps)m\rfloor\beta + (\lfloor(1-3\eps)m\rfloor\capa - s_S)\gamma + \sum_{t \in\types} n_t \alpha_t}\\
			\text{s.t. } & \beta + s_c \gamma + \displaystyle{\sum_{t \in \types} \alpha_t n_{tc}} & \geq & (1-\eps) v_c &  \text{for all } c \in \configs \\
			& \alpha_t & \geq & 0 & \text{for all } t \in \types \\
			& \beta, \gamma & \geq & 0 
		\end{array}
	\end{equation}	
	
	Let~$r^*$ be minimal such that~(D$_{r^*}$) is declared feasible. %
	Let~$v_D^{(1-\eps)}$ denote the optimal solution value of~\eqref{eq:mik:dual:appeps}. As~$(\alpha^*,\beta^*, \gamma^*)$ is feasible with objective value at most $r^*$, we have~$v_D^{(1-\eps)} \leq r^* $.
	Let~$v^{(1-\eps)}$ denote the optimal solution value of its dual, i.e., of the following~LP. %

	\begin{equation}\tag{P$^{(1-\eps)}$}\label{eq:mik:lpeps}
		\begin{array}{llcll}
			\max & \displaystyle{\sum_{c \in \configs} y_c (1-\eps) v_c }\\
			\text{subject to } & \displaystyle{\sum_{c \in \configs} y_c s_c} & \leq & \lfloor(1-3\eps)m\rfloor \capa- s_S \\
			& \displaystyle{\sum_{c \in \configs} y_c} & \leq & \lfloor(1-3\eps)m\rfloor \\
			& \displaystyle{\sum_{c \in \configs} y_c n_{tc}} & \leq & n_t & \text{for all } t \in \types \\
			& y_c & \geq & 0 & \text{for all } c \in \configs
		\end{array}
	\end{equation}
	Then,~$y=0$ is feasible for~\eqref{eq:mik:lpeps}, and by weak duality, we have
	\begin{equation*}%
		v^{(1-\eps)} \leq v_D^{(1-\eps)} \leq r^*.
	\end{equation*}
	Note that~\eqref{eq:mik:ilp} and~\eqref{eq:mik:lpeps} have the same feasible region and their objective functions only differ by the factor~$(1-\eps)$. This implies that	
	\begin{equation}\label{eq:mik:v*<=r*}
		v_{\text{LP}} = \frac{v^{(1-\eps)}}{1-\eps} \leq \frac{r^*}{1-\eps}.
	\end{equation}
	Because of this relation between~$v_{\text{LP}}$ and~$r^*$ it suffices to find a feasible solution for~\eqref{eq:mik:ilp} with objective function value close to~$r^*$ in order to prove the lemma. 
	
	To this end, let~$\configs_r$ be the configurations that correspond to the inequalities added by the Ellipsoid Method while solving~(D$_r$) for~$r = \frac{r^*}{1+\eps}$. Consider the problems~\eqref{eq:mik:ilp} and~\eqref{eq:mik:dual:app} restricted to the variables~$y_c$, for~$c \in \configs_r$, and to the constraints corresponding to~$c \in \configs_r$, respectively, and denote these restricted LPs by~(P$'$) and~(D$'$). Let~$v'$ and~$v_D'$ be their respective optimal values.
	
	It holds that~$v_D' > r$ as the Ellipsoid Method also returns infeasibility for~(D$'$) when run on~(D$'$) extended by the objective function constraint for~$r$. As~$y=0$ is feasible for~(P$'$) and~$\alpha = 0$,~$\beta = \max_{c \in \configs_r} v_c$, and~$\gamma = 0$ are feasible for~(D$'$), their objective function values coincide by strong duality, i.e.,~$v' = v_D' > r $. If we have an optimal solution to~(P$'$), then this solution is also feasible for~\eqref{eq:mik:ilp} and achieves an objective function value 
	\begin{equation*}
		v' > \frac{r^*}{1+\eps}\geq \frac{1-\eps}{1+\eps}v_{\text{LP}},
	\end{equation*} 
	where we used Equation~\eqref{eq:mik:v*<=r*} for the last inequality. 
	
	It remains to show that the Ellipsoid Method can be applied to the setting presented here and that the running time of the just described algorithm is indeed bounded by a polynomial in~$\log n$,~$\frac1\eps$, and~$\log U$.	
	Recall that~$U$ is an upper bound on the absolute values of the denominators and numerators appearing in~\eqref{eq:mik:dual:app}, i.e., on $\capa m$ and $ nv_{\max}$. Observe that by \cref{lem:mik:separation}, the separation oracle runs in time~$\OO\big( \frac{\log^4 n}{\eps^{14}}\big)$. The number of iterations of the Ellipsoid Method will be bounded by a polynomial in~$\log U$ and~$\tilde{n} \in \OO\big(\frac{\log^2 n }{\eps^4}\big)$. Here,~$\tilde n$ is an upper bound on the number of variables in the problem~(D$_r$) (and hence also in~\eqref{eq:mik:dual:appeps}). 
	
	The feasible region of~(D$_r$) is a subset of the feasible region of~\eqref{eq:mik:dual:appeps}, even when the objective function constraint is added to the latter LP. The Ellipsoid Method usually is applied to full-dimensional, bounded polytopes that guarantee two bounds: If the polytope is non-empty, then its volume is at least~$v > 0$. The polytope is contained in a ball of volume at most~$V$. As shown in the book by Bertsimas and Tsitsiklis~\cite{BertsimasT1997}, these assumptions can always be ensured and the parameters~$v$ and~$V$ can be chosen as polynomial functions of~$\tilde n$ and~$U$. Since we cannot check feasibility of~(D$_r$) directly, we choose the parameters~$v$ and~$V$ as described in~\cite[Chapter 8]{BertsimasT1997} for the problem~\eqref{eq:mik:dual:appeps} extended by the objective function constraint for~$r$. After~$N = \OO\big( \tilde n \log \frac V v\big)$ iterations, the modified Ellipsoid Method either finds a feasible solution to~\eqref{eq:mik:dual:appeps} with objective function value at most~$r$ or correctly declares (D$_r$) infeasible. In \cite[Chapter 8]{BertsimasT1997} it is shown that the number of iterations~$N$ satisfies~$N = \OO(\tilde{n}^4 \log (\tilde n U))$ and that the overall running time is polynomially bounded in~$\tilde n$ and~$\log U$. 
	
	Hence,~(P$'$), the problem~\eqref{eq:mik:ilp} restricted to variables corresponding to constraints added by the Ellipsoid Method, has at most~$N$ variables and, thus, a polynomial time algorithm for linear programs can be applied to~(P$'$) to obtain an optimal solution in time~$\big(\frac{\log U}{\eps}\big)^{\OO(1)}$. 
\end{proof}

\paragraph*{Integrally Packing Fractional Solutions}
\label{subsec:packing}

One of the main ingredients to the dynamic algorithms in this section is a configuration ILP. As solving general ILPs is $\NP$-hard, in a first step, we relax the integrality constraints and accept fractional solutions before rounding the obtained solution to an integral one. The first lemma of this section describes how to obtain an integral solution with slightly more knapsacks given a fractional solution to a certain class of packing ILPs. Even after rounding, the configuration ILPs only take care of integrally packing big items, i.e., items with~$s_j \geq \eps \capa_i$. Therefore, the second lemma focuses on packing small items integrally given an integral packing of big items that reserves enough space for packing these items fractionally using resource augmentation. 

Formally, we consider a packing problem of items into a given set~$K$ of knapsacks with capacities~$\capa_i$. These knapsacks are grouped to obtain the set~$\mathcal G$ where group~$g \in \mathcal G$ contains~$m_g$ knapsacks and has total capacity~$S_g$. 
The objective is to maximize the total value without violating any capacity constraint. Each item~$j$ has a certain type~$t$, i.e., value $v_j = v_t$ and size~$s_j = s_t$, and in total there are~$n_t$ items of type~$t$. Items can either be packed as single items or as part of configurations. A configuration~$c$, that packs~$n_{c,t}$ items of type~$t$, has total \mbox{value~$v_c = \sum_{t} n_{c,t} v_t$} and size~$s_c=\sum_{t} n_{c,t} s_t$. The set~$E$ represents the items and the configurations that we are allowed to pack for maximizing the total value. Without loss of generality, we assume that for each element~$e \in E$ there exists at least one knapsack~$i$ where this element fits, i.e,~$s_e \leq S_i$.

Let~$0 \leq \delta \leq 1$ and~$s \geq 0$. Later we will choose $\delta=1-\Theta(\eps)$ since intuitively an $\Theta(\eps)$-fraction of the knapsacks remains unused. Consider the packing ILP for the above described problem with variables~$z_{e,g}$, where~$e \in E$ and~$g \in \mathcal{G}$. The ILP may additionally contain constraints of the form \[
\sum_{e \in E, g \in \mathcal G'} s_e z_{e,g} \leq \delta \sum_{g \in \mathcal{G}'} \capa_g - s \text{ and } \sum_{e \in E', g \in \mathcal{G}'} z_{e,g} \leq \delta \sum_{g \in \mathcal G'} m_g \, , 
\] i.e., the elements assigned to a subset of knapsack types~$\mathcal G'$ do not violate the total capacity of a~$\delta$-fraction of the knapsacks in~$\mathcal G'$ while reserving a space of size~$s$ and a particular subset~$E'$ of these elements uses at most a~$\delta$-fraction of the available knapsacks.

Let~$v(z)$ be the value attained by a certain solution~$z$ and let~$n(z)$ be the number of non-zero variables of~$z$. The following lemma shows that there is an integral solution of value at least~$v(z)$ using at most~$n(z)$ extra knapsacks. The high-level idea of the proof is to round down each non-zero variable~$z_{e,g}$ and pack the corresponding elements as described by~$z_{e,g}$. For achieving enough value, we additionally place one extra element~$e$ into the knapsacks given by resource augmentation for each variable~$z_{e,g}$ that was subjected to rounding.

More precisely, for each element~$e$ and each knapsack group~$g$, we define~$\bar{z}_{e,g}' = \lfloor z_{e,g} \rfloor$ and~\mbox{$\bar{z}_{e,g}'' = \lceil z_{e,g} - \bar{z}_{e,g}'\rceil$}. Note that~$\bar{z}'+ \bar{z}''$ may require more items of a certain type than are available. Hence, for each item type~$t$ that is now packed more than~$n_t$ times, we reduce the number of items of type~$t$ in~$\bar{z}'+ \bar{z}''$ by either adapting the chosen configurations if~$t$ is packed in a configuration or by decreasing the variables of type~$z_{t,g}$ if items of type~$t$ are packed as single items in knapsacks of group~$g$. Let~$z'$ and~$z''$ denote the solutions obtained by this transformation. For some elements~$e$, the packing described by~$z_{e,g}'+ z_{e,g}''$ may now use more or less elements than~$z_{e,g}$ due to the just described reduction of items.	

\begin{restatable}[]{lemma}{hgRoundingConfigLP}\label{lem:hg:roundingConfigLP}
	Any fractional solution~$z$ to the packing ILP described above can be rounded to an integral solution with value at least~$v(z)$ using at most~$n(z)$ additional knapsacks of capacity~$\max_{i \in K} S_i$. 
\end{restatable}

\begin{proof}
	Consider a particular item type~$t$. If~$\bar{z}'+\bar{z}''$ packs at most~$n_t$ items of this type, then the value achieved by~$z$ for this particular item type is upper bounded by the value achieved by~$z'+ z''$. If an item type was subjected to the modification, then~$z'+z''$ packs exactly~$n_t$ items of this type while~$z$ packs at most~$n_t$ items. This implies that~$v(z'+ z'') \geq v(z)$. %
	
	It remains to show how to pack~$\bar{z}'+\bar{z}''$ (and, thus,~$z'+ z''$) into the knapsacks given by~$K$ and potentially~$n(z)$ additional knapsack. Clearly,~$\bar{z}'$ can be packed exactly as~$z$ was packed. If~$z_{e,g} = 0$ for~$e \in E$ and~$g \in \mathcal G$, then~$\bar{z}'_{e,g} = 0$. Hence, the number of non-zero entries in~$\bar{z}''$ is bounded by~$n(z)$. Consider one element~$e \in E$ and a knapsack group~$g$ with~$\bar{z}_{e,g}'' = 1$ and let~$i$ be a knapsack where~$e$ fits. Pack~$e$ into~$i$. 
	
	Since reducing the number of packed items of a certain type only decreases the size of the corresponding configuration or the number of individually packed elements, the solution~$z'+ z''$ can be packed exactly as described for~$\bar{z}'+\bar{z}''$. Therefore, we need at most~$n(z)$ extra knapsacks to pack~$z''$, which concludes the proof.
\end{proof}

Having found a feasible solution with the Ellipsoid Method, we use Gaussian elimination to obtain a basic feasible solution with no worse objective function value. We note that this procedure has a running time bounded by~$(N |\types|)^{\OO(1)}$, where~$N$ is the number of non-zero variables in the solution found by the Ellipsoid Method. 
Since basic feasible solutions have  at most $|\types|+2$ non-vanishing variables, the assumptions~$\frac{16}{\eps^7} \log^2 n \leq m$ and~$m < n$ imply~$\frac{16}{\eps^7} \log^2 m \leq m$. This in turn guarantees~$|\types|+2 \leq \lfloor \eps m\rfloor $. Hence, rounding the solution as described above uses at most~$\lfloor (1-2\eps)m\rfloor $ knapsacks and achieves a value of at least~$v_{\text{LP}}$. 

\begin{corollary}\label{cor:mik:bigm}
	If~$\frac{16}{\eps^7} \log^2 n \leq m$, any feasible solution of the LP relaxation of \eqref{eq:mik:ilp} with at most~$N$ non-zero variables can be rounded to an integral solution using at most~$\lfloor (1-2\eps)m \rfloor$ knapsacks with total value at least $v_{\text{LP}}$ in time~$(N|\types|)^{\OO(1)}$. 
\end{corollary}

Given an integral packing of big items, we explain how to pack small items, i.e., items with~$s_j < \eps \capa$, using resource augmentation.  More precisely, let~$K$ be a set of knapsacks and let~$J_S'\subseteq J$ be a subset of items that are small with respect to every knapsack in~$K$. Let~$J'\subset J$ be a set of items admitting an integral packing into~$m = |K|$ knapsacks that preserves a space of at least~$s(J_S')$ in these~$m$ knapsacks. We develop a procedure to extend this packing to an integral packing of all items~$J'\cup J_S'$ in~$\lceil (1+\eps)m \rceil $ knapsacks where the~$\lceil \eps m\rceil $ additional knapsacks can be chosen to have the smallest capacity of knapsacks in~$K$.

We use a packing approach similar to \textsc{Next Fit} for the problem \textsc{Bin Packing}. That is, consider an arbitrary order of the small items and an arbitrary order of the knapsacks filled with big items. We open the first knapsack in this order for small items. If the next small item~$j$ still fits into the open knapsack, we place it there and decrease the remaining capacity accordingly. If it does not fit anymore, we pack this item into the next empty slot of an additional knapsacks (possibly opening a new one), close the current original knapsack, and open the next one for packing small items. We call such an item~\emph{cut}.

\begin{restatable}{lemma}{hgPackP}\label{lem:hg:packP}
	The procedure described above feasibly packs all items~$J'\cup J_S'$ in~$\lceil (1+\eps)m\rceil $ knapsacks where the~$\lceil \eps m \rceil $ additional knapsacks can be chosen to have the smallest capacity of knapsacks in~$K$.  
\end{restatable}

\begin{proof}
	We start by showing that all small items are packed after the last original knapsack is closed. Toward a contradiction, suppose that there is a small item~$j$ left \emph{after} all original knapsacks were closed while packing small items. As a knapsack is only closed if the current small item does not fit anymore, this implies that the volume of all small items that are packed so far have a total volume at least as large as the total remaining capacity of knapsacks in~$K$ after packing~$J'$. Since~$j$ is left unpacked after all original knapsacks have been closed, the total volume of all items in~$J'\cup J_S'$ is strictly larger than the total capacity of the original knapsacks in~$K$. This contradicts the assumption imposed on~$J_B'$ and on~$J_S'$. Hence, all items in~$J_S'$ are packed. Therefore, the packing created by the procedure is integral and feasible. 
	
	It remains to bound the number of additional knapsacks. Observe that each item that we packed into a knapsack given by resource augmentation while an original knapsack was still available, implied the closing of the current knapsack and the opening of a new one. Hence, for each original knapsack at most one small item was placed into the additional knapsacks. Thus, at most $m$ small items are packed into the additional knapsacks.
	Since by definition of small items at least~$\frac1\eps$ items fit into one additional knapsack, we only need~$\lceil \eps m\rceil $ extra knapsacks for such items. 
\end{proof}

\paragraph*{Answering Queries}

Note that, throughout the course of the dynamic algorithm, we only implicitly store solutions. In the remainder of this section, we explain how to answer the queries stated in the main part and bound the running times of the corresponding algorithms.
We refer to the time frame between two updates as a \emph{round} and introduce a counter $\tau$ that is increased after each update and denotes the current round.
Since answers to queries have to stay consistent in a round, we \emph{cache} existing query answers by additionally storing a round~$t(j)$ and a knapsack~$k(j)$ for each item in the data structure for items where~$t(j)$ stores the last round in which item~$j$ has been queried and~$k(j)$ points to the knapsack of~$j$ in round~$t(j)$. 
Storing~$t(j)$ is necessary since resetting the cached query answers after each update takes too much running time. If~$j$ was {not selected} in~$t(j)$, we store and return this with~$k(j) =0$.

Let~$\bar y_c$, for~$c \in \configs$, be the packing for the big items in terms of the variables of the configuration ILP. During the Ellipsoid Method and the rounding of the fractional solution to an integral solution, the set~$\overline \configs := \{c \in \configs: \bar{y}_{c} \geq 1 \}$ was constructed. We assume that this set is ordered in some way and stored in a list. In the following we use the position of $c \in \overline \configs$ in that list as the index of $c$. For assigning~$\bar{y}_c$ distinct knapsacks to~$c \in \overline \configs$, we use the ordering of the configurations and map the knapsacks $\sum_{c' = 1}^{c-1} \bar{y}_{c'} + 1 , \ldots, \sum_{c' = 1}^{c} \bar{y}_{c'}$ to~$c$. 

For small items, we store all items in a balanced binary search tree sorted by non-increasing density. For simplicity, let~$P_S = \{1,\ldots,j^\star-1\}$ be the set of items (sorted by non-increasing density) that translate the guess~$v_S$ into the size~$s_S$ of small items in the current solution. Item~$j^\star$ is packed into its own knapsack. Any item~$j \leq j^\star-1$ is either packed regularly into the empty space of a knapsack with a configuration or it is packed into a knapsack designated for packing cut small items. Therefore, we maintain two pointers:~$\kappa^r$ points to the next knapsack where a small item is supposed to go if it is packed \emph{regularly} and~$\kappa^c$ points to the knapsack where the next \emph{cut} small item is packed. We initialize these values with~$\kappa^r = 1$ and~$\kappa^c = \lfloor (1-2\eps)m \rfloor + 1$. To determine if an item is packed regularly or as cut item, we store in~$\rho^r$ the remaining capacity of~$\kappa^r$ initialized with~$\kappa^r = S - s_1$ where~$s_1$ is the size of the first configuration in~$\overline \configs$. We store in~$\rho^c$ the remaining slots of small items in knapsack~$\kappa^c$ and initialize this with~$\rho^c = \frac1\eps$. 

For each type~$t$ of big items, we maintain a pointer~$\kappa_t$ to the knapsack where the next queried item of type~$t$ is supposed to be packed. Moreover, the counter~$\eta_t$ stores how many slots~$\kappa_t$ still has available for items of type~$t$. These two values are initialized with the first knapsack that packs items of type~$t$ and~$\eta_t = n_{c,t}$ where~$c$ is the configuration of~$\kappa_t$. If no items of type~$t$ are packed, we set~$\kappa_t=0$. Let~$\bar{n}_t$ denote the number of items of type~$t$ belonging to solution~$\bar{y}$. We will only pack the first, i.e., smallest,~$\bar n_t$ items of type~$t$. 
Figure~\ref{fig:mik:queries} depicts the pointers and counters after some items already have been queried. 

\begin{figure}[tbh]
	\centering 
	\begin{tikzpicture}[scale = .9]
		\foreach \x in {0,...,14}{
			\draw[draw=black, fill=none] (\x-.1,0) rectangle (\x+.6, 3.75); 
		}

		\draw[ubred, dashed, thick] (8.75,4) -- (8.75, -.25);
		\foreach \x in {9,...,11}{
			\foreach \y in {0,...,4}
			\draw[draw=ubred, fill = ubred!10] (\x, 0.1+\y*.75) rectangle (\x+.5, 0.65 + \y*.75);
		}
		\draw[draw=ubred, fill = ubred!10] (14, 0.1) rectangle (14.5, 0.65);
		
		\foreach \y/\s in {0/.05, 1/.2, 2/.15, 3/.1, 4/.2} {
			\draw[draw=none, fill = ubred] (9, 0.1+\y*.75) rectangle (9 + .5, 0.65 + \y*.75-\s);
		}
		\draw[draw=none, fill = ubred] (10, 0.1) rectangle (10 + .5, 0.45);

		\foreach \x in {0,1}{
			\draw[draw=utnavy, fill = utnavy!10] (\x, 0.1) rectangle (\x+.5, 3.3);
			\draw[draw=ubred, fill = ubred!10] (\x, 3.4) rectangle (\x+.5, 3.65);
		}
		\foreach \x in {2,...,5}{
			\draw[draw=utblue, fill = utblue!10] (\x, 0.1) rectangle (\x+.5, 1.6);
			\draw[draw=utblue, fill = utblue!10] (\x, 1.7) rectangle (\x+.5, 3.2);
			\draw[draw=ubred, fill = ubred!10] (\x, 3.3) rectangle (\x+.5, 3.65);
		}
		\foreach \x in {6,...,8}{
			\draw[draw=utgreen, fill = utgreen!10] (\x, 0.1) rectangle (\x+.5, 2);
			\draw[draw=utforest, fill = utforest!10] (\x, 2.1) rectangle (\x+.5, 3);
			\draw[draw=ubred, fill = ubred!10] (\x, 3.1) rectangle (\x+.5, 3.65);
		}
		
		\foreach \x/\y/\s/\col in {0/0.1/3.1/utnavy, 2/.1/1.5/utblue, 2/1.7/1.4/utblue, 3/.1/1.4/utblue, 6/.1/1.8/utgreen, 6/2.1/.9/utforest, 7/.1/1.9/utgreen,8/.1/1.7/utgreen, 7/2.1/.8/utforest}{
			\draw[fill =\col, draw=none] (\x,\y) rectangle (\x + .5, \y + \s);
		}
		
		\foreach \x/\y/\s in {0/3.4/.2, 2/3.3/.35, 3/3.3/.1, 4/3.3/.1, 5/3.3/.25, 6/3.1/.3}{
			\draw[fill = ubred, draw = none] (\x,\y) rectangle (\x+.5, \y + \s);
		}

		\node[font = \small] at (10.25,-.25) {$\kappa^c$};	
		\node[font = \small] at (10.25,-.7) {$\rho^c = 4$};
		
		\node[font = \small] at (6.25,-.25) {$\kappa^r$};
		\node[font = \small] at (6.25, -.7) {$\rho^r =$ \tikz{\draw[draw=white] (0,0) -- (.25,0);} $-$ \tikz{\draw[draw = white] (0, .15) -- (0, -.15);} }; 
		
		\draw[fill = ubred!10, draw = ubred, xshift = 6.25cm, yshift = -.7cm] (0, .275) rectangle (.25, -.275);
		\draw[fill = ubred, draw = none, xshift = 7cm, yshift = -.7cm] (0, .15) rectangle (.25, -.15);

		\node[font = \small] at (1.25,-.32) {$\kappa_1$};	
		\node[font = \small] at (1.25,-.73) {$\eta_1 = 1$};
		\node[font = \small, white] at (0.25,1.75) {$1$};
		
		\node[font = \small] at (3.25,-.32) {$\kappa_2$};	
		\node[font = \small] at (3.25,-.73) {$\eta_2 = 1$};
		\node[font = \small, black] at (2.25,.85) {$2$};
		\node[font = \small, black] at (2.25,2.3) {$2$};	
		\node[font = \small, black] at (3.25,.85) {$2$};	
		
		\node[font = \small] at (8.25,-.32) {$\kappa_4$};	
		\node[font = \small] at (8.25,-.73) {$\eta_4 = 1$};
		\node[font = \small, white] at (6.25,2.5) {$4$};
		\node[font = \small, white] at (7.25,2.5) {$4$};

		\node[font = \small, black] at (6.25,1) {$3$};	
		\node[font = \small, black] at (7.25,1) {$3$};
		\node[font = \small, black] at (8.25,1) {$3$};

		\draw [decorate,decoration={brace,amplitude=6pt},ubred]  (8.9, 4.1) -- (11.6, 4.1) node [black,midway,yshift =  24pt,font=\small, black] {$\lceil \eps m \rceil$  knapsacks} ;
		\draw [decorate,decoration={brace,amplitude=6pt},ubred]  (-.1, 4.1) -- (8.6, 4.1) node [black,midway,yshift =  24pt,font=\small, black] {$\lfloor (1-2\eps) m \rfloor$ knapsacks};
		\node[yshift = 12pt, font = \small] at (4.25,4.1) {with configurations}; 
		\node[yshift = 18pt, font = \small] at (14.25, 4.1) {$j^\star$};
		\node[yshift = 12pt, font = \small] at (10.25,4.1) {for cut items};
	\end{tikzpicture}
	\caption{Pointers and counters used for answering queries: Lightly colored rectangles represent slots to be filled with items. Big (blue and green) items are packed one item per slot. Item type~$3$ does not have any slots left. Small (red) items are packed either until the slot is filled (left side) or one item per slot (right side). The not yet queried, small item~$j^\star$ gets its own knapsack.}\label{fig:mik:queries}
\end{figure}
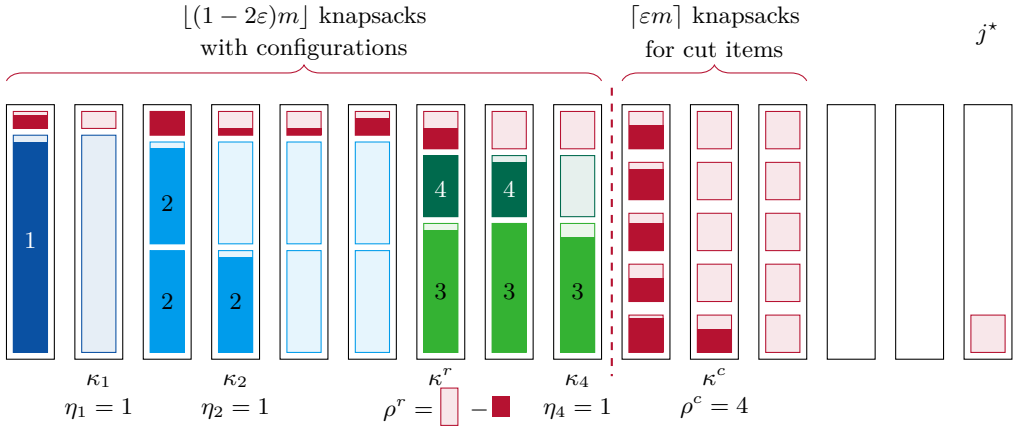

Consider a queried small item~$j$. If~$t(j) = \tau$, we return~$k(j)$. Otherwise, set~$t(j) = \tau$ and determine whether~$j$ is currently part of the solution. If~$j$ does not belong to the densest~$j^\star$ items, we return~$k(j) = 0$. Otherwise, we determine where~$j$ is packed. If~$j = j^\star$, we return~$k(j) = m$. Else, we figure out whether~$j$ is packed into the knapsack~$\kappa^r$ or into~$\kappa^c$. If~$\rho^r \geq s_j$, we simply update~$\rho^r$ to~$\rho^r - s_j$ and return~$k(j) = \kappa^r$. Otherwise, we decrease~$\rho^c$ by one and pack~$j$ as cut item in~$\kappa^c$. If~$\rho^c = 0$ holds after the update, we increase~$\kappa^c$ by one and set~$\rho^c = \frac1\eps$. Further, we need to close~$\kappa^r$ and update~$\kappa^r$ and~$\rho^r$ accordingly. To this end, we increase~$\kappa^r$ by one and determine~$\rho^r$, the remaining capacity in knapsack~$\kappa^r$. Then, we return~$k(j)$.

Consider a queried big item~$j$. If~$t(j) = \tau$, we return~$k(j)$. Otherwise, we set~$t(j) = \tau$ and compute whether item~$j$ is packed by the current solution. Let~$V_\ell$ be the value class of~$j$. If~$\ell \notin \{\llmin, \ldots, \lmax\}$, we return~$k(j) = 0$. Otherwise, we retrieve the type~$t$ of item~$j$. Given~$t$, we determine if~$j$ belongs to the first~$\bar n_t$ items of type~$t$. If this is not the case, we return~$k(j) =0$. If this is the case, then we set~$k(j) = \kappa_t$ instead and we decrease~$\eta_t$ by one. If this remains non-zero, we return~$k(j) = \kappa_t$. Otherwise, we find the next knapsack that packs items of type~$t$ and update~$\kappa_t$ and~$\eta_t$ accordingly before returning~$k(j)$.

\subparagraph*{Answering Item Queries.}
\begin{enumerate}
	\item[1)] \textbf{Check cache.} Let~$\tau$ be the current round and let~$j$ be the queried item. If~$t(j) = \tau$, return~$k(j)$. 
	\item[2)] \textbf{Answer queries for non-cached items.} Set~$t(j) = \tau$. If~$s_j \leq \eps \capa$, item~$j$ is small. Otherwise,~$j$ is big.
	
	\textbf{Small items.} If~$j > j^\star+1$, return \textsc{not selected} and set~$k(j) = 0$. 
	
	If~$j = j^\star+1$, return~$k(j) = m$. 
	
	Otherwise, determine if~$j$ is packed regularly or as cut item: If~$s_j \leq \rho^r$, return~$k(j) = \kappa^r$ and update~$\rho$ accordingly. Otherwise, return~$k(j) = \kappa^c$. Decrease $\rho^c$ by one and if~$\rho^c$ now holds, increase~$\kappa^c$ by one and set~$\rho^c = \frac1\eps$. Increase~$\kappa^r$ by one and update~$\rho^r$ accordingly to reflect the empty space in~$\kappa^r$. 
	
	\textbf{Big items.} Determine the value class~$V_\ell$ of~$j$. If~$\ell < \lmin$, return \textsc{not selected}. Otherwise, determine the item type~$t$ of~$j$ by retracing the steps of the  {oblivious} linear grouping. 
	
	If~$j$ is not among the first~$\bar n_t$ items of type~$t$, return \textsc{not selected} and set~$k(j) = 0$. 
	
	Otherwise, return~$k(j) = \kappa_t$ and decrease~$\eta_t$ by one. If~$\eta_t = 0$, increase~$\kappa_t$ to the next knapsack for type~$t$ and update~$\eta_t$ accordingly. If no such knapsack exists, set~$\kappa_t = 0$. 
\end{enumerate}

For being able to return the solution value in constant query time, we actually compute the solution value once after each update operation and store it. More precisely, the value achieved by the small items,~$v_S$ can be computed with a prefix computation of the first~$j^\star$ items in the density-sorted tree for small items. For computing the value of big items, we consider each value class~$V_\ell$ with~$\ell \in \{\llmin ,\ldots,\lmax\}$ individually. Per value class and per item type, we use prefix computation to determine the value~$v_t$ of the first~$\bar n_t$ items of type~$t$. \cref{lem:mik:querying:solval} guarantees that the running time is indeed upper bounded by the update time and, thus, does not change the order of magnitude.

\subparagraph*{Answering the Solution Value Query.}
\begin{enumerate}
	\item[1)] \textbf{Value of small items.} Calculate~$v_S = \sum_{j = 1}^{j^\star} v_j$ with prefix computation. 
	\item[2)] \textbf{Value of big items.} For each item type~$t$, calculate~$v_{B,t}$ the value of the first~$\bar n_t$ items[3)] of type~$t$ using prefix computation. 
	\item \textbf{Value.} Return~$v_S + \sum_{t \in \types} v_{B,t}$.
\end{enumerate}

When queried the complete solution, we return a list of packed items together with their respective knapsacks. To this end, we start by querying the~$j^\star$ densest small items using the algorithm for item queries. For big items, we query the first~$\bar n_t$ items of each item type~$t \in \types$.

\subparagraph*{Answering the Solution Query}
\begin{enumerate}
	\item[1)] \textbf{Small items.} Query each item~$j=1,\ldots j^\star+1$ and return the solution.
	\item[2)] \textbf{Big items.} For each type~$t \in \types$, query the first~$\bar n_t$ items and return the solution.
\end{enumerate}

We prove the parts of the following lemmas individually. 

\mikQueries*

\begin{restatable}{lemma}{mikQueriesCorrectness}\label{lem:mik:queries:correctness}
	The solution determined by the query algorithms is feasible and achieves the claimed total value.
\end{restatable}

\begin{proof}
	By construction of~$t(j)$ and~$k(j)$, the answers to queries happening between two consecutive updates are consistent. 
	
	For small items, observe that~$1,\ldots,j^\star$ are the densest small items in the current instance. By~\cref{lem:hg:packP}, the packing obtained by our algorithms is feasible for these items.
	
	For big items, we observe that their actual size is at most the size of their item types. Hence, packing an item of type~$t$ where the implicit solution packs an item of type~$t$ is feasible. The algorithms correctly pack the first~$\bar n_t$ items of type~$t$. A knapsack with configuration~$c\in \overline \configs$ correctly obtains~$n_{c,t}$ items of type~$t$. Moreover, each configuration~$c \in \overline \configs$ gets assigned~$\bar y_c$ knapsacks. Hence, the algorithm packs exactly the number of big items as dictated by the implicit solution~$\bar y$. 
\end{proof} 

\begin{restatable}{lemma}{mikQueriesBig}\label{lem:mik:querying:big} 
	The data structures for big items can be generated in time 	$\OO\big(\frac{\log^4 n}{\eps^9}\big)$. Queries for big items can be answered in time $\OO\big(\log n + \log \frac{\log n}{\eps}\big)$. %
\end{restatable}

\begin{proof} 
	We assume that~$\overline \configs$ is already stored in some list. We start by formally mapping knapsacks to configurations. To this end, we create a list~$\alpha = (\alpha_c)_{c \in \overline \configs}$, where~$\alpha_c = \sum_{c'=1}^{c-1} \bar y_{c'}$ is the first knapsack with configuration~$c \in \overline \configs$. Using~$\alpha_{c} = \alpha_{c-1} + \bar y_{c-1}$, we can compute these values in constant time. Hence, by iterating once through~$\overline \configs$, list~$\alpha$ can be generated in $\OO(|\overline \configs|)$. 
	
	We start by recomputing the indices needed for the  {oblivious} linear grouping approach. For each value class~$V_\ell$ with~$\ell \in \{\llmin ,\ldots,\lmax\}$, we access the items corresponding to the boundaries of the item types~$\types_\ell$ in order to obtain the item types~$\types_\ell$. By construction, these types are already ordered by non-decreasing size~$s_t$. By \cref{lem:hg:RunningTime}, these item types can be computed in time $\OO\big(\frac{\log^4 n}{\eps^4}\big)$ and stored in one list~$\types_\ell$ per value class~$V_\ell$. 
	
	For maintaining and updating the pointer~$\kappa_t$, we generate a list~$\configs_t$ of all configurations~$c \in \overline \configs$ with~$n_{c,t} \geq 1$. By iterating through each~$c \in \overline \configs$, we can add~$c$ to the list of~$t$ if~$n_{c,t} \geq 1$. We additionally store~$n_{c,t}$ and~$\alpha_c$ in the list~$\configs_t$. While iterating through the configurations, we additionally compute~$\bar n_t = \sum_{c \in \overline \configs} \bar y_c n_{c,t}$ and store~$\bar n_t$ in the same list as the item types~$\types_\ell$. 
	Note that, since the list of $\overline \configs$ is ordered by index, the created lists $\configs_t$ are also sorted by index. For each item type, we point $\kappa_t$ to the first knapsack of the first added configuration~$c$ and set~$\eta_t = n_{c,t}$. If the list of an item type remains empty, we set~$\kappa_t = 0$. 	
	Since each configuration contains at most $\frac1\eps$ item types, the lists~$\configs_t$ can be generated in time~$\OO\big(\frac{|\overline \configs||\types|}{\eps}\big)$. 
	
	Now consider a queried big item $j$. In time $\OO(\log n)$, we can decide whether $j$ has already been queried in the current round. If not, let $V_\ell$ be the value class of $j$, which was computed upon arrival of~$j$. If~$\ell \notin \{\llmin, \ldots, \lmax\}$, then~$j$ does not belong to the current solution and no data structures need to be updated. Otherwise, the type of $j$ is determined by accessing the item types~$\types_\ell$ in time~$\OO\big(\log \frac{\log n}{\eps}\big)$. Once~$t$ is determined,~$\bar{n}_t$ can be added to the left boundary of type~$t$ in order to determine if~$j$ is packed or not. If~$j$ belongs to the current solution, pointer~$\kappa_t$ dictates the answer to the query. 
	
	In order to update $\kappa_t$ and $\eta_t$, we extract~$c$, the configuration of knapsack~$\kappa_t$ in time~$\OO(\log |\overline \configs|)$ by binary search over the list~$\alpha$. If~$\kappa_t + 1 < \alpha_{c+1}$, then~$\kappa_t$ is increased by one and~$\eta_t$ set to~$n_{c,t}$ in constant time. If not, then the next configuration~$c'$ containing~$t$ can be found with binary search over the list~$\configs_t$ in time~$\OO(\log |\overline \configs|)$. If no such configuration is found, we set~$\kappa_t = 0$. Otherwise, we set~$\kappa_t = \alpha_{c'}$ and~$\eta_t = n_{c',t}$.  
	Overall, queries for big items can be answered in time $\OO\big(\max \big\{\log |\overline \configs|, \log \frac{\log n}{\eps}\big\}\big)$. 
	
	Observing that~$|\overline \configs| \in \OO(|\types|) = \OO \big(\frac{\log^2 n}{\eps^4}\big)$ completes the proof.
\end{proof}

\begin{restatable}{lemma}{mikQueriesSmall}\label{lem:mik:querying:small}
	Given the data structures for big items, the data structures for small items can be generated in time $\OO\big( \log \frac{\log n}{\eps} \big)$. The running time for answering queries for small items is~$\OO\big(\log n + \max \big\{ \log \frac{\log n}{\eps}, \frac{1}{\eps}\big\} \big)$. 
\end{restatable}

\begin{proof}
	We initialize~$\kappa^r = 1$ and~$\rho^r = \capa - s_1$ where~$s_1$ is the total size of the configuration assigned to the first knapsack. For packing cut items, we use the pointer~$\kappa^c$ to the current knapsack for cut items while~$\rho^c$ stores the remaining slots of small items. We initialize these values with~$\kappa^c= \lfloor (1-2\eps)m \rfloor + 1$ and~$\rho^c = \frac{1}{\eps}$. These initializations can be computed in time~$\OO(\log |\overline \configs|)$ (for extracting~$s_1$).
	
	Now consider a queried small item~$j$. In time~$\OO(\log n)$ we can decide whether~$j$ has already been queried in the current round. In constant time, we can decide whether~$j > j^\star$. If~$j > j^\star$, the answer is \textsc{not selected}. If~$j = j^\star$, we return~$m$. If~$j< j^\star$, the algorithm only needs to decide if~$j$ is packed into~$\kappa^{r}$ or~$\kappa^c$, which can be done in constant time. Finally,~$\kappa^{r}$ and~$\kappa^c$ as well as~$\rho^r$ and~$\rho^c$ need to be updated.
	While~$\kappa^c$,~$\kappa^{r}$, and~$\rho^c$ can be updated in constant time, we need to compute the configuration~$c$ and remaining capacity~$\capa - s_c$ of knapsack~$\kappa^{r}$ when the pointer is increased. By using binary search over the list~$\alpha$, the configuration can be determined in time~$\OO(\log |\overline \configs|)$. Once the configuration is known,~$\rho^r$ can be calculated in time~$\OO\big(\frac{1}{\eps}\big)$. Overall, queries for small items can be answered in time~$\OO\big(\log n + \max \big\{ \log |\overline \configs|, \frac{1}{\eps}\big\} \big)$. 
	
	Using that~$|\overline \configs| \in \OO(|\types|) = \OO\big(\frac{\log^2n}{\eps^4}\big)$ concludes the proof.
\end{proof}

\begin{restatable}{lemma}{mikQueriesSolVal}\label{lem:mik:querying:solval}
	The total solution value can be computed in~$\OO\big(\frac{\log^3 n}{\eps^4}\big)$. 
	A query for the solution value can be answered in time~$\OO(1)$.
\end{restatable}

\begin{proof}
	The true value~$\tilde v_S$ achieved by the small items can be determined by computing the prefix of the first~$j^\star$ items in the density-sorted tree for small items in time~$\OO(\log n)$ by \cref{lem:data-structure}. 
	
	For computing the value of a big item, we consider each value class~$V_\ell$ with~$ \ell \in \{\llmin, \ldots,  \lmax\}$ individually. There are at most~$\OO\big(\frac{\log n}{\eps^2}  \big)$ many value classes by \cref{lem:hg:GuessLMax}. For one value class, in time $\OO\big(\frac{\log n}{\eps^2}\big)$, iterate through the item types~$t$. For each item type, we can access the total value of the first~$\bar n_t$ items in time~$\OO(\log n)$ by \cref{lem:data-structure}. 
	
	As these running times are subsumed by the running time of the update operation, we actually compute the solution value once after each update operation and store the value allowing for constant running time to answer the query. 
\end{proof}

\begin{restatable}{lemma}{mikQueriesSol}\label{lem:mik:querying:sol}
	A query for the complete solution can be answered in time 
	$\OO\big(|P| \frac{\log^4n }{\eps^4 }  \log \frac{\log n}{\eps}\big)$, where~$P$ is the set of items in our solution.
\end{restatable}

\begin{proof}
	The small items belonging to~$P$ can be accessed in time~$\OO(j^\star\log n)$ by \cref{lem:data-structure}. By \cref{lem:mik:querying:small}, their knapsacks can be determined in time~$\OO\big(\log n + \max \big\{ \log \frac{\log n}{\eps}, \frac{1}{\eps}\big\} \big)$.
	
	For big items, we consider again at most~$\OO\big(\frac{\log n}{\eps^2}\big)$ many value classes individually. In time~$\OO\big(\frac{\log n}{\eps^2}\big)$, we access the boundaries of the corresponding item types. In time~$\OO(\bar n_t \log n )$, we can access the~$\bar n_t$ items of type~$t$ belonging to our solutions by \cref{lem:data-structure}. \cref{lem:mik:querying:big} ensures that their knapsacks can be determined in time~$\OO\big(\log n + \log \frac{\log n}{\eps}\big)$.
	
	In total, this bounds the running time by~$\OO\big(|P|  \frac{\log^4n }{\eps^4} \log \frac{\log n}{\eps} \big)$. 
\end{proof}

\section{Knapsacks with Resource Augmentation}
\label{app:multidiffknapsack}\label{sec:MDK:aug}
In this section, we consider instances for \multiknapsack with many knapsacks and arbitrary capacities. 
We show how to efficiently maintain a $(1+\eps)$-approximation when given, as resource augmentation, $L$ additional knapsacks that have the same capacity as a largest knapsack in the input instance, where~$L \in   \big(\frac{\log n}\eps\big)^{\OO(\epsfrac)}$. While we may pack items into the additional knapsacks, an optimal solution is not allowed to use them.
The algorithm will again solve the LP relaxation of a configuration ILP and round the obtained solution to an integral packing. However, in contrast to the problem for identical knapsacks, not every configuration fits into every knapsack and we therefore cannot just reserve a fraction of knapsacks in order to pack the rounded configurations since the knapsack capacities might not suffice.
For this reason, we employ resource augmentation in the case of arbitrary knapsack capacities.

Again, we assume that item values are rounded to powers of $(1+\eps)$ which results in value classes~$V_\ell$ of items with value~$v_j = (1+\eps)^\ell$. We prove the following theorem. 

\mmdkthm*

\subsection{Algorithm}
\paragraph*{Data structures}
In this section, we maintain three different types of data structures. For storing every item~$j$ together with its size~$s_j$, its value~$v_j$, and the index of its value class~$\ell_j$, we maintain one balanced binary search tree where the items are sorted by non-decreasing time of arrival. For each value class~$V_\ell$, we maintain one balanced binary tree for sorting the items with~$\ell_j = \ell$ in order of non-decreasing size. We store the knapsacks sorted in non-increasing capacity in one balanced binary tree. 

\paragraph*{Algorithm} 
The algorithm we develop in this section is quite similar to the dynamic algorithm for \multiknapsack with identical capacities. First, we use  {oblivious} linear grouping for the current set of items to obtain item types. However, in contrast to identical knapsacks, one particular item may be big with respect to one knapsack, small with respect to another, and may not even fit in a third knapsack. Thus, we use the item types to partition the knapsacks into groups to simulate knapsacks with identical capacities; see Figure~\ref{fig:ordinary:groups}. Within one group, we give an explicit packing of the big items into slightly less knapsacks than belonging to the group by solving a configuration ILP. For packing small items, we would like to use a guess of the size of small items per groups and later use again \textsc{Next Fit} to pack them integrally. However, since items classify as big in one knapsack group and as small in another group, instead of guessing the size of small items per knapsack group, we incorporate their packing into the configuration ILP by reserving sufficient space for the small items in each group. More precisely, we assign items as big items via configurations or as small items by number to the various groups. The remainder of the algorithm is straight-forward: we relax the integrality constraint to find a fractional solution and use the tools developed in \cref{lem:hg:roundingConfigLP,lem:hg:packP} to obtain an integral packing. Figure~\ref{fig:ordinary:solution} shows a possible solution including some knapsacks given by resource augmentation.

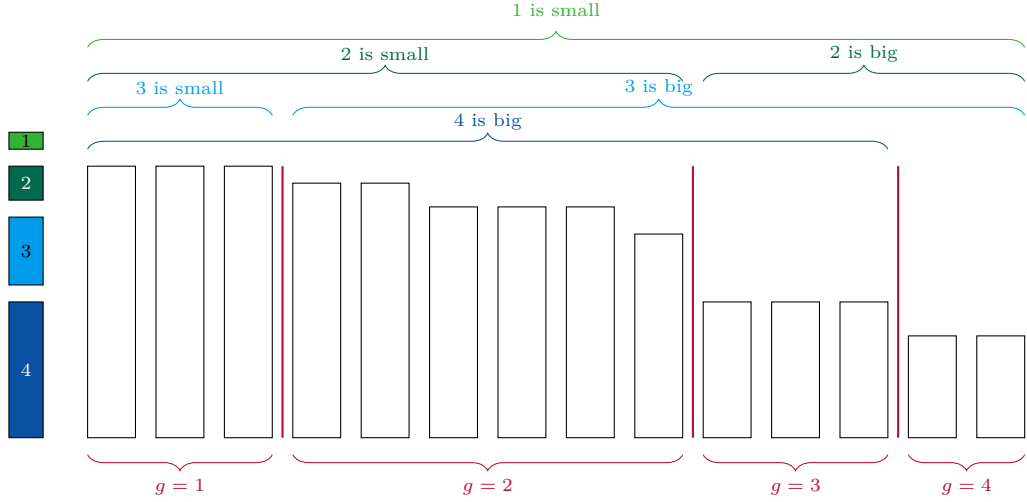
\begin{figure}[tbh]
	\centering 
	\begin{tikzpicture}[scale = .9]
		\foreach \x/\h/\n in {12/1.5/1, 9/2/2, 8/3/0, 5/3.4/2, 3/3.75/1, 0/4/2}{ %
			\foreach \y in {0,...,\n}{
				\draw[draw=black, fill=none] (1.25+\x+\y-.1,0) rectangle (1.25+\x+\y+.6, \h); 
			}
		}
		
		\foreach \y/\h/\col in {4.25/.25/utgreen, 3.5/.5/utforest, 2.25/1/utblue, 0/2/utnavy}{
			\draw[draw = black, fill = \col] (0, \y) rectangle (.5, \y + \h); 
		}
		
		\foreach \y/\h/\col/\l in {4.25/.25/black/1, 3.5/.5/white/2, 2.25/1/black/3, 0/2/white/4}{
			\node[\col, font = \scriptsize] at (0.25, \y + \h/2) {\l};
		}

		\draw [decorate,decoration={brace,amplitude=6pt},utnavy]  (1.15, 4.25) -- (12.85, 4.25) node [utnavy,midway,yshift =  +.37cm,font=\scriptsize] {4 is big};
		\draw [decorate,decoration={brace,amplitude=6pt},utblue]  (1.15, 4.75) -- (3.85, 4.75) node [utblue,midway,yshift =  +.37cm,font=\scriptsize] {3 is small};
		\draw [decorate,decoration={brace,amplitude=6pt},utblue]  (4.15, 4.75) -- (14.85, 4.75) node [utblue,midway,yshift =  +.37cm,font=\scriptsize] {3 is big};

		\draw [decorate,decoration={brace,amplitude=6pt},utforest]  (10.15, 5.25) -- (14.85, 5.25) node [utforest,midway,yshift =  +.37cm,font=\scriptsize] {2 is big};	
		\draw [decorate,decoration={brace,amplitude=6pt},utforest]  (1.15, 5.25) -- (9.85, 5.25) node [utforest,midway,yshift =  +.37cm,font=\scriptsize] {2 is small};
		\draw [decorate,decoration={brace,amplitude=6pt},utgreen]  (1.15, 5.75) -- (14.85, 5.75) node [utgreen,midway,yshift =  +.5cm,font=\scriptsize] {1 is small};	
		
		\foreach \x in {4,10,13} {
			\draw[thick, ubred] (\x, 0) -- (\x, 4); 
		}
		
		\draw[ubred,decorate,decoration={brace,amplitude=6pt}] (3.85, -.25) -- (1.15, -.25) node [ubred, midway, yshift = -12pt, font=\scriptsize] {$g=1$}; 
		\draw[ubred,decorate,decoration={brace,amplitude=6pt}] (9.85, -.25) -- (4.15, -.25) node [ubred, midway, yshift = -12pt, font=\scriptsize] {$g=2$};
		\draw[ubred,decorate,decoration={brace,amplitude=6pt}] (12.85, -.25) -- (10.15, -.25) node [ubred, midway, yshift = -12pt, font=\scriptsize] {$g=3$};
		\draw[ubred,decorate,decoration={brace,amplitude=6pt}] (14.85, -.25) -- (13.15, -.25) node [ubred, midway, yshift = -12pt, font=\scriptsize] {$g=4$};

	\end{tikzpicture}
	\caption{Item types and resulting knapsack groups}\label{fig:ordinary:groups}
\end{figure}

\begin{enumerate}
	\item[1)] \textbf{Linear grouping of big items:} Guess~$\lmax$, the index of the highest value class that belongs to \opt and use  {oblivious} linear grouping with~~$J'= J$ and $n' = n$ to obtain~$\types$, the set of item types~$t$ with their multiplicities~$n_t$.
	
	\item[2)] \textbf{Knapsack Grouping:} Consider the knapsacks sorted increasingly by their capacity and determine for each item size for which knapsacks a corresponding item would be big or small. This yields a set $\groups$ of $O( \frac{\log^2n}{\eps^4})$ many knapsack groups.
	Denote by $\smalls_g$ the set of all item types that are small with respect to group $g$, and by $\capa_g$ the total capacity of all knapsacks in group $g$. Let~$m_g$ be the number of knapsacks in group~$g$ and let~$\groups^{(1/\eps)}$ be the groups in~$\groups$ with~$m_g \geq \frac1\eps$. For each $g \in \groups^{(1/\eps)}$, define~$\capa_{g,\eps}$ as the total capacity of the smallest~$\eps m_g$ many knapsacks in $g$. Similar to the ILP for identical knapsacks, the ILP reserves some knapsacks to pack small \enquote{cut} items. We distinguish between $\groups^{(1/\eps)}$ and $\groups\setminus\groups^{(1/\eps)}$ to restrict only large enough groups $g$, i.e, $g \in \groups^{(1/\eps)}$, to the $(1-\eps)m_g$ most valuable knapsacks of $g$.
	
	\item[3)] \textbf{Configurations:} For each group $g\in \groups$, create all possible configurations consisting of at most $\frac{1}{\eps}$ items which are big with respect to knapsacks in $g$.
	This amounts to $O((\frac{\log^2n}{\eps^4})^\epsfrac)$ configurations per group.
	Order the configurations decreasingly by size and denote the set of such configurations by $\configs_g = \{ c_{g,1}, c_{g,2} \ldots c_{g,k_g} \}$. 
	Let~$m_{g,\ell}$ be the total number of knapsacks in group~$g$ in which we could possibly place configuration~$c_{g,\ell}$. 
	Further, denote by~$n_{c,t}$ the number of items of type~$t$ in configuration~$c$, and by~$s_c$ and~$v_c$ the size and value of $c$ respectively. 
	
	\item[4)] \textbf{Configuration ILP:} Solve the following configuration ILP with variables~$y_c$ and~$z_{g,t}$. Here, $y_c$ counts how often a certain configuration~$c$ is used, and~$z_{g,t}$ counts how many items of type~$t$ are packed in knapsacks of group~$g$ if type~$t$ is small with respect to~$g$.
	Note that by the above definition of $\configs_g$, we may have duplicates of the same configuration for several groups. 
	\begin{equation}\tag{P}\label{eq:mmdk:ilp}
		\begin{array}{llcll}
			\max & \displaystyle{\sum_{g \in \groups} \sum_{c \in \configs_g} y_{c} v_{c} + \sum_{g \in \groups}\sum_{t \in \smalls_g} z_{g,t} v_t } \\
			\text{s.t. } & \displaystyle{\sum_{h = 1}^{\ell} y_{c_{g,h}} }& \leq & m_{g,\ell}& \text{for all } g \in \groups, \ell \in [k_g] \\
			& \displaystyle{\sum_{c \in \configs_g} y_c  }& \leq & (1-\eps) m_{g} & \text{for all } g \in \groups^{(1/\eps)} \\
			& \displaystyle{\sum_{c \in \configs_g} y_{c} s_{c_{g,h}} + \sum_{t \in \smalls_g} z_{g,t} s_t  }  & \leq & \capa_{g} & \text{for all } g \in \groups\setminus \groups^{(1/\eps)} \\
			& \displaystyle{\sum_{c \in \configs_g} y_{c} s_{c_{g,h}} + \sum_{t \in \smalls_g} z_{g,t} s_t  }  & \leq & \capa_{g} - \capa_{g,\eps}& \text{for all } g \in \groups^{(1/\eps)} \\	 
			& \displaystyle{\sum_{g \in \groups} \sum_{c \in \configs_g} y_{c} n_{c,t} +\sum_{g \in \groups: t \in \smalls_g} z_{g,t}  } & \leq & n_t & \text{for all } t \in \types\\	
			& y_{c} & \in & \mathbb{Z}_{\geq 0} & \text{for all } g\in\groups, c \in \configs_g \\
			& z_{g,t} & \in & \mathbb{Z}_{\geq 0} & \text{for all } t \in \types, g \in \groups \\
			& z_{g,t} & = & 0 & \text{for all} t \in \types, g \in \groups \,:\, t \notin \smalls_g
		\end{array}
	\end{equation}
	The first inequality ensures that the configurations chosen by the ILP actually fit into the knapsacks of the respective group while the second inequality ensures that an~$\eps$-fraction of knapsacks in~$\groups_{1/\eps}$ remains empty for packing small \enquote{cut} items. The third and fourth inequality guarantee that the total volume of large and small items together fits within the designated total capacity of each group.
	Finally, the fifth inequality makes sure that only available items are used by the ILP.
	
	\item \textbf{Obtaining an integral solution:} After relaxing the above ILP and allowing fractional solutions, we are able to solve it efficiently. Let~$\opt_{\text{LP}}$ be an optimal (fractional) solution to~\eqref{eq:mmdk:ilp} with objective function value~$v_{\text{LP}}$. With \cref{lem:hg:roundingConfigLP} we obtain an integral solution that uses the additional knapsacks given by the resource augmentation with value at least~$v_{\lp}$. Let~$P_F$ denote this final solution.

	\item \textbf{Packing small items:} Observe that small item types~$t \in \smalls_g$ are only packed fractionally by~$P_F$. \cref{lem:hg:packP} provides us with a way to pack the small items integrally. 
\end{enumerate}

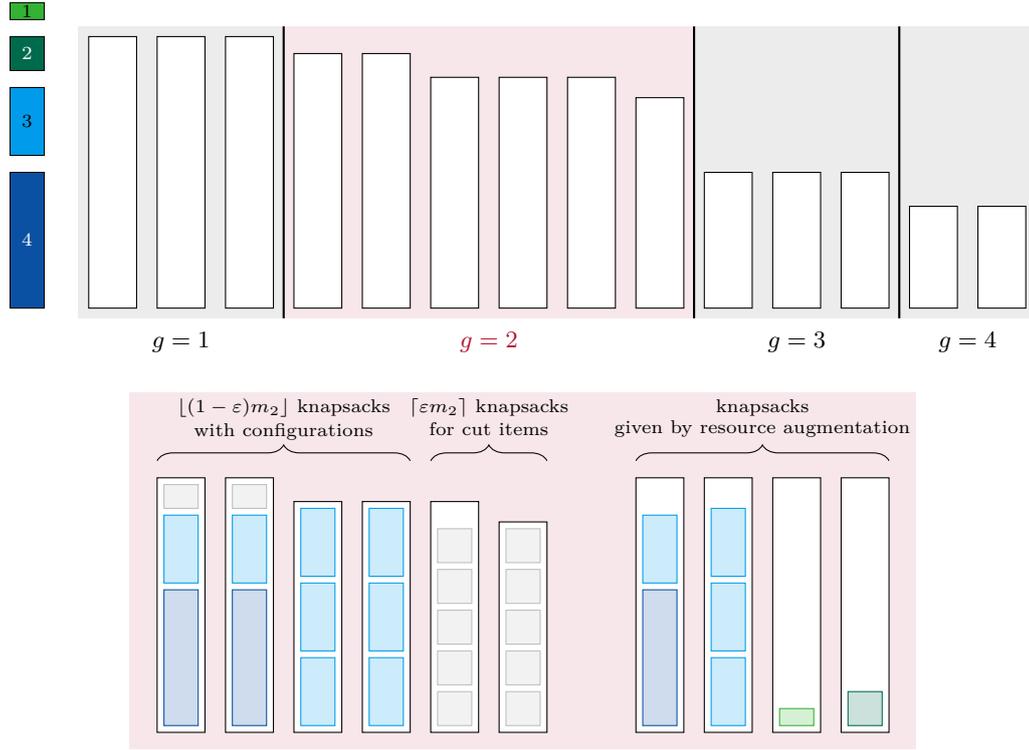
\begin{figure}[tbh]
	\centering 
	\begin{tikzpicture}[scale = .9]
		\fill[draw=none, fill = ubred!10] (4,-.15) rectangle (10,4.15); 	
		\fill[draw=none, fill = lightgray!30] (13,-.15) rectangle (15,4.15); 
		
		\foreach \x/\w in {1/3, 10/3}{
			\fill[draw=none, fill = lightgray!30] (\x,-.15) rectangle (\x+\w,4.15); 	
		}
		
		\fill[draw = none, fill = ubred!10] (1.75, -1.25) rectangle (13.25, -6.5); 
		\foreach \x/\h/\n in {12/1.5/1, 9/2/2, 8/3.1/0, 5/3.4/2, 3/3.75/1, 0/4/2}{ %
			\foreach \y in {0,...,\n}{
				\draw[draw=black, fill=white] (1.25+\x+\y-.1,0) rectangle (1.25+\x+\y+.6, \h); 
			}
		}
		
		\foreach \y/\h/\col in {4.25/.25/utgreen, 3.5/.5/utforest, 2.25/1/utblue, 0/2/utnavy}{
			\draw[draw = black, fill = \col] (0, \y) rectangle (.5, \y + \h); 
		}
		
		\foreach \y/\h/\col/\l in {4.25/.25/black/1, 3.5/.5/white/2, 2.25/1/black/3, 0/2/white/4}{
			\node[\col, font = \scriptsize] at (0.25, \y + \h/2) {\l};
		}
		
		\foreach \x in {4,10,13} {
			\draw[thick, black] (\x, -.15) -- (\x, 4.15); 
		}
		
		\node[font = \small] at (11.5,-.5) {$g=3$};	
		\node[font = \small] at (2.5,-.5) {$g = 1$};	
		\node[font = \small, ubred] at (7,-.5) {$g=2$};	
		\node[font = \small] at (14,-.5) {$g=4$};

		\foreach \x/\h/\n in {8/3.1/0, 5/3.4/2, 3/3.75/1, 10/3.75/3}{ %
			\foreach \y in {0,...,\n}{
				\draw[draw=black, fill=white, xshift = -2cm, yshift=-6.25cm] (1.25+\x+\y-.1,0) rectangle (1.25+\x+\y+.6, \h); 
			}
		}
		
		\begin{scope}[xshift = 2cm]
			\foreach \x/\y/\h/\col/\n in {0/-6.25/2/utnavy/1, 0/-4.15/1/utblue/1, 0/-3.05/.35/lightgray/1, 2/-6.25/1/utblue/1, 2/-5.15/1/utblue/1, 2/-4.05/1/utblue/1, 7/-6.25/2/utnavy/0, 7/-4.15/1/utblue/0, 8/-6.25/1/utblue/0, 8/-5.15/1/utblue/0, 8/-4.05/1/utblue/0, 9/-6.25/.25/utgreen/0, 10/-6.25/.5/utforest/0, 
				4/-6.25/.5/lightgray/1, 4/-5.65/.5/lightgray/1, 4/-5.05/.5/lightgray/1, 4/-4.45/.5/lightgray/1, 4/-3.85/.5/lightgray/1 }{
				\foreach \i in {0,...,\n}{
					\fill[fill =\col!20, draw = \col] (\x+.25+\i, \y+.1) rectangle (\x+.75+\i, \y+\h+.1);
				}
			}
			
			\draw [decorate,decoration={brace,amplitude=6pt},black]  (4.15, -2.25) -- (5.85,-2.25) node [black,midway,yshift =  20pt,font=\scriptsize, black] {$\lceil \eps m_2 \rceil$  knapsacks} ;
			\draw [decorate,decoration={brace,amplitude=6pt},black]  (.15, -2.25) -- (3.85, -2.25) node [black,midway,yshift =  20pt,font=\scriptsize, black] {$\lfloor (1-\eps) m_2 \rfloor$ knapsacks};
			\node[yshift = 11pt, font = \scriptsize] at (2,-2.25) {with configurations}; 
			\node[yshift = 12pt, font = \scriptsize] at (5,-2.25) {for cut items};
			
			\draw [decorate,decoration={brace,amplitude=6pt},black]  (7.15, -2.25) -- (10.85, -2.25) node [black,midway,yshift =  20pt,font=\scriptsize, black] {knapsacks};
			\node[yshift = 12pt, font = \scriptsize] at (9,-2.25) {given by resource augmentation};	
			
		\end{scope}

	\end{tikzpicture}
	\caption{Possible solution of the algorithm: Group~$2$ accommodates the knapsacks for cut small items within the original knapsacks. Group~1, 3, and 4 use resource augmentation.}\label{fig:ordinary:solution}
\end{figure}

\paragraph*{Queries} Since we do not maintain an explicit packing of any item, we define and update pointers for each item type that dictate the knapsacks where the corresponding items are packed. We note that special pointers are also used for packing items into the additional knapsacks given by resource augmentation. To stay consistent between two update operations, we cache query answers for the current round in the data structure that store items. We give the details in the next section.

\begin{itemize}
	\item \textbf{Single Item Query:} For a queried item, we retrieve its item type and check if it belongs to the smallest items of this type that our implicit solution packs. In this case, we use the pointer for this item type to determine its knapsack. 
	\item \textbf{Solution Value Query:} After having found the current solution, we use prefix computation for every value class for the corresponding item types to calculate and store the actual solution value. Then, we return this value on query. 
	\item \textbf{Entire Solution Query:} With prefix computation on each value class, we determine the packed items. Then, the single item query is used to determine their knapsack. 
\end{itemize}

\subsection{Analysis}
We start again by showing that the loss in the objective function value due to the linear grouping of items is bounded by a factor of at most~$\frac{(1-\eps)(1-2\eps)}{(1+\eps)^2}$ with respect to~$v(\opt)$.  To this end, let~\opt be an optimal solution to the current, non-modified instance and let $J$ be the set of items with values already rounded to powers of $(1+\eps)$.
Setting~$J' = J$, we apply \cref{theo:harmonic} to obtain the following corollary. Here,~$\opt_{\types}$ is a optimal solution for the instance induced by the item types~$\types$ with multiplicities~$n_t$.

\begin{corollary}\label{cor:mmdk:HarmonicGrouping}
	There exists an index~${\lmax}$ such that~$v(\opt_{\types}) \geq \frac{(1-\eps)(1-2\eps)}{(1+\eps)^2} v(\opt)$. 
\end{corollary}

We have thus justified the restriction to item types~$\types$ instead of packing the actual items.
In the next two lemmas, we show that \eqref{eq:mmdk:ilp} is a linear programming formulation of the \multiknapsack problem described by the set~$\types$ of item types and their multiplicities and that we can obtain a feasible integral packing (using resource augmentation) if we have a fractional solution (without resource augmentation) to \eqref{eq:mmdk:ilp}. Let~$v_{\text{LP}}$ be the optimal objective function value of the LP relaxation of~\eqref{eq:mmdk:ilp}. 

Similar to the proof of \cref{lem:mik:OPTofConfigILP}, we restrict an optimal solution $\opt_{\types}$ to the~$\lfloor (1-\eps) m_g\rfloor $ most valuable knapsacks of a group~$g$ if~$m_g \geq \frac 1\eps$ and otherwise we do not restrict the part of the solution corresponding to a group~$g$ with~$m_g < \frac 1\eps$. %

\begin{restatable}{lemma}{mmdkOptOfConfigLP}\label{lem:mmdk:OptOfConfigLP}
	It holds that~$v_{\text{LP}} \geq (1-2\eps) v(\opt_{\types})$. 
\end{restatable}

\begin{proof}
	We show the statement by explicitly stating a solution~$(y,z)$ that is feasible for~\eqref{eq:mmdk:ilp} and achieves an objective function value of at least~$(1-2\eps) v(\opt_{\types})$. 
	
	Consider a feasible optimal packing $\opt_{\types}$ for item types. The construction of~$(y,z)$ considers each group~$g \in \groups$ separately. We fix a group~$g \notin \groups^{(1/\eps)}$. Let~$y_c$ count how often a configuration~$c \in \configs_g$ is used in $\opt_{\types}$ and let~$z_{g,t}$ denote how often an item that is small with respect to~$g$ is packed by~$\opt_{\types}$ in group~$g$. By construction, the first and the third constraint of~\eqref{eq:mmdk:ilp} are satisfied. The part of the solution~$(y,z)$ corresponding to group~$g$ achieves the same value as~$\opt_{\types}$ restricted to this group. 	
	
	If~$g \in \groups^{(1/\eps)}$, i.e., if there are at least~$\frac{1}{\eps}$ knapsacks in group~$g$, consider the~$\lfloor (1-\eps)m_g \rfloor$ most valuable knapsacks in group~$g$ with respect to $\opt_{\types}$. Define~$y_c$ to count how often $\opt_{\types}$ uses configuration~$c \in \configs_c$ in this reduced knapsack set and let~$z_{g,t}$ denote how often $\opt_{\types}$ uses item type~$t \in \smalls_g$ in these knapsacks. Clearly, this solution satisfies the first constraint of~\eqref{eq:mmdk:ilp}. By construction,~$\sum_{c \in \configs_g} y_c \leq \lfloor (1-\eps)m_g \rfloor$ and, hence, the second constraint of the ILP is also satisfied. Clearly, the $\lfloor (1-\eps)m_g \rfloor$ most valuable knapsacks can be packed into the $\lfloor (1-\eps)m_g \rfloor$ largest knapsacks in~$g$, which implies the feasibility for the fourth constraint of the ILP. Observe that $\lfloor (1-\eps) m_g \rfloor \geq (1-\eps) m_g - 1 \geq (1-2\eps) m_g$. Thus, the value of the corresponding packing is at least a~$(1-2\eps)$-fraction of the value that~$\opt_{\types}$ obtains with group~$g$. 
	
	As~$(y,z)$ uses no more items of a certain item type than~$\opt_{\types}$ does, the last constraint of the ILP is also satisfied. Hence, $(y,z)$ is feasible and \[
	v_{\text{LP}} \geq \sum_{g \in \groups} \Big(\sum_{c \in \configs_g} y_c v_c + \sum_{t \in \smalls_g} z_{g,t} v_t\Big) \geq (1-2\eps) v(\opt_{\types}),
	\]
	with which we conclude the proof. 
\end{proof}

The next corollary shows how to round any fractional solution of~\eqref{eq:mmdk:ilp} to an integral solution (possibly) using additional knapsacks given by resource augmentation. It follows immediately from \cref{lem:hg:roundingConfigLP} if we bound the number of variables in~\eqref{eq:mmdk:ilp}. To this end, we observe that~$|\groups|$ and $|\types|$ are in $\OO \big( \frac{\log^2 n}{\eps^4} \big)$, and~$|\configs_g\ | \in \big(\frac{\log n}{\eps}\big)^{\OO(1/\eps)}$ for every group~$g \in \groups$. Let~$L'$ denote the exact number of variables and let~$L = L' + |\groups|$. Thus,~$ L \in \big(\frac{\log n}{\eps}\big)^{\OO(1/\eps)}$.

\begin{corollary}\label{cor:mmdk:RoundSolOFConfigLP}
	Any feasible solution~$(y,z)$ of the LP relaxation of~\eqref{eq:mmdk:ilp} with objective value~$v$ can be rounded to an integral solution with value at least~$v$ using {at most $L$ extra knapsacks.}  
\end{corollary}

In the next lemma, we bound the value obtained by our algorithm in terms of~$v(\opt)$, for an optimal solution \opt. Let~$P_F$ be the solution returned by our algorithm. 

\begin{lemma}\label{lem:mmdk:Combine}
	$v(P_F) \geq \frac{(1-2\eps)^2(1-\eps)}{(1+\eps)^2} v(\opt)$. 
\end{lemma}

\begin{proof}
	Fix an optimal solution~$\opt$. Observe that our algorithm outputs the solution~$P_F$ with the maximum value over all guesses of~$\lmax$, the index of the highest value class in \opt. Hence, we find a guess~$\lmax$ and a corresponding solution~$P$ that satisfies~$v(P) \geq \frac{(1-2\eps)^2(1-\eps)}{(1+\eps)^2} v(\opt)$. 
	
	Let~$\lmax = \max \{\ell: V_\ell \cap \opt \neq \emptyset \}$. Then,~$\lmax$ is considered in some round of the algorithm. Let~$v_{\text{ILP}}$ be the optimal solution value of the configuration ILP~\eqref{eq:mmdk:ilp} and let~$v_{\text{LP}}$ be the solution value of its LP relaxation. \cref{cor:mmdk:RoundSolOFConfigLP} provides a way to round the corresponding LP solution~$(y,z)$ to an integral solution~$(\bar{y},\bar{z})$ using at most~$L$ extra knapsacks with objective function value at least~$v_{\text{LP}} \geq v_{\text{ILP}}$. The construction of~$(\bar y, \bar z)$ guarantees that only small items in the original knapsacks might be packed fractionally. 
	
	Consider one particular group~$g$. \cref{lem:hg:packP} shows how to pack the small items assigned by~$(\bar z_g)$ to group~$g$ into~$\lceil (1+ \eps) m_g \rceil$ knapsacks. If~$m_g < \frac1\eps$, we use one extra knapsack per group to pack the cut items. If~$m_g \geq \frac1\eps$, then~$g \in \groups^{(1/\eps)}$ which implies that the configuration ILP (and its relaxation) already reserved~$\lceil \eps m_g \rceil$ knapsacks of this group for packing small items. Hence, the just obtained packing~$P$ is feasible. 	
	By \cref{cor:mmdk:HarmonicGrouping,lem:mmdk:OptOfConfigLP},
	\begin{align*}
		v(P_F) & \geq v(P) \geq \frac{(1-2\eps)^2(1-\eps)}{(1+\eps)^2} v(\opt), 
	\end{align*}	
	which gives the desired bound on the approximation ratio. 
\end{proof}

Now, we bound the running time of our algorithm. %
\begin{lemma}\label{lem:mmdk:RunningTime}
	In time $\big(\frac1\eps \log n \big)^{\OO(\epsfrac)} (\log m \log S_{\max} \log \vvmax)^{\OO(1)}$, the dynamic algorithm executes one update operation. 
\end{lemma}

\begin{proof} 
	By assumption, upon arrival, the value of each item is rounded to natural powers of~$(1+\eps)$. The algorithm starts with guessing~$\lmax$, the largest index of a value class to be considered in the current iteration. There are~$\log \vvmax$ many guesses possible, where~$\vvmax$ is the highest value appearing in the current instance. 
	
	By \cref{lem:hg:RunningTime}, the  {oblivious} linear grouping of all items has at most~$\OO\big(\frac{\log^4 n}{\eps^4}\big)$ iterations. 
	
	Let the knapsacks be sorted by increasing capacity and stored in a binary balanced search tree as defined in \cref{lem:data-structure}. Then, the index of the smallest knapsack~$i$ with~$\capa_i \geq \capa$ or the largest knapsack with~$\capa_i \leq \capa$ can be determined in time~$\OO(\log m)$, where~$\capa$ is a given number. Thus, the knapsack groups depending on the item types can be determined in time~$\OO\big(\log m \frac{\log^2 n}{\eps^4}\big)$ as the number of item types is bounded by~$\OO\big(\frac{\log^2 n}{\eps^4}\big)$. 	
	The number of big items per knapsack is bounded by~$\frac1\eps$ and, hence, the number of configurations is bounded by 	$\OO\Big(\frac{\log^2 n}{\eps^4}\big(\frac{\log^2n}{\eps^4}\big)^{\epsfrac}\Big)$. 
	
	Let~$N$ be the number of variables in the configuration ILP. We have~$N \in \big(\frac{\log n}{\eps}\big)^{\OO(1/\eps)}$. Hence, there is a polynomial function~$g(N, \log \capa_{\max}, \log \vvmax)$ that bounds the running time of finding an optimal solution to the LP relaxation of the configuration ILP~\cite{BertsimasT1997,PapadimitriouS82}. Clearly, the computational complexity of setting up and rounding the fractional solution is dominated by solving the LP. 	
	Thus,~$\big(\frac1\eps \log n \big)^{\OO(\epsfrac)} (\log m \log S_{\max} \log \vvmax)^{\OO(1)}$ bounds the running time.

	In similar time, we can store~$y$ and~$z$, the obtained solutions to the configuration LP. Let~$\bar y$ and~$\bar z$ be the variables obtained by (possibly) rounding down~$y$ and~$z$ and let~$\tilde y$ and~$\tilde z$ be the variables corresponding to the resource augmentation as in \cref{lem:hg:roundingConfigLP}. The time needed to obtain these variables is dominated by solving the LP relaxation~of~the~configuration~ILP.
\end{proof}

\paragraph*{Answering Queries}

Since we only store implicit solutions, it remains to show how to answer the corresponding queries. In order to determine the relevant parameters of a particular item, we assume that all items are stored in one balanced binary search tree that allows us to access one item in time~$\OO(\log n)$ by \cref{lem:data-structure}. We additionally assume that this balanced binary search tree also stores the value class of an item. We use again the round parameter~$t(j)$ and the corresponding knapsack~$k(j)$ to cache given answers in order to stay consistent between two updates. If~$j$ was \textsc{not selected} in round~$t(j)$, we represent this by~$k(j) = 0$. 
We assume that these two parameters are stored in the same binary search tree that also stores the items and, thus, can be accessed in time~$\OO(\log n)$.

We now design an algorithm for non-cached items. The high-level idea is similar to the algorithm developed in \cref{sec:mik-and-mmdk} for identical knapsacks. As the knapsacks have different capacities in this section, the relative size of an item depends on the particular knapsack group: An item can be big with respect to one knapsack and small with respect to another. Thus, the distinction between small and big items does not hold for all knapsacks simultaneously anymore and needs to be handled carefully. More precisely, upon query of an item~$j$ of type~$t$, we start by determining the group~$\gamma_t$ in which the next item of type~$t$ is packed. The pointers and counters we use correspond mostly to the ones in \cref{sec:mik-and-mmdk} except that we additionally have a dependency on the particular group~$g$ for each parameter. Additionally, we use~$R^{(\eps)}_g$,~$R^{(y)}_g$ and~$R^{(z)}_g$ to refer to knapsacks given by resource augmentation for group~$g$.

If~$t$ is small with respect to~$\gamma_t$, then~$j$ is packed by \textsc{Next Fit} either as regular or as cut item. We use the two pointers~$\kappa_g^r$ for packing small items regularly in group~$g$ and~$\kappa_g^c$ for packing cut items. If there are at most~$\frac1\eps -1$ knapsacks in group~$g$, then~$\kappa_g^c$ points to the knapsack~$R^{(\eps)}_g$ given by resource augmentation. Otherwise, the configuration ILP left the smallest~$\lceil \eps m_g \rceil$ knapsacks in group~$g$ empty for packing cut small items. Further, we use~$R^{(z)}_{g,t}$ to refer to the knapsack given by resource augmentation that is used for packing one item of type~$t$ if the variable~$z_{g,t}$ was subjected to rounding. Since we may only pack as many items of type~$t$ in group~$g$ as indicated by the implicit solution, the counter~$\eta_t^S$ determines how many items of type~$t$ can still be packed in group~$\gamma_t$ if~$t$ is small with respect to~$\gamma_t$. 

If~$t$ is big with respect to~$\gamma_t$, then~$j$ is packed in the next slot for items of type~$t$ determined by the configuration ILP. To this end, we use again the counter~$\kappa_t$ to determine the knapsack where the next item of type~$t$ is packed and the counter~$\eta_t^B$ to determine how many items of type~$t$ can still be packed in knapsack~$\kappa_t$ if~$t$ is big with respect to~$\gamma_t$. The knapsack~$R^{(y)}_{c,g}$, for~$c \in \configs_g$, refers to the knapsack given by resource augmentation used when the variable~$y_{c,g}$ was subjected to rounding.

\cref{tab:mmdk:query} summarizes the  parameters and counters used to answer queries, and in Figure~\ref{fig:ordinary:queries}, we give an example of the current packing after some items have been queried. Next, we  define the data structures for answering queries before we formally explain how to answer queries.

\begin{figure}[tbh]
	\centering 
	\begin{tikzpicture}[scale = .9]
		\fill[draw=none, fill = ubred!10] (4,-.15) rectangle (10,4.15); 
		
		\foreach \x/\w in {1/3, 10/3/, 13/2}{
			\fill[draw=none, fill = lightgray!30] (\x,-.15) rectangle (\x+\w,4.15); 	
		}
		
		\fill[draw = none, fill = ubred!10] (1.25, -1.25) rectangle (13.75, -7.5); 

		\foreach \x/\h/\n in {12/1.5/1, 9/2/2, 8/3.1/0, 5/3.4/2, 3/3.75/1, 0/4/2}{ %
			\foreach \y in {0,...,\n}{
				\draw[draw=black, fill=white] (1.25+\x+\y-.1,0) rectangle (1.25+\x+\y+.6, \h); 
			}
		}
		
		\foreach \y/\h/\col in {4.25/.25/utgreen, 3.5/.5/utforest, 2.25/1/utblue, 0/2/utnavy}{
			\draw[draw = black, fill = \col] (0, \y) rectangle (.5, \y + \h); 
		}
		
		\foreach \y/\h/\col/\l in {4.25/.25/black/1, 3.5/.5/white/2, 2.25/1/black/3, 0/2/white/4}{
			\node[\col, font = \scriptsize] at (0.25, \y + \h/2) {\l};
		}
		
		\foreach \x in {4,10,13} {
			\draw[thick, black] (\x, -.15) -- (\x, 4.15); 
		}
		
		\node[font = \small] at (11.5,-.5) {$\gamma_1=3$};	
		\node[font = \small] at (11.5,-.9) {$\eta_1^S = 3$};
		
		\node[font = \small] at (2.5,-.5) {$\gamma_4=1$};	

		\node[font = \small, ubred] at (7,-.5) {$\gamma_2=\gamma_3=2$};	
		\node[font = \small, ubred] at (7,-.9) {$\eta_2^S = 1$};

		\foreach \x/\h/\n in {8/3.1/0, 5/3.4/2, 3/3.75/1, 10/3.75/1, 13/3.75/1}{ %
			\foreach \y in {0,...,\n}{
				\draw[draw=black, fill=white, xshift = -2.5cm, yshift=-6.25cm] (1.25+\x+\y-.1,0) rectangle (1.25+\x+\y+.6, \h); 
			}
		}
		
		\begin{scope}[xshift = 1.5cm]
			\foreach \x/\y/\h/\col/\n in {0/-6.25/2/utnavy/1, 0/-4.15/1/utblue/1, 0/-3.05/.35/lightgray/1, 2/-6.25/1/utblue/1, 2/-5.15/1/utblue/1, 2/-4.05/1/utblue/1, 7/-6.25/2/utnavy/0, 7/-4.15/1/utblue/0, 8/-6.25/1/utblue/0, 8/-5.15/1/utblue/0, 8/-4.05/1/utblue/0, 10/-6.25/.25/utgreen/0, 11/-6.25/.5/utforest/0, 4/-6.25/.5/lightgray/1, 4/-5.65/.5/lightgray/1, 4/-5.05/.5/lightgray/1, 4/-4.45/.5/lightgray/1, 4/-3.85/.5/lightgray/1 }{
				\foreach \i in {0,...,\n}{
					\fill[fill =\col!20, draw = \col] (\x+.25+\i, \y+.1) rectangle (\x+.75+\i, \y+\h+.1);
				}
			}
			
			\foreach \x/\y in {0/-4.15}{
				\fill[fill = utblue, draw=none] (\x+.25, \y+.1) rectangle (\x+.75, \y+1.1);	
				\node [font = \scriptsize, black] at (\x + .5, \y + .6) {3};	
			}
			
			\foreach \x/\y in {0/-3.05,10/-6.25}{
				\fill[fill = utgreen, draw = none] (\x+.25, \y+.1) rectangle (\x + .75, \y + .35);
				\node[font = \scriptsize, black] at (\x + .5, \y + .225) {1};
			}
			
			\fill[fill = utforest, draw=none] (4.25,-6.15) rectangle (4.75, -5.65);
			\node[font = \scriptsize, white] at (4.5,-5.9) {2};

			\draw [decorate,decoration={brace,amplitude=6pt},black]  (4.15, -2.25) -- (5.85,-2.25) node [black,midway,yshift =  20pt,font=\scriptsize, black] {$\lceil \eps m_2 \rceil$  knapsacks} ;
			\draw [decorate,decoration={brace,amplitude=6pt},black]  (.15, -2.25) -- (3.85, -2.25) node [black,midway,yshift =  20pt,font=\scriptsize, black] {$\lfloor (1-\eps) m_2 \rfloor$ knapsacks};
			\node[yshift = 11pt, font = \scriptsize] at (2,-2.25) {with configurations}; 
			\node[yshift = 12pt, font = \scriptsize] at (5,-2.25) {for cut items};
			
			\draw [decorate,decoration={brace,amplitude=6pt},black]  (7.15, -2.25) -- (8.85, -2.25) node [black,midway,yshift =  14pt,font=\scriptsize, black] {knapsacks in $R^{(y)}_2$};	
			
			\draw [decorate,decoration={brace,amplitude=6pt},black]  (10.15, -2.25) -- (11.85, -2.25) node [black,midway,yshift =  14pt,font=\scriptsize, black] {knapsacks in $R^{(z)}_2$};

			\node [font=\small,black] at (1.5, -6.74) {$\kappa_3$};
			\node [font=\small,black] at (1.5, -7.25) {$\eta_3^B = 1$};
			
			\node [font=\small,black] at (2.5, -6.65) {$\kappa^{(r)}_2$};
			\node [font=\small,black] at (4.5, -6.65) {$\kappa^{(c)}_2$};
			\node [font=\small,black] at (4.5, -7.23) {$\rho^{(c)}_2 = 4$};

		\end{scope}
		
	\end{tikzpicture}
	\caption{Counters and pointers for answering queries: Gray rectangles inside knapsacks represent small items. The next item of type~2 (dark green) is placed in the knapsack given by resource augmentation~$R_2^{(z)}$ since~$\eta_2^S = \tilde z_{2,2}$. Items of type~1 (light green) already filled all their slots in group~$2$ and are now placed in group~$3$.}\label{fig:ordinary:queries}
\end{figure}
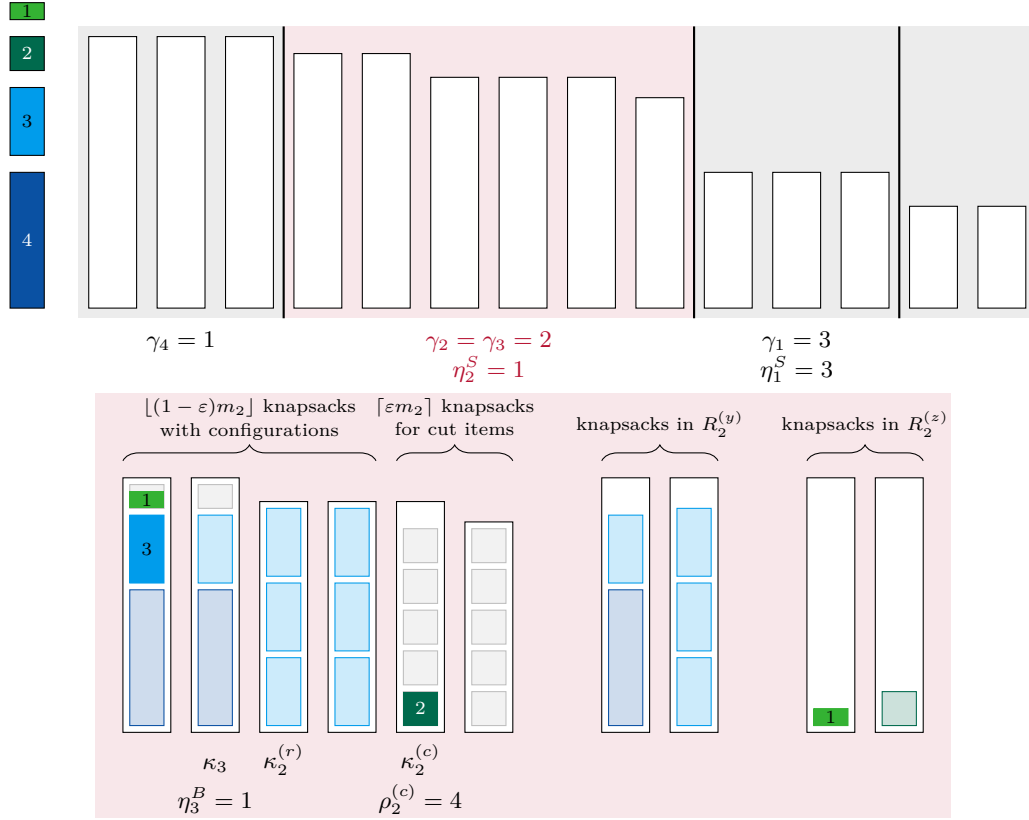

\begin{table}\caption{\label{tab:mmdk:query} Counters and pointers used during querying items}
	\begin{tabular}{c|l}
		Counter/Pointer & Meaning \\\hline 
		$\overline \configs_{g}$ & Configurations that are used by group~$g$ \\ 
		$\alpha_{c,g}$ & First knapsack with configuration~$c$ in group~$g$ \\
		$R_{c,g}^{(y)}$ & Knapsack in~$R^{(y)}$ used for group~$g$ and configuration~$c$ \\
		$R_{g,t}^{(z)}$ & Knapsack in~$R^{(z)}$ used for group~$g$ and type~$t$ \\
		$R_g^{(\eps)}$ & Knapsack in~$R^{(\eps)}$ used for group~$g$ with~$m_g < \frac1\eps$ \\ \hline
		$\groups_t$ & Knapsack groups where items of type~$t$ are packed \\
		$\configs_{g,t}$ & List of configurations~$c \in \overline \configs_g$ with~$n_{c,t} \geq 1$ \\
		$\gamma_t$ & Current knapsack group where items of type~$t$ are packed \\		
		$\kappa_{t}$ & Current knapsack for packing items of a big type~$t$ \\	
		$\eta_{t}^S$ & Remaining number of slots for items of type~$t$ in~$\gamma_t$ \\
		$\eta_{t}^B$ & Remaining number of slots for items of type~$t$ in~$\kappa_t$ \\\hline
		$\kappa_g^r$ & Current knapsack in~$g$ for packing small items regularly \\
		$\kappa_g^c$ & Current knapsack in~$g$ (or in~$R^{(\eps)}$) for packing cut small items \\
		$\rho_g^r$ & Remaining capacity in~$\kappa_g^r$ for packing small items \\
		$\rho_g^c$ & Remaining number of slots for small items in~$\kappa_g^c$  \\
	\end{tabular}
\end{table}

\subparagraph{Data structures} We assume that the knapsacks are sorted by non-increasing capacity and stored in one binary search tree together with~$S_i$, the capacity of the knapsacks. The knapsacks given by resource augmentation are stored in three different lists,~$R^{(y)}$,~$R^{(z)}$, and~$R^{(\eps)}$, needed due to rounding~$y$ or~$z$ or because~$m_g < \frac{1}{\eps}$, respectively. 
The knapsack groups are stored in the list~$\groups$ sorted by non-increasing knapsack capacity. For each group~$g$, we additionally store the number~$m_g$ of knapsacks belonging to~$g$. %

Let~$\bar y, \tilde y, \bar z,$ and~$\tilde z$ be the implicit solution of the algorithm. Here~$\bar *$ refers to packing configurations or items into the original knapsacks while~$\tilde *$ refers to the knapsacks given by resource augmentation. Let~$\overline \configs_g$ be the set of configurations~$c$ with~$\bar y_{c,g}+ \tilde y_{c,g} \geq 1$ ordered in non-increasing size~$s_c$ and stored in one list per group. In the following, we use the position of a configuration~$c \in \overline \configs_g$ in that list as the index of~$c$. For mapping the configurations to knapsacks, we assign the knapsacks~$\sum_{g'= 1}^{g-1} m_{g'} + \sum_{c'= 1}^{c-1} \bar y_{c',g} + 1, \ldots, \sum_{g'=1}^{g-1} m_{g'} +  \sum_{c'= 1}^c \bar y_{c',g}$  to configuration~$c$. For the knapsacks in the resource augmentation, we set~$R^{(y)}_{c,g} = \sum_{g'= 1}^{g-1} \sum_{c' \in \overline \configs_{g'}} \tilde y_{c',g'} + \sum_{c' \leq c } \tilde y_{c',g}$ for each group~$g$ and each configuration~$c \in \overline \configs_g$. 

For each item type~$t$, let~$\bar n_t$ denote the number of items of type~$t$ in the solution. We maintain a pointer~$\gamma_t$ to the group where the next queried item of type~$t$ is supposed to go. We initialize~$\gamma_t$ with the first group that packs items of type~$t$. Since the number of items of type~$t$ assigned to group~$g$ as small items is determined by~$\bar z_{g,t}+ \tilde z_{g,t}$, we additionally use the counter~$\eta_{t}^S$, initialized with~$\bar z_{\gamma_t,t}+ \tilde z_{\gamma_t,t}$, to reflect how many slots group~$\gamma_t$ still has for items of type~$t$. For accessing the knapsacks~$R^{(z)}$ given by resource augmentation, we set~$R^{(z)}_{g,t} = \sum_{g'= 1}^{g-1} \sum_{t' \in \types} \tilde z_{g',t'} + \sum_{t'=1}^t \tilde z_{g,t'}$ for each group~$g$ and item type~$t$. Note that~$z_{g,t} = 0$ holds if~$t$ is big with respect to~$g$.  

When packing small items in group~$g$, we use group pointers~$\kappa_{g}^r$ and~$\kappa_{g}^c$ to refer to the knapsack for packing items regularly or for packing cut items. The pointer~$\kappa_g^r$ is initialized with~$\kappa^r_g = \sum_{g' = 1}^{g-1} m_{g'} + 1$. Further, we use~$\rho_g^r$ to store the remaining capacity for small items in~$\kappa_g^r$ and initialize it with~$\rho_g^r = S_{\kappa^r_g} - s_1$, where~$s_1$ is the size of the first configuration in group~$g$.  If~$m_g \geq \frac1\eps$, we set~$\kappa^c_g = \sum_{g' = 1}^{g-1} m_{g'} + \lfloor (1-\eps) m_g \rfloor   + 1$, while~$m_g < \frac1\eps$ implies that~$\kappa^c_g$ points to the knapsack~$R^{(\eps)}_g$ given by resource augmentation. The counter~$\rho_g^c$ stores again the remaining slots for cut small items in group~$g$ and is initialized with~$\frac1\eps$. 

If~$t$ is big with respect to~$\gamma_t$, we use the pointer~$\kappa_t$ to direct us to the particular knapsack where the next item of type~$t$ goes, while~$\eta_t^B$ stores how many slots~$\kappa_t$ still has available for items of type~$t$. Initially,~$\kappa_t$ points to the first knapsack with a configuration that contains~$t$ in the first group where~$t$ is packed as big item. 
If~$c$ is the corresponding configuration, we set~$\eta_t^B = n_{c,t}$. Because of resource augmentation,~$\kappa_t$ may point to a knapsack in~$R^{(y)}$, the additional knapsacks for rounding~$y$.

\subparagraph*{Answering Item Queries.}
\begin{enumerate}
	\item[1)] \textbf{Check cache.} Let~$\tau$ be the current round and let~$j$ be the queried item. If~$t(j) = \tau$, return~$k(j)$. 
	\item[2)] \textbf{Answer queries for non-cached items.} Set~$t(j) = \tau$ and determine~$t$, the type of~$j$. Let~$\gamma$ be the group of~$t$. If~$\gamma = 0$, return \textsc{not selected}. Decide if~$j$ is small or big with respect to the group~$\gamma$. 
	
	\textbf{Small items.} If~$\eta_{t}^S  = z''_{\gamma,t}$, determine if~$j$ goes to the resource augmentation~$R^{(z)}$: 
	
	If~$z_{\gamma,t}'' = 1$, set~$k(j)$ to the knapsack in~$R^{(z)}$ reserved for~$z''_{\gamma,t}$ and increase~$\gamma_t$ to the next group for type~$t$. If no such group exists, set~$\gamma_t = 0$. Otherwise, update~$\eta_t$ and possibly~$\kappa_t$ accordingly. 
	
	If~$z_{\gamma,t}'' = 0$, increase~$\gamma_t$ to the next group for type~$t$ and go to Step~2. If no such group exists, set~$\gamma_t=0$,~$k(j) = 0$, and return \textsc{not selected.}
	
	Otherwise, determine if~$j$ is packed regularly or as a cut item. If~$s_t \leq \rho_{\gamma}^r$, return~$k(j) = \kappa_\gamma^r$ and decrease~$\rho_\gamma^r$ accordingly. Otherwise, return~$k(j) = \kappa_\gamma^c$ and decrease~$\rho_\gamma^c$ by one. If now 
	$\rho_\gamma^c = 0$, increase~$\kappa_\gamma^c$ by one and set~$\rho_\gamma^c = \frac 1\eps$. 
	
	\textbf{Big items.} If~$\gamma_t^B = 0$, return \textsc{not selected} and set~$k(j) = 0$. Otherwise, return~$k(j) = \kappa_t$ and decrease~$\eta_t$ by one. If this implies~$\eta_t = 0$, let~$c$ be the configuration of~$\kappa_t$.
	
	If~$\kappa_t \in R^{(y)}$, let~$c'$ be the next configuration for type~$t$ in group~$\gamma$ and update~$\kappa_t$ and~$\eta_t^B$ accordingly. If no such configuration exists, increase~$\gamma_t$ to the next group for type~$t$ and update~$\kappa_t$ and~$\eta_t^B$ accordingly. If no such group exists, set~$\gamma_t = 0$. 	
	
	If~$\kappa_t$ belongs to the original knapsacks and is the last knapsack assigned to configuration~$c$, check if there is resource augmentation for configuration~$c$. In this case, point~$\kappa_t$ to the knapsack reserved for rounding~$y_{c,\gamma}''$. Otherwise, let~$c'$ be the next configuration for type~$t$ in group~$\gamma$ and update~$\kappa_t$ and~$\eta_t^B$ accordingly. If no such configuration exists, increase~$\gamma_t$ to the next group for type~$t$ and update~$\kappa_t$ and~$\eta_t^B$ accordingly. If no such group exists, set~$\gamma_t = 0$. 
	
	Otherwise, increase~$\kappa_t$ by one and update~$\eta_t^B$ accordingly.	
\end{enumerate}

For calculating the value of the current solution, we need to calculate the total value of the first~$\bar n_t$ items. We do this by iterating through the value classes once and per value class, we iterate once through the list~$\types_\ell$ of item types for value class~$V_\ell$ to access the number~$\bar n_t$. Then, we use prefix computation twice in order to access the total value of the first~$\bar n_t$ items of type~$t$. Again, we do this computation once after each update operation. \cref{lem:mmdk:queriesSolVal} bounds the running time of these calculations and shows that incorporating these does not change the order of magnitude of the running time given in \cref{lem:mmdk:RunningTime}.

\subparagraph*{Answering the Solution Value Query.}
\begin{enumerate}
	\item[1)] \textbf{Value per item type.} For each item type~$t$, calculate~$v_{t}$, the total value of the first~$\bar n_t$ items with prefix computation.
	\item[2)] \textbf{Value.} Return~$\sum_{t \in \types} v_t$. 
\end{enumerate}

For returning the complete solution, we iterate once through the value classes and for each value class, we iterate through the list~$\types_\ell$ to access the number~$\bar n_t$. Then, we use prefix computation based on the indices of the items for accessing the first~$\bar n_t$ items of type~$t$. Then, we access and query each item individually. 

\subparagraph*{Answering the Solution Query.}
\begin{enumerate}
	\item[1)] For each item type~$t$, query the first~$\bar n_t$ items and return these items with their knapsacks. 
\end{enumerate}

We prove the parts of the next lemma again separately. 

\begin{lemma}
	The solution determined by the query algorithm is feasible as well as consistent and achieves the claimed total value. The query times of our algorithm are as follows. 
	\begin{enumerate}
		\item[(i)] Single item queries can be answered in time~$\OO\big(\log m + \frac{\log n}{\eps^2}\big)$.
		\item[(ii)] Solution value queries can be answered in time~$\OO(1)$. 
		\item[(iii)] Queries of the entire solution~$P$ are answered in time~$\OO\big(
		|P| \frac{\log^3 n }{\eps^4} \big( \log m + \frac{\log n}{\eps^2}\big) \big)$.  
	\end{enumerate}
\end{lemma}

\begin{lemma}\label{lem:mmdk:queries:correctness}
	The query algorithms return a feasible and consistent solution obtaining the total value given by the implicit solution.
\end{lemma}

\begin{proof}
	By construction of~$k(j)$ and~$t(j)$, the solution returned by the query algorithms is consistent between updates. 
	
	Observe that~$\bar y$ and~$\bar z$ is a feasible solution to the configuration ILP~\eqref{eq:mmdk:ilp}. Hence, showing that the algorithm does not assign more than~$\bar y_{c,g}$ times configuration~$c$ and not more than~$\bar z_{g,t}$ items of type~$t$ to group~$g$ is sufficient for having a feasible packing of the corresponding elements into the~$\lfloor (1-\eps)m_g\rfloor $ largest knapsacks of group~$g$ if~$m_g \geq \frac1\eps$ or into the~$m_g$ knapsacks of group~$g$ if~$m_g < \frac1\eps$. When defining~$L$, we made sure that the items and configurations specified by~$\tilde y$ and~$\tilde z$ fit into the knapsacks given by resource augmentation. 
	
	If the item type~$t$ is small with respect to the group~$g$, then at most~$\bar z_{g,t}$ items of type~$t$ are packed in group~$g$. Thus, \cref{lem:hg:packP} ensures that all small items assigned to group~$g$ fit in the knapsacks for regular and the cut items. Moreover, the treatment of~$\eta_t^S = \tilde z_{g,t}$ guarantees that the value obtained by small items packed in~$g$ and its additional knapsacks is as in the implicit solution. 
	
	If~$t$ is big with respect to group~$g$, then the constructions of~$\kappa_t$ and~$\eta_t^B$ ensure that exactly~$\sum_{c \in \overline \configs_g} (\bar y_c+ \tilde y_c ) n_{c,t}$ items of type~$t$ are packed in group~$g$ and in~$R_g^{(y)}$. 
	Hence, the total value achieved is as given by the implicit solution. 
\end{proof}

\begin{restatable}{lemma}{mmdkQueries}\label{lem:mmdk:queries}
	The data structures can be generated in~$\OO\Big( \frac{\log^4 n}{\eps^8} \big( \log m + \frac{\log^{2/\eps} n}{\eps^{4/\eps}} \big) \Big)$ many iterations. Queries for a particular item can be answered in~$\OO\big(\log m + \frac{\log n}{\eps^2}\big)$ many steps. %
\end{restatable}

\begin{proof}
	We start by retracing the steps of the  {oblivious} linear grouping in order to obtain the set~$\types$ of item types. We store the types~$\types_\ell$ of one value class in one list, sorted by non-decreasing size. By \cref{lem:hg:RunningTime}, the set~$\types$ can be determined in time~$\OO\big(\frac{\log^4 n}{\eps^4}\big)$. 
	
	We first argue about the generation of the data structures and the initialization of the various pointers and counters. We start by generating a list~$(\alpha_{c,g})_{c \in \overline \configs_g}$ for each group~$g$ where~$\alpha_{c,g}$ stores the first (original) knapsack of configuration~$c \in \overline \configs_g$, i.e.,
	\[\alpha_{c,g} = \alpha_{c-1,g} + \bar y_{c-1,g} + 1,\]
	where~$\alpha_{0,g} = \sum_{g'=1}^{g-1} m_{g'}$ and~$y_{0,g} = 0$. 	
	Next, we set 
	\begin{align*} 
		R^{(z)}_{g,t} & = \sum_{g'= 1}^{g-1} \sum_{t' \in \types} \tilde z_{g',t'} +   \sum_{t'=1}^t \tilde z_{g,t'}
		\shortintertext{and}
		R^{(y)}_{c,g} & = \sum_{g'= 1}^{g-1} \sum_{c' \in \overline \configs_{g'}} \tilde y_{c',g'}   +\sum_{c' \in \configs_g, c'\leq c} \tilde y_{c',g} \, ,
	\end{align*} 
	where~$R^{(z)}_{g,t}$ corresponds to the resource augmentation needed because of rounding~$z_{g,t}$ and~$R^{(y)}_{c,g}$ corresponds to the resource augmentation for rounding~$y_{c,g}$. These lists can be generated by iterating through the list~$\overline \configs_g$ for each group~$g$ in time~$\OO(\sum_{g \in \groups} |\configs_g|) = \OO\big( \frac{\log^2 n}{\eps^4} \frac{\log^{2/\eps} n}{\eps^{4/\eps}}\big)$.  
	
	For maintaining and updating the pointer~$\gamma_t$, we generate the list~$\groups_t$ that contains all groups~$g$ where items of type~$t$ are packed in the implicit solution. By iterating through the groups once more and checking~$\sum_{c \in \overline \configs_g}(\bar y_{c,g}+ \tilde y_{c,g})n_{c,t} \geq 1$ or $\bar z_{g,t} + \tilde z_{g,t} \geq 1$, we can add the corresponding groups~$g$ to~$\groups_t$. Then,~$\gamma_t$ points to the head of the list. 
	While iterating through the groups, we also calculate~$\bar n_t = \sum_{g \in \groups} 
	\big(\sum_{c \in \configs_g '} (\bar y_{c,g}+ \tilde y_{c,g}) + \bar z_{g,t}+ \tilde z_{g,t}\big)$ and store the corresponding value together with the item type. The lists~$\groups_t$ can be generated in~$\OO(|\types|\sum_{g\in \groups} |\configs_g|) %
	= \OO\Big( \frac{\log^4 n}{\eps^8} \frac{\log^{2/\eps} n}{\eps^{4/\eps}}   \Big) $ many iterations.
	
	For maintaining and updating the pointer~$\kappa_t$, we create the list~$\configs_{g,t}$ storing all configurations~$c \in \overline \configs_g$ with~$n_{c,t} \geq 1$. While iterating through the groups and creating~$\groups_t$, we also add~$c$ together with~$n_{c,t}$ to the list~$\configs_{g,t}$ if~$n_{c,t} \geq 1$. Initially,~$\kappa_t$ points to the head of~$\configs_{g,t}$, where~$g$ is the first group that packs~$t$ as big item. If~$c$ is the corresponding configuration, we start with~$\eta_t^B = n_ {c,t}$. The time needed for this is bounded by~$\OO(|\types| \sum_{g \in \groups} |\configs_g|) = \OO\Big( \frac{\log^4 n}{\eps^8} \frac{\log^{2/\eps} n}{\eps^{4/\eps}} \Big)$. 
	
	The pointer~$\kappa^r_g$ is initialized with~$\kappa^r_g = \sum_{g' = 1}^{g-1} m_{g'} + 1$. By using binary search on the list~$\overline \configs_g$, we get~$s_1$, the total size of configuration~1 assigned to~$\kappa^r_g$, and binary search over the knapsacks allows us to obtain~$S_{\kappa^r_g}$, the capacity of knapsack~$\kappa^r_g$. Thus,~$\rho_g^r = S_{\kappa^r_g} - s_1$ can be initialized in time~$\OO(\sum_{g \in \groups} (\log(|\overline \configs_g| + \log m)) = \OO\Big( \frac{\log^2 n}{\eps^5} \big(\log \frac{\log n}{\eps} + \log m \big)\Big)$. 
	
	If~$m_g \geq \frac1\eps$, we set~$\kappa^c_g = \sum_{g' = 1}^{g-1} m_{g'} + \lfloor (1-\eps) m_g \rfloor   + 1$, while~$m_g < \frac1\eps$ implies that~$\kappa^c_g$ points to the knapsack~$R^{(\eps)} = |\{g': g' \leq g, m_{g'} < \frac1\eps \}|$ given by resource augmentation. The time needed for initializing~$\kappa_g^c$ is~$\OO(|\groups|)$. In order to determine the position of the next cut item, we also maintain~$\rho^c$, initialized with~$\rho^c_g = \frac1\eps$, that counts how many slots are still left in knapsack~$\kappa^c_g$. 
	
	Now consider the query for an item~$j$. We can decide in time~$\OO(\log n)$ if~$j$ has already been queried in the current round. Upon arrival of~$j$, we calculated the index~$\ell$ of its value class. If~$\ell \in \{\llmin, \ldots, \lmax\}$, then the item types~$\types_\ell$ together with their first and last item can be determined in time~$\OO\big(\frac{\log n}{\eps^2}\big)$ by retracing the steps of the linear grouping,. By binary search, the item type of~$j$ can be determined in time~$\OO\big(\log \frac{\log n}{\eps}\big)$. Once the item type is known, we check if~$j$ belongs to the first~$\bar n_t$ items of this type. If not, then \textsc{not selected} is returned. Otherwise, the pointer~$\gamma_t$ answers the question in which group item~$j$ is packed. 
	
	If~$j$ is small and~$\eta_t^S > \tilde z_{t,\gamma_t}$, the knapsack~$k(j)$ can be determined in constant time by nested case distinction and having the correct pointer (either~$\kappa_{\gamma_t}^r$ or~$\kappa_{\gamma_t}^c$) dictate the answer. %
	In order to bound the update time of the data structures, note that packing~$j$ as regular item only implies the updates of~$\rho_{\gamma_t}$ and of~$\eta_t^S$, which take constant time. Hence, it remains to consider the case where~$j$ is packed as a cut item. The capacity of the new knapsack~$\kappa_{\gamma_t}^r$ can be determined in~$\OO(\log m)$ by binary search over the knapsack list while the configuration~$c$ of the new knapsack~$\kappa_{\gamma_t}^r$ and its total size are determined by binary search over the list~$\alpha_{\gamma_t}$ %
	in time~$\OO\left(\log  |\overline \configs_{\gamma_t}| \right) = \OO\left( \frac 1\eps \log \frac{\log n} \eps \right)$. Then,~$\rho_{\gamma_t}^r = S_{\kappa_{\gamma_t}^r} - s_c$ can be computed with constantly many operations. %
	If~$\rho^c_{\gamma_t} = 0$ after packing~$j$ in~$\kappa^c_{\gamma_t}$, we increase the knapsack pointer by one and update~$\rho^c_{\gamma_t} = \frac{1}{\eps}$. In case~$\eta_t^S = \tilde z_{\gamma_t,t} = 1$, item~$j$ is packed in the knapsack~$R^{(z)}_{\gamma_t,t}$ which can be decided in constant time. Otherwise the group pointer~$\gamma_t$ is increased and either~$\eta_{t}^S$ is updated according to the new group or~$\kappa_t$ and~$\eta_t^B$ are used. Updating~$\gamma_t$ can be done by binary search over the list~$\groups_t$ in time~$\OO(\log |\groups|)$. The pointer~$\gamma_t$ is updated at most once \emph{before} determining~$k(j)$. Hence, the case distinction on the relative size of~$t$ is invoked at most twice.
	
	If~$j$ is big, the pointer~$\kappa_{\gamma_t}$ dictates the answer which can be returned in time~$\OO(1)$. For bounding the running time of the possible update operations, observe that~$\eta_{t}$ is updated in constant time with values bounded by~$n$. If~$\eta_{t}^B =0$ after the update, the knapsack pointer~$\kappa_{t}$ needs to be updated as well. The most time consuming update operations are finding a new configuration~$c'$ and possibly even a new group~$g'$. Finding configuration~$c' \in \overline\configs_{\gamma_t,t}$ can be done by binary search in time $\OO\left(\log  |\overline \configs_{\gamma_t}| \right) = \OO\left( \frac 1\eps \log \frac{\log n} \eps \right)$. To update~$\kappa_{t}$ and~$\eta_{t}^B$, we extract~$\kappa_t$ from the list~$\alpha_{\gamma_t}$ and~$n_{c',t}$ from the list~$\overline \configs_{\gamma_t,t}$ in time~$\OO\left(\log  |\overline \configs_{\gamma_t}| \right) = \OO\left( \frac 1\eps \log \frac{\log n} \eps \right)$ by binary search.
	If the algorithm needs to update~$\gamma_t$ as well, this can be done by binary search on the list~$\groups_t$ in time~$\OO(\log|\groups_t|) = \OO\left(\log \frac{\log (n)}{\eps} \right)$.
	
	In both cases, the running time of answering the query and possibly updating data structures is bounded by the running time of the linear grouping step and by the routine to access one particular knapsack, i.e., by~$\OO\big(\log m + \frac{\log n}{\eps^2}\big)$.
\end{proof}

\begin{restatable}{lemma}{mmdkSolValQuery}\label{lem:mmdk:queriesSolVal}
	The solution value can be calculated in time~$\OO\big( \frac{\log^3 n}{\eps^4}  \big)$.  
\end{restatable}

\begin{proof}
	For obtaining the value of the current solution, we calculate the total value of the first~$\bar n_t$ items. We do this by iterating through the value classes once and per value class, we iterate once through the list~$\types_\ell$ to access the number~$\bar n_t$. Then, we use prefix computation twice in order to access the total value of the first~$\bar n_t$ items of type~$t$. \cref{lem:data-structure} bounds this time by~$\OO(\log n)$. 	
	By \cref{lem:hg:NoOfTypes}, the number of item types is bounded by~$\OO\big( \frac{\log^2 n}{\eps^4} \big)$. Combining these two values bounds the total running time by~$\OO\big( \frac{\log^3 n}{\eps^4}  \big)$. As this time is clearly dominated by obtaining the implicit solution in the first place, we calculate and store the solution value when computing the implicit solution value and thus are able to return it in constant time. 
\end{proof}

\begin{restatable}{lemma}{mmdkSolQuery}\label{lem:mmdk:queriesSol}
	In time~$\OO\big(
	|P| \frac{\log^3 n }{\eps^4} \big( \log m + \frac{\log n}{\eps^2}\big) \big)$ a query for the complete solution~$P$ can be answered.
\end{restatable}

\begin{proof}
	For returning the complete solution, we determine the packed items and query each packed item individually. \cref{lem:mmdk:queries} bounds their query times by~$\OO\big(\log m + \frac{\log n}{\eps^2}\big)$ while \cref{lem:data-structure} bounds the running time for accessing item~$j$. \cref{lem:hg:NoOfTypes} bounds the number of item types by~$\OO\big( \frac{\log^2 n }{\eps^4}  \big)$. In total, the running time is bounded by~$\OO\big(
	|P| \frac{\log^3 n }{\eps^4} \big( \log m + \frac{\log n}{\eps^2}\big) \big)$, where~$P$ is the current solution. 
\end{proof}

\subsubsection{Proof of main result} %

\begin{proof}[Proof of \cref{theo:mmdk}]
	\cref{lem:mmdk:Combine} gives the bound on the approximation ratio of our algorithm and \cref{lem:mmdk:RunningTime} bounds the running time of an update operation. Further, \cref{lem:mmdk:queries} gives the running time for query operations.
\end{proof}

\section{\multiknapsack}
\label{app:general}

\paragraph*{Analysis.}
We consider the iteration in which all the guesses, ${\lmax}$, $k$ and $S_{\re}$ are correct.
Let $\mcp_1$ be the set of solutions on the \red knapsacks (without the additional virtual knapsack) and the \green knapsacks such that the total size of \red items placed in \green knapsacks lies in the range $[S_\re, (1+\eps)S_\re]$.
Denote by $\opt_1$ a solution of highest value in $\mcp_1$.
Altering~\opt by deleting the \yellow knapsacks gives a solution in $\mcp_1$ of value at least~$(1-\eps) \cdot v(\opt)$. This holds since for correct guesses the \yellow knapsacks by definition contribute at most an $\eps$-fraction to $\opt$. Further, the correctness of the guessed $S_{\re}$ implies that the altered $\opt$ is indeed a packing in $\mcp_1$.
\begin{observation}\label{lem:MDKG:deleteYellow}
	For $\opt_1$ defined as above, we have $v(\opt_1) \geq (1-\eps) \cdot v(\opt)$.
\end{observation}

\begin{lemma}\label{lem:MDKG:redSubinstanceGood}
	Consider an optimal solution $\opt_\re$ to the \red subproblem, i.e., exclude items in $J_\ye$ but include the virtual knapsack.
	Then~$v(\opt_\re) \geq v(\opt_{1,\re}) - 2\eps \cdot v(\opt)$, where we use the shorthand $\opt_{1,\re} \coloneqq (\opt_1 \cap J_\re) \setminus J_\ye$.
\end{lemma}

\begin{proof}
	Consider the \red items in $\opt_1$ that are not in $J_\ye$. 
	Leave items on \red knapsacks in their current position and place \red items on \green knapsacks into the virtual \red knapsack.
	The latter is possible with the exception of possibly an $\eps$-fraction of the items (with respect to size) due to $S_\re$ being rounded down.
	Deleting the least dense items until the remainder fits into the virtual knapsack causes a loss of at most an $\eps$-fraction of the value of $\opt_1$ plus an additional \red item $j_\re$.
	This item $j_\re$ contributes at most an $\eps$-fraction to \opt as its value is not larger than that of the least valuable element in $J_\ye$ which has a value of less than~$\eps v(\opt)$.
\end{proof}

\begin{lemma}\label{lem:MDKG:finalSolGood}
	Let $P_F$ be the final solution the algorithm computes. Then~$v(P_F)\geq (1-7\eps)v(\opt)$.
\end{lemma}
\begin{proof}
	
	Consider $P_\re$, the solution of the \red subproblem returned by the algorithm of \Cref{sec:MDK:aug} (including virtual knapsack and resource augmentation). We know that~$v(P_\re) \geq (1-\eps) \cdot v(\opt_\re) \geq v(\opt_{1,\re}) - 3\eps \cdot v(\opt)$ by \Cref{theo:mmdk} and \Cref{lem:MDKG:redSubinstanceGood}.
	
	Let~$\opt_\gr \coloneqq \opt_1 \cap J_\gr$, and~$P_2 \coloneqq P_\re \cup \opt_\gr \cup J_\ye$.
	Then, $ \opt_1 = \opt_{1,\re} \cup (\opt_1 \cap J_\ye) \cup (\opt_1 \cap J_\gr)$ implies
	\begin{align*}
		v(\opt_1) &= v(\opt_{1,\re}) + v(\opt_1 \cap J_\ye) ~+~ v(\opt_1 \cap J_\gr)\\
		& \leq v(P_\re) + 3\eps v(\opt) + v(J_\ye) + v(\opt_\gr) \\
		& \leq v(P_2) + 3\eps v(\opt).
	\end{align*}
	With \Cref{lem:MDKG:deleteYellow} we then obtain $v(P_2) \geq v(\opt) - 4 \eps v(\opt)$.

	We now modify~$P_2$ to obtain a solution $P_3$ that lacks the virtual \red knapsack and deals with bundles instead.
	Build~$\frac{L_\gr}\eps$ equal-sized bundles from~$P_\re$ as in Step~5. %
	Place these bundles fractionally on the remaining space of the \green knapsacks that is left after~$\opt_\gr$ is packed. 
	This space is sufficient by definition of~$S_\re$ and $\mcp_1$. 
	Arrange the bundles such that the lowest-value ones are placed fractionally and removing them from the solution incurs a loss of at most~$\eps v(\opt)$.
	Further, remove the items placed fractionally among bundles. 
	Since there are at most $\frac{L_\gr}\eps$ of these with value smaller than the $\frac{L_\gr}{\eps^2}$ items in $J_\ye$, this incurs a
	loss of at most $\eps v(\opt)$.
	
	Therefore, $v(P_3) \geq v(P_2)-2\eps v(\opt)$. Moreover, the portion of~$P_3$ on
	\green knapsacks is a valid solution for the request sent to the \green subproblem.
	Therefore, using \Cref{thm:MDK:few}, the overall solution~$P_F$ satisfies $v(P_F)\geq
	(1-7\eps)v(\opt)$.
\end{proof}

\begin{lemma}
	The algorithm has update time $\big(\frac 1\eps\log(n\vvmax)\big)^{f(1 / \eps)}  + \OO(\frac 1 \eps \log \vmax \log n)$, where $f$ is a quasi-linear function.
\end{lemma}

\begin{proof}
	Guessing~$k$ adds a factor of $\frac 1 \eps$ to the update time. Placing the $\frac{L_\gr}{\eps^2}$ most valuable \red items on \yellow knapsacks and removing them from data structures takes time $\OO(\frac{L_\gr}{\eps^2}\log n)$ which is within the time bound. The same holds for the updates of the \red and \green data structures and for solving the subproblems with the algorithms of \Cref{sec:MDK:aug,sec:MDK:few}.
	
	Cutting the items placed in the virtual \red knapsack info $\frac {L_\gr}{\eps}$ equal-sized bundles can be archived efficiently as follows.
	Compute the total size of these items, using the number of items used for each
	of the $\OO(\frac{\log^2 n}{\eps^4})$ item types and deduce the size of a
	bundle. Sort the item types, e.g., by value then size, and then iteratively pack items of the same type by computing how many items of this type fit in the next non-empty
	bundle. This takes time $\OO(\frac{\log^2 n}{\eps^4} \cdot \frac {L_\gr}{\eps})$ which is sufficient.
	
	Additionally, the maintenance of data structures is dominated in running time by that of the subproblems. These takes time~$\OO(\frac 1 \eps \log \vmax \log n)$ and cause the additive factor.
\end{proof} 

\begin{lemma}\label{lem:MDKG:query}
	The query times of our algorithm are as follows.
	\begin{inparaenum}[(i)]
		\item Single item queries are answered in time~$\OO(\frac {\log  n}{\eps^2} )$.
		\item Solution value queries are answered in time~$\OO(1)$.
		\item Queries of the entire solution~$P$ are answered in time~$\OO(\smash{\frac{\log^4 n}{\eps^6}}\abs{P})$.
	\end{inparaenum}
\end{lemma}

\section{Hardness of Approximation}
\label{sec:hard-ks}

The following theorems provide a justification why our algorithms for \multiknapsack have different running times depending on the number of knapsacks. As Chekuri and Khanna \cite{ChekuriK05} observed, \multiknapsack with~$m=2$ does not admit an FPTAS unless~$\classP=\NP$.

\begin{theorem}[Proposition 2.1 in \cite{ChekuriK05}]
	If \multiknapsack with two identical knapsacks has an FPTAS, then \textsc{Partition} can be solved in polynomial time. Hence there is no FPTAS for \multiknapsack even with~$m=2$, unless $\classP = \NP$. 
\end{theorem}
In the fully dynamic setting, this implies that there is no dynamic algorithm with running time polynomial in~$\log n$ and~$\frac{1}{\eps}$ unless~$\classP = \NP$. We are able to extend this result to the case where~$m \leq \frac{1}{3\eps}$. The statement focuses on dynamic algorithms, our main interest here, but the proof does not use the dynamic nature, only the final complexity.

\begin{restatable}[]{theorem}{HardnessMultiKnapsack}
	\label{theo:HardnessMultiKnapsack}
	Unless $\classP=\NP$, there is no fully dynamic algorithm for \multiknapsack that maintains a $(1-\eps)$-approximate solution in update time polynomial in $\log n$ and $\frac 1 \eps$, for~$m < \frac{1}{3\eps}$. 
\end{restatable}

\begin{proof}
	Consider the strongly $\NP$-hard problem 3-\textsc{Partition}~\cite{garey79} where there are~$3m$ items with sizes~$a_j\in \mathbb{N}$ such that~$\sum_{j = 1}^{3m} a_j = m A$. We use the restricted-input variant where sizes belong to $(A/2,~A/4)$ so that only subsets of size 3 may sum to $A$. The task is to decide whether there exists a partition~$\bigcup_{i=1}^m J_i = [3m]$ such that~$|J_i| = 3$ and~$\sum_{j \in J_i} a_j = A$ for~$1 \leq i \leq m$. 
	
	Consider the following instance for \dmk: There are~$m$ knapsacks with~$\capa = A$ and~$3m $ many items. Each item corresponds to a 3-\textsc{Partition} item with~$s_j = a_j$ and~$v_j = 1$ for~$1 \leq j \leq 3m$. Observe that the 3-\textsc{Partition} instance is a \textsc{Yes}-instance if and only if the optimal solution to the \knapsack problem contains~$3m$ items. Indeed, in such a \knapsack instance, each knapsack must have 3 items.
	
	If \dmk admits a dynamic algorithm with approximation guarantee at least~$(1-\eps)$ and running time polynomial in~$\frac{1}{\eps}$ and~$\log n_0$ where $ m < \frac{1}{3\eps}$, such an algorithm is able to optimally solve the \knapsack instance reduced from 3-\textsc{Partition}. Thus, such an algorithm decides 3-\textsc{Partition} in polynomial time which is not possible, unless $\classP=\NP$. 
\end{proof} 

Note that this result can be extended to a larger number of knapsacks by adding an appropriate number of sufficiently small knapsacks, i.e., polynomially many in $n$.

\end{document}